\documentclass[letterpaper, USenglish, thm-restate]{lipics-v2021}

\newcommand{\prodasm}[1]{\textbf{PROD}(#1)}

\hideLIPIcs
\nolinenumbers

\author{Tim Gomez}{MIT - CSAIL, United States}{tagomez7@mit.edu}{}{}
\author{Elise Grizzell}{UTRGV, United States}{elise.grizzell01@utrgv.edu}{https://orcid.org/0000-0001-8876-2277}{}
\author{Asher Haun}{UTRGV, United States}{asher.haun01@utrgv.edu}{https://orcid.org/0000-0002-2821-6263}{}
\author{Ryan Knobel}{UTRGV, United States}{ryan.knobel01@utrgv.edu}{}{}
\author{Tom Peters}{TU Eindhoven, The Netherlands}{t.peters1@tue.nl}{https://orcid.org/0000-0002-2702-7532}{}
\author{Robert Schweller}{UTRGV, United States}{ robert.schweller@utrgv.edu}{}{}
\author{Tim Wylie}{UTRGV, United States}{timothy.wylie@utrgv.edu}{}{}

\funding{This research was supported in part by National Science Foundation Grant CCF-2329918.}



\title{Intrinsic Universality in\texorpdfstring{\\}{}Seeded Active Tile Self-Assembly}
\authorrunning{T.\ Gomez et al.}
\titlerunning{Intrinsic Universality in Seeded Active Tile Self-Assembly}

\begin{document}

\maketitle
\begin{abstract}

The Tile Automata (TA) model describes self-assembly systems in which monomers can build structures and transition with an adjacent monomer to change their states. This paper shows that seeded TA is a non-committal intrinsically universal model of self-assembly. We present a single universal Tile Automata system containing approximately 4600 states that can simulate (a) the output assemblies created by any other Tile Automata system~$\Gamma$, (b) the dynamics involved in building $\Gamma$'s assemblies, and (c) $\Gamma$'s internal state transitions. It does so in a non-committal way: it preserves the full non-deterministic dynamics of a tile's potential attachment or transition by selecting its state in a single step, considering all possible outcomes until the moment of selection.

The system uses supertiles, each encoding the complete system being simulated. The universal system builds supertiles from its seed, each representing a single tile in~$\Gamma$, transferring the information to simulate $\Gamma$ to each new tile. Supertiles may also asynchronously transition states according to the rules of $\Gamma$. This result directly transfers to a restricted version of asynchronous Cellular Automata: pairwise Cellular Automata.

\end{abstract}
\newpage
\section{Introduction}\label{sec:introduction}
Tile self-assembly is a model that attempts to exploit the computational capabilities of nucleic acids.
DNA molecules can form complex structures, and in controlling the growth of those structures, we can utilize their powers to perform computations.
In recent years, a diverse set of new abstractions and models have been conceived, the most prominent of which has been the (two-dimensional) \emph{abstract Tile Assembly Model} (aTAM)~\cite{winfree1998algorithmic}.
In this model, a \emph{tile} is a non-rotatable unit square with specified \emph{glues} on each side, modeling a single monomer.
Two tiles can attach if their glues match.
A \emph{tile assembly system} is a set of these tile types and a temperature $\tau$.
Research into these models usually revolves around the types of \emph{assemblies} that can be created with specific sets of tile types.

In this paper, we work in a related model, derived by combining elements of tile self-assembly and the local state changes of asynchronous Cellular Automata: \emph{seeded Tile Automata (TA)}~\cite{chalk2018freezing}. A Tile Automata \emph{system} $\Gamma$ has a set of \emph{states} $\Sigma$. These states contain no glues, contrary to the aTAMs tile types.
Instead, tiles with an \emph{initial state} $\sigma\in \Lambda$ ($\Lambda \subseteq \Sigma$) can attach to the \emph{seed}~$s$ if the system contains an \emph{affinity rule} for their respective tile types that has an equal or higher strength than the system \emph{temperature} $\tau$. Should a single pair of tiles lack sufficient strength to bind to the assembly, they may bind \emph{cooperatively} by adding the strengths of affinities of neighboring tiles to reach $\tau$.
Contrary to the \emph{passive} aTAM, tiles in the \emph{active} TA system can change their state.
More restricted than most Cellular Automata systems, only \emph{two} tiles directly adjacent to one another can transition their states if the system contains the corresponding \emph{transition rule}. 

Here, we study the creation of an \emph{intrinsically universal} (IU) Tile Automata system $\Gamma_U$, a system with a finite state set capable of creating not only the final assemblies of any other arbitrary Tile Automata system $\Gamma$ but also replicating the exact assembly process and any additional computations achieved via transitions. Our universal tile assembly system can simulate systems that contain more states than~$\Sigma_U$ does and even simulate itself. To do this, we sacrifice scale.
We use many tiles to create a \emph{supertile}, that simulates a single tile in $\Gamma$.

In this paper, we show that \emph{non-committal intrinsic universality} is impossible in any \emph{passive} system, such as the aTAM.
This means that the dynamics of attachment and transitions of a tile assembly system cannot be faithfully simulated by achieving the final determinations of each in a single step.
Instead, they are \emph{committal} intrinsically universal, meaning that they need multiple attachment and or transition steps to replicate the decision process of a single step in the target system.
On first sight, this appears to contradict previous work showing the aTAM is intrinsically universal~\cite{Doty_Lutz_Patitz_Schweller_Summers_Woods_2012}.
However, that paper contained a subtle error which was later addressed by making the definition of intrinsic universality (IU) slightly weaker~\cite{Meunier_Patitz_Summers_Theyssier_Winslow_Woods_2014}.
We will refer to this weakened version as committal IU.
Besides our negative result, we show that the seeded Tile Automata model with its infinite state changes is in fact non-committal intrinsically universal, using approximately 4600 states.

Intrinsic universality is motivated by creating a universal tile set small enough to be stored in a lab refrigerator for real-world experimentation.
Although 4600 tiles is still a large number of states and is not optimal, 4600 tiles is about ten million tile types less than the previously stated committal intrinsic universality result for two-dimensional aTAM~\cite{Doty_Lutz_Patitz_Schweller_Summers_Woods_2012}.
Importantly, our initial state~set~$\Lambda$ is only a single tile type. While current laboratory capabilities lag the ability to implement this universal tile set as of today, there have been recent advancements in for example the ability to replace tiles experimentally~\cite{petersen2018information,sarraf2023modular} and in the aTAM a tile set capable of universal 6-bit computation was created~\cite{woods2019diverse}.
The aTAM has also been proven to be intrinsically universal in  3D~\cite{Hader_Koch_Patitz_Sharp_2019}, and synchronous Cellular Automata have been shown to be intrinsically universal in 1D, 2D, and 3D~\cite{Banks_1970,durand1997intrinsic,margenstern2006algorithm}.

The question of whether 2D asynchronous Cellular Automata is intrinsically universal is currently open, though work towards a 1D version has been done~\cite{worsch2013towards}.
Tile Automata can be viewed as a restricted version of asynchronous Cellular Automata in which the neighborhood size is 2, the radius is 1, the system is non-deterministic, and the updating is asynchronous.
Therefore, our results directly carry over to this restricted version of Cellular Automata. 

\subsection{Previous Work}\label{subsec:previous_work}

\subparagraph{Cellular Automata.}
The study of self-simulation, and new types of universalities is as old as the field of Cellular Automata itself, with von Neumann introducing the model to build a self-replicating machine \cite{neumann1966theory}. Though it was Banks in 1970 who explicitly coined the term intrinsic universality \cite{Banks_1970}, von Neumann's initial construction was later proven to be intrinsically universal. Conway's famous Game of Life cellular automaton was proven to be intrinsically universal~\cite{durand1999game}. Intrinsic universality in CA has been extensively studied \cite{Briceno_Rapaport_2021,durand1997intrinsic,Goles_Meunier_Rapaport_Theyssier_2011,iirgen1987simple,Ollinger_Richard_2011,Ollinger_2002, Ollinger_2003, Ollinger_2008,worsch2013towards}. Specifically, four different updating schemes for Asynchronous CA were shown to be IU in \cite{worsch2013towards}. These updating schemes restrict which cells can be updated at each time step. The closest related updating scheme to Tile Assembly is ``fully asynchronous'' where only one cell may update at a time.\footnote{For the case of Tile Automata and Surface Chemical Reaction Networks it better stated as ``one rule'' can be applied at a time because two cells can be updated in one update.}

\subparagraph{Passive Self-Assembly.} 
Intrinsic universality first crossed into the self-assembly world in~\cite{Doty_Lutz_Patitz_Summers_Woods_2010}, where a universal tile~set was introduced for systems with tiles that bond with exactly strength 2. Two years later, the first properly intrinsically universal tile set, one that can simulate the full aTAM at any temperature, was presented in \cite{Doty_Lutz_Patitz_Schweller_Summers_Woods_2012}.
These papers both used the definition of intrinsic universality that we call non-committal.
However, these definitions were later corrected to the version that we call committal~\cite{Meunier_Patitz_Summers_Theyssier_Winslow_Woods_2014}.
It was also shown that a single polygon tile type with the ability to flip, translate, and rotate can simulate any aTAM system through several intermediate simulations \cite{Demaine_Demaine_Fekete_Patitz_Schweller_Winslow_Woods_2012}. The aTAM was found not to be committal intrinsically universal at Temperature-1~\cite{Meunier_Patitz_Summers_Theyssier_Winslow_Woods_2014}, and in directed and non-directed planar systems~\cite{Hader_Koch_Patitz_Sharp_2019}.  Directed 3D and Spatial aTAM were proven to be IU \cite{Hader_Koch_Patitz_Sharp_2019}.
The 2-handed self-assembly model is, in general, not intrinsically universal; however, there are intrinsically universal tile sets for each temperature \cite{demaine2016two}. Work towards a universal tile set in Wang Tiles, which studies whether a given tile set can infinitely, and potentially periodically, tile a plane, has also been investigated \cite{lafitte2007universal,lafittealmost,lafitte2010tilings}.

%

\subparagraph{Simulation between Tile Assembly and CA.}
The aTAM can simulate some versions of CA. In particular, it was found that the aTAM can simulate only finite CA~\cite{Hendricks_Patitz_2013}.
The TA model does not have this restriction, as we can infinitely tile the plane with our seed assembly and use transitions to simulate infinite CA.
Where the aTAM is asynchronous, nondeterministic, and finite, Cellular Automata is potentially generally synchronous, deterministic or nondeterministic, and infinite. Tile Automata is asynchronous, deterministic or nondeterministic, and finite. Additionally, Tile Automata is restricted to a neighborhood size of two.

Notable, IU in CA is usually possible with systems that contain a very limited number of states.
However, in self-assembly, the simulating system does not only need to simulate the local interactions between existing states, but importantly also build new tiles in valid locations.
Therefore, IU systems in tile self-assembly tend to use a lot more states.

\subsection{Our Contributions}\label{subsec:contributions} In this paper, we push forward the study of IU systems in a few ways. 
First, we prove that any passive self-assembly model (such as the aTAM) and variants of active self-assembly with bounded state changes cannot adhere to the stronger non-committal definition of intrinsic universality for self-assembly initially presented in~\cite{Doty_Lutz_Patitz_Schweller_Summers_Woods_2012}.
However, this was later corrected and since then, a slightly more permissive definition for the simulation of dynamics for intrinsically universal systems has been used within self-assembly~\cite{Meunier_Patitz_Summers_Theyssier_Winslow_Woods_2014}.
Although this is indication that the problem with modeling dynamics within passive models is known, to our knowledge, this has not been formally proven before.  

Then, we show that in 2D, the seeded Tile Automata model, with unbounded state changes, does indeed adhere to this stronger non-committal definition of intrinsic universality.
We do this by presenting a temperature-1 seeded TA system, and configuration of an initialized seed assembly, that is IU for all seeded temperature-1 systems in approximately 4600 states. We then show that any temperature TA system can be simulated by a temperature-1 TA system. We also prove that the effect of temperature simulation on the scale of the system's supertiles is bounded. No additional states in the IU system's state set are required to simulate systems greater than temperature-1, extending our result to all seeded TA systems. Following this, we show that, due to the mechanics of TA, our construction can be adapted to prove that 2D Asynchronous Cellular Automata, with a cardinal radius of 1 and neighborhood size of 2, is also IU in approximately 2600 states, which although inefficient, is the first 2D ACA IU result.
These positive results are summarized in Table~\ref{tab:results} together with other known IU results.

Section~\ref{sec:preliminaries} starts by giving precise definitions of the model.
Then, we show that bounded state change systems can never be IU in Section~\ref{sec:impossibility_for_passive_systems}.
Opposing this negative result, we continue to show that Tile Automata systems with their unlimited state changes are IU.
Due to the volume of necessary details, the paper first gives a high-level overview in Section~\ref{sec:hloverview}, that discusses the main gadgets and the framework of how the pieces work together. We reference the more detailed later sections that follow the overview.
Section~\ref{sec:temperature_simulation} then covers the temperature simulation part of the IU framework in depth.
Next, sections~\ref{sec:supertile_section},~\ref{sec:attachment},~and~\ref{sec:transitioning_tiles} detail the supertiles, their construction and how they transition respectively in full detail.
We analyse the number of states in Section~\ref{sec:metrics} and proof the correctness of the simulation in Section~\ref{sec:seeded_results}.
We continue to show how our result transfers over to Cellular Automata in Section~\ref{sec:iu_ta_simulates_2d_async_ca}.
We then summarize the conclusion with Section~\ref{sec:conclusion}.

\begin{table}[t]
    \centering
\begin{tabular}{| c | c  | c | c | c | c |}\hline
\multicolumn{6}{|c|}{\textbf{Intrinsic Universality Across Models}}\\ \hline\hline
 \textbf{Model}  & \textbf{D} & \textbf{$N$} & \textbf{$|T|$ / $|\Sigma|$} &  \textbf{Scale ($S$)} & \textbf{Reference} \\ \hline
 aTAM & 2D & 5 & $>$ 10M & $O(n^4\log(n))$ & \cite{Doty_Lutz_Patitz_Schweller_Summers_Woods_2012} \\ \hline
 aTAM & 3D & 7 & 152\,000 & $O(n^2 \log(n\tau ))$ & \cite{Hader_Koch_Patitz_Sharp_2019} \\ \hline
   \textbf{Seeded TA Temp-}$1$ & \textbf{2D} & \textbf{5} & \textbf{4600} & $O(n^3)$ & \textbf{Theorem \ref{lem:iut1}} \\ \hline
  \textbf{Seeded TA} & \textbf{2D} & \textbf{5} & \textbf{4600} & $O(\min((\tau n)^3,n^9))$ & \textbf{Theorem \ref{thm:genIU}} \\ \hline\hline

 Async. Cellular Automata &  1D & 3 & $O(1)$ & unknown & \cite{worsch2013towards} \\ \hline
   \textbf{Block-Pairwise ACA} & \textbf{2D} & \textbf{2} & \textbf{2600} & $O(n^3)$ & \textbf{Theorem \ref{thm:block-pairwise-aca}} \\ \hline
  \textbf{Pairwise ACA} & \textbf{2D} & \textbf{2} & \textbf{$O(1)$} & $O(n^3)$ & \textbf{Theorem \ref{thm:pairwise-aca}} \\ \hline

 
 
\end{tabular}

\vspace{.2cm}

\caption{An overview of known and new simulation results for types of asynchronous cellular automata including tile assembly models. D is the dimension, $N$ is neighborhood size of the input system, $|T|$ and $|\Sigma|$ are, respectively, the number of tile or state types in the universal system, $S$ is the scale factor, $n$ is the number of states in the input system, 
and $\tau$ is the system's temperature.}

\label{tab:results}
\end{table}

\section{Preliminaries}\label{sec:preliminaries}
In this section, we cover the basics of the Tile Automata model. We use many of the same definitions as in \cite{alaniz2023building,chalk2018freezing}. An example of a Tile Automata system can be seen in Figure~\ref{fig:example_system}.

\begin{figure}[t]
    \centering
    \includegraphics[width=1.\textwidth]{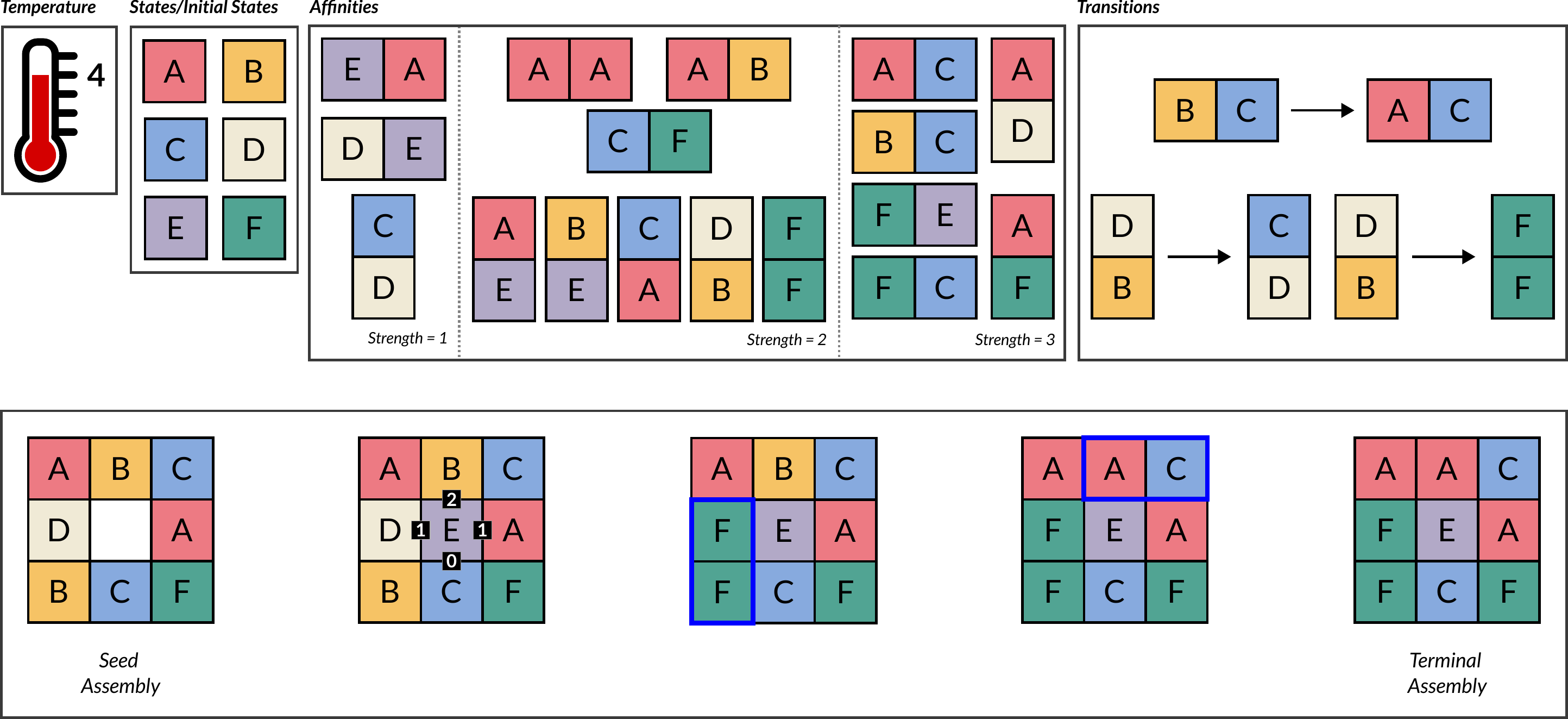}
    \caption{Example of a Tile Automata system with 6 states, a system temperature of 4, affinities of strengths 1, 2, and 3 vertical and horizontal transitions, and a seed assembly. The assembly sequence to a terminal assembly is also shown with the changes highlighted. Due to the affinity strengthening restriction, there is no detachment.}
    \label{fig:example_system}
\end{figure}

\textbf{Tiles.} Let $\Sigma$ be a set of \emph{states}. A tile $t = (\sigma, p)$ is a non-rotatable unit square placed at point $p \in \mathbb{Z}^2$ and has a state of $\sigma \in \Sigma$. Let $\sigma(t)$ be the state of $t$.  Let $\phi$ denote a special state called the \emph{empty} state.

\textbf{Affinity Function.} An \emph{affinity function} $\Pi$ over a set of states $\Sigma$ takes an ordered pair of states $(\sigma_1, \sigma_2) \in \Sigma \times \Sigma$ and an orientation $d \in D$, where $D = \{\perp,\vdash\}$, and outputs an element of $\mathbb{Z}^{0+}$. The orientation $d$ is the relative position of $\sigma_1$ to $\sigma_2$, with $\vdash$  meaning horizontal and $\perp$ meaning vertical. State $\sigma_1$ is the west or north state respectively. We refer to the output as the \emph{Affinity Strength} between these two states.

\textbf{Transition Rules.}
A \emph{Transition Rule} consists of two ordered pairs of states $(\sigma_1, \sigma_2), (\sigma_3, \sigma_4)$ and an orientation $d \in D$, where $D = \{\perp,\vdash\}$. The rule denotes that if the tiles with states $(\sigma_1, \sigma_2)$ are next to each other in orientation $d$ ($\sigma_1$ as the west/north state) they may be replaced by the states $(\sigma_3, \sigma_4)$.

\textbf{Assembly.} An assembly $A$ is a set of tiles (with states in $\Sigma$), such that no two tiles occupy the same position, i.e., for every pair of tiles $t_1 = (\sigma_1, p_1), t_2 = ( \sigma_2, p_2)$, it holds that $p_1 \neq p_2$. For an assembly $A$, let $A(x,y)$ denote the state of the tile with location $(x,y) \in \mathbb{Z}^2$ in $A$ if such a tile exists and $\phi$ (the \emph{empty} state) otherwise.  For a set of states $\Sigma$, let $A^\Sigma$ denote the set of all assemblies over state set $\Sigma$. 

Let the \emph{bond graph}~$B_G(A)$ be formed by taking a node for each tile in $A$ and adding an edge between neighboring tiles $t_1 = (\sigma_1, p_1)$ and $t_2 = (\sigma_2, p_2)$ in orientation~$d$ with a weight equal to $\Pi(\sigma_1, \sigma_2, d)$.
We say an assembly $A$ is $\tau$-stable for some $\tau \in \mathbb{Z}^{0+}$ if the minimum cut through $B_G(A)$ is greater than or equal to $\tau$.

\subsection{The Seeded Tile Automata Model}\label{subsec:the_seeded_ta_model}
In this paper, we investigate the \emph{Seeded Tile Automata} model, which differs from the non-seeded Tile Automata model defined above, by only allowing single tile attachments to a growing seed, similar to the aTAM. Here we use many of the same definitions as in \cite{alaniz2023building}.

\textbf{Seeded Tile Automata.}
A Seeded Tile Automata system~$\Gamma$ is a 6-tuple $\{ \Sigma, \Lambda, \Pi, \Delta, s, \tau\}$ where $\Sigma$ is a set of states, $\Lambda \subseteq \Sigma$ a set of \emph{initial states}, $\Pi$ is an \emph{affinity function}, $\Delta$ is a set of \emph{transition rules},  $s$ is a stable assembly called the \emph{seed} assembly consisting of tiles in states contained in $\Sigma$, and $\tau \in \mathbb{Z^+}$ is the \emph{temperature} (or \emph{threshold}). When we refer to a \emph{tile set} (or equivalently \emph{rule set}) we mean the four parameters $(\Sigma, \Pi, \Delta, \tau)$, that is, the states, the affinity function, the transition rules, and the temperature.  A system $\Gamma = \{ \Sigma, \Lambda, \Pi, \Delta, s, \tau\}$ is said to \emph{use} rule/tile set $(\Sigma, \Pi, \Delta, \tau)$.

\textbf{Attachment Step.}
A tile $t = (\sigma, p)$ may attach to an assembly $A$ at temperature $\tau$ to build an assembly $A' = A \bigcup t$ if $A'$ is $\tau$-stable and $\sigma \in \Lambda$. We denote this as $A \rightarrow_{\Lambda, \tau} A'$.

\textbf{Transition Step.}
An assembly $A$ can transition to an assembly $A'$ if there exist two neighboring tiles $t_1 = (\sigma_1, p_1), t_2 = (\sigma_2, p_2) \in A$ (where $t_1$ is the west or north tile) such that there exists a transition rule in $\Delta$ with the first pair being $(\sigma_1, \sigma_2)$ and ${A' = (A \setminus \{t_1, t_2\} ) \bigcup \{ t_3 = (\sigma_3, p_1), t_4 = (\sigma_4, p_2)\}}$. We denote this as $A \rightarrow_\Delta A'$.

\textbf{Affinity Strengthening.}
We only consider transitions rules that are affinity strengthening, meaning for each transition rule $((\sigma_1, \sigma_2), (\sigma_3, \sigma_4), d)$, the bond between $(\sigma_3, \sigma_4)$ must be at least the affinity strength of $(\sigma_1, \sigma_2)$ and it must also maintain or increase any other neighbor affinities. Formally, $\Pi(\sigma_3, \sigma_4, d) \geq \Pi(\sigma_1, \sigma_2, d)$ and $ \Pi(\sigma_3,\sigma_i, d) \geq \Pi(\sigma_1,\sigma_i, d)$ and $\Pi(\sigma_4,\sigma_j, d) \geq \Pi(\sigma_2,\sigma_j, d)$ $\forall i,j \in \Sigma$. This ensures that transitions may not induce cuts in the bond graph.

\textbf{Producibles.}
We refer to both attachment and transition steps as production steps and say that $A \rightarrow^{\Gamma}_1 A'$ if either $A \rightarrow_{\Lambda, \tau} A'$ or $A \rightarrow_{\Delta} A'$.  For any sequence of assemblies $\{A_1, A_2, \ldots A_k\}$ such that $A_i \rightarrow^{\Gamma}_1 A_{i+1}$ for all $1\leq i < k$, we say that $A_k$ is producible from $A_1$, and write $A_1\rightarrow^{\Gamma} A_k$.  Note that for any assembly $A$, $A\rightarrow^{\Gamma} A$.  We say $A\rightarrow^{\Gamma}_{\geq 1} B$ if $A\rightarrow^{\Gamma} B$ and $A\neq B$.
For a Tile Automata system $\Gamma = \{ \Sigma, \Lambda, \Pi, \Delta, s, \tau\}$ we refer to the set 
$\prodasm{\Gamma}= \{s\} \bigcup \{A | s\rightarrow^\Gamma A\}$ as the \emph{producible assemblies} of $\Gamma$.

\textbf{Terminal Assemblies.}
The set of terminal assemblies for a Tile Automata system $\Gamma = \{ \Sigma, \Lambda, \Pi, \Delta, s, \tau\}$ is written as $TERM(\Gamma)$. This is the set of assemblies that cannot grow or transition any further. Formally, an assembly $A \in TERM(\Gamma)$ if $A \in \prodasm{\Gamma}$ and there does not exists any assembly $A' \in \prodasm{\Gamma}$ such that $A\rightarrow^{\Gamma}_1 A'$.

\textbf{Unique Assembly.}
A Tile Automata system $\Gamma = \{ \Sigma, \Lambda, \Pi, \Delta, s, \tau\}$ \emph{uniquely} assembles an assembly $A$ if $A \in TERM(\Gamma)$, and for all $A' \in \prodasm{\Gamma}, A' \rightarrow^{\Gamma} A$.

\subsection{Simulation}\label{subsec:simulation}
In this section, we formally define the concept of one tile automata system \emph{non-committally} simulating another. We use a standard $m$-block simulation in which each tile of an assembly is simulated by a larger $m\times m$ block of tiles in the simulating system.
The definition presented here is the same as that originally presented in \cite{Doty_Lutz_Patitz_Schweller_Summers_Woods_2012}, which we call non-committal IU.
However, as stated before, that paper contained a subtle error that was later corrected to become committal IU.

The difference lies in the \emph{models} concept, see Definition~\ref{def:models}.
In non-committal IU, this definition contains a universal quantifier, whereas the committal version contains a weaker statement.
From here on, we focus on two tile Automata systems $\Gamma_T$ and $\Gamma_S$, where $\Gamma_S$ denotes a system that purports to \emph{simulate} system $\Gamma_T$. Let $\Sigma_{T}$ and $\Sigma_{S}$ denote the set of states used in $\Gamma_T$ and $\Gamma_S$, respectively.

\textbf{$m$-block Supertiles.}
An \emph{m-block supertile} over a set of states $\Sigma$ is a partial function~${\lambda: \mathbb{Z}_m \times \mathbb{Z}_m \rightarrow \Sigma}$, where $\mathbb{Z}_m=\{0,1,\dots,m-1\}$.  Let $B_m^{\Sigma}$ be the set of all \emph{m}-block supertiles over $\Sigma$.  The \emph{m}-block with no domain is said to be \emph{empty}.  For any assembly $\mathcal{A}$ over state space~$\Sigma$, define ${\mathcal{A}}_{x,y}^m$ to be the \emph{m}-block defined by ${\mathcal{A}}_{x,y}^m (i,j) = {\mathcal{A}}(mx + i, my + j)$ for $0 \le i, j < m$.

\textbf{Supertile representation and mapping.}
We refer to a function $R: B_m^{\Sigma_{S}} \rightarrow \Sigma_{T} \bigcup \{\phi\}$ as an \emph{m-block representation function}.  We require $R(B)=\phi$ for the empty $m$-block, and for any non-empty $m$-block $B$ for which $R(B)=\phi$, we say $B$ maps to a \emph{ghost tile}.  For a given $m$-block representation function $R$, define the partial function $R^*: A^{\Sigma_S} \rightarrow A^{\Sigma_T}$ such that~${R^*(\mathcal{A}) = \mathcal{A'}}$ if and only if $A'(x,y) = R(A^m_{x,y})$ for all $(x,y)\in \mathbb{Z}^2$.

\textbf{c-Fuzz.} 
The concept of c-fuzz is basically that a macroblock may have a bounded number of ``extra'' tiles attached to it without altering its mapping. This allows a simulating system to make minor intermediate attachments while enacting the simulation. Another way to think of c-fuzz is as a reasonable allowance for limited-size non-empty macro-blocks (that map to an empty tile in the simulated system) to be used in the simulation process. Formally, a mapping $R^*(A) = A'$ is said to map to $A$ with at most $c$-fuzz, for some $c \in \mathbb{Z}^+$, if and only if for all non-empty blocks $A^m_{x,y}$ it is the case that $R(A^m_{x + u, y + v}) \neq \phi$ for some integers~$u,v \in [-c , c]$.  In other words, any non-empty macro blocks that map to $\phi$ (i.e., ghost tiles) are only at most $c$ macroblocks away from a macroblock that maps to a real (non-empty) tile.  We say a Tile Automata system achieves $c$-fuzz under mapping $R^*$ if each producible assembly of the system achieves at most $c$-fuzz when mapped by $R^*$.

\begin{definition} [Equivalent Productions]
    We say $\Gamma_S$ has equivalent productions to $\Gamma_T$ (under~$R$) with up to $c$-fuzz, and write $\Gamma_S \Leftrightarrow_c \Gamma_T$, if the following hold:
    \begin{enumerate}
        \item $\left\{R^*(A') | A' \in \prodasm{\Gamma_S}\right\} = \prodasm{\Gamma_T}$.
        \item $\Gamma_S$ achieves $c$-fuzz under $R^*$.
    \end{enumerate}
\end{definition}

\begin{definition}[Follows]
    We say that $\Gamma_T$ follows $\Gamma_S$ (under $R$), and write $\Gamma_T \dashv_R \Gamma_S$, if~$A' \to^{\Gamma_S} B'$, for some $A',B' \in \prodasm{\Gamma_S}$, implies that $R^*(A') \to^{\Gamma_T} R^*(B')$. 
\end{definition}

\begin{definition}[Non-Committally Models]\label{def:models}
    We say that $\Gamma_S$ (non-committally) models $\Gamma_T$, and write $\Gamma_S \models_R \Gamma_T$,  if $A \to^{\Gamma_T} B$ for some $A,B \in \prodasm{\Gamma_T}$, implies that for all $A'$ such that $R^*(A')=A$, $A' \to^{\Gamma_S} B'$ for some $B'\in\prodasm{\Gamma_S}$ with $R^*(B')=B$. 
\end{definition}

\begin{definition}[Non-Committal Simulation]
    We say $\Gamma_S$ (non-committally) simulates $\Gamma_T$ if for some $c\in \mathbb{Z^+}$, $\Gamma_S \Leftrightarrow_c \Gamma_T$ (equivalent productions), $\Gamma_T \dashv_R \Gamma_S$ and $\Gamma_S \models_R \Gamma_T$ (equivalent dynamics).  We say the simulation is \emph{clean} if it holds for $c=1$, and we say the simulation achieves $c$-fuzz more generally.
\end{definition}

\begin{definition}[Non-committal Intrinsic Universality.]
    A rule (tile) set $I = \{ \Sigma, \Pi, \Delta, \tau\}$ is said to be intrinsically universal for a set of systems $U$ if for all $\Gamma_T \in U$, there exists a system $\Gamma_S = \{ \Sigma, \Lambda_T, \Pi, \Delta, s_T, \tau\}$ that non-committally simulates $\Gamma_T$.  
    The set $U$ itself is said to be intrinsically universal if there exists a rule set $I$ used by some system within $U$ such that $I$ is intrinsically universal for $U$.  A model is said to be intrinsically universal if the set of all systems within that model is intrinsically universal.
\end{definition}

We use the term \emph{non-committal} simulation to emphasize that the simulation definition we use is stronger than what is used in prior work, which we call \emph{committal}.  For the remainder of this paper, we deal exclusively with \emph{non-committal} simulation and will just use the term \emph{simulation} when not directly comparing with previous versions of simulation.
\section{Impossibility for Passive or Bounded State Change Systems}\label{sec:impossibility_for_passive_systems}

In this section, we show that systems lacking the full state changing capability of the seeded Tile Automata model cannot achieve intrinsic universality under non-committal simulation.  This includes well-studied models such as the Abstract Tile Assembly Model (aTAM)~\cite{winfree1998algorithmic} and \emph{freezing} variants of the seeded Tile Automata Model~\cite{chalk2018freezing}.  The key aspect of non-committal simulation that is impossible for these models is the non-committal modeling requirement of our simulation definition.  In this section we show the impossibility of simulating a specific \emph{passive} system, with the key use of the non-committal modeling requirement being used to prove Lemma~\ref{lemma:extensibility}.

\begin{definition}[$k$-burnout, bounded, unbounded, passive, freezing]
    For a non-negative integer~$k$, a system is a $k$-burnout system if each tile in an assembly is restricted to only changing state at most $k$ times.  A system is called bounded if it is $k$-burnout for some $k$, and unbounded otherwise.  $0$-burnout systems are termed as passive.  A freezing system~\cite{chalk2018freezing} is one in which state change rules are such that a tile can never return to a previous held state.
\end{definition}

\begin{observation}
    Any aTAM system is a passive ($0$-burnout) system, and any freezing seeded TA system that uses $|\Sigma|$ states is bounded and a $|\Sigma|$-burnout system. 
\end{observation}

\begin{definition}
    Define $X_n$ to be the passive seeded Tile Automata system consisting of states~${\Sigma = \{S,a_1,a_2,\ldots, a_n\}}$ with seed tile in state $S$, and east-west affinity between $S$ and each $a_i$ of strength equal to the system temperature $\tau$.  Let $s\cdot a_i$ denote the producible assembly of this system obtained by attaching a tile of state $a_i$ to the east of the seed tile.
\end{definition}

For the remainder of this section, let $R$ denote a proposed $m$-block mapping function from macro blocks from a proposed simulator system to tiles from the system $X_n$.  

\begin{lemma}\label{lemma:extensibility}
For any system $Y$ that simulates $X_n$ under mapping $R$, and for any valid assembly sequence $\langle A_{\pi_1},\ldots, A_{\pi_m} \rangle$ of $Y$ such that for all $1\leq i \leq m$, $R^*(A_{\pi_i}) = s$, either:
\begin{enumerate}
\item For all $1\leq i \leq n$ there exists an assembly $A_i$ such that $A_{\pi_m} \rightarrow^Y_1 A_i$ and $R^*(A_i)=s\cdot x_i$, or
\item there exists an assembly $A_{\pi_{m+1}}$ such that $A_{\pi_m} \rightarrow^Y_1 A_{\pi_{m+1}}$ and $R^*(A_{\pi_{m+1}}) = s$.
\end{enumerate}
\end{lemma}

\begin{proof}
    Suppose constraint (1) does not hold for such a sequence $\langle A_{\pi_1},\ldots A_{\pi_m} \rangle$ of $Y$, i.e. suppose that for some $1\leq i \leq n$ there does not exists an assembly $A_i$ such that $A_{\pi_m} \rightarrow^Y_1 A_i$ and $R^*(A_i)= s\cdot a_i$.  Since $Y \models_R X_n$ ($Y$ non-committally models $X_n$), and $s \rightarrow^{X_n}s\cdot a_i$, it must be that there exists some assembly $A_i$ such that $A_{\pi_m} \rightarrow^Y A_i$ and $R^*(A_i)=s\cdot a_i$, which by definition means there exists an assembly sequence $\langle A_{\pi_m}, B, \ldots, A_i \rangle$.  Therefore,~${A_{\pi_m} \rightarrow^Y_1 B}$, and since $X_n \dashv_R Y$ ($X_n$ follows $Y$), we know that $R^*(B)= s$, which means the sequence~${\langle A_{\pi_1},\ldots, A_{\pi_m} \rangle}$ can be extend with assembly~$B$ to satisfy constraint (2).
\end{proof}

\begin{lemma}\label{lemma:equalLastAssembly}
For any bounded system $Y$ that simulates $X_n$ under mapping $R$ there must exist $A, A_1,\ldots, A_n \in \textbf{PROD}_Y$ such that $R^*(A)=s$, and for all $1\leq i \leq n$, $R^*(A_i)=s\cdot a_i$ and $A \rightarrow^Y_1 A_i$.
\end{lemma}

\begin{proof}
    Suppose a bounded system $Y$ simulates $X_n$.  Since $Y$ is bounded, there must exist~${M \in \mathbb{Z^+}}$ such that for all assembly sequences $\langle A_{\pi_1},\ldots, A_{\pi_m} \rangle$ of $Y$ where $R^*(A_{\pi_i}) = s$, it is the case that $m \leq M$.  This is the case since each assembly $A_{\pi_i}$ maps to a single tile under~$R^*$, thereby limiting the size of each $A_{\pi_i}$ to a finite integer based on the (finite) scale-factor of the simulation and the (finite) fuzz factor $c$ of the simulation.  The number of state changes and tile attachments for each assembly $A_{\pi_i}$ therefore has a finite bound in a system with a finite burnout number.

    Since the length of such sequences cannot be extended infinitely, there must exist a sequence $\langle A_{\pi_1},\ldots A_{\pi_m} \rangle$ for which no additional $A_{\pi_{m+1}}$ exists for which $A_{\pi_m} \rightarrow^Y_1 A_{\pi_{m+1}}$ and $R^*(A_{\pi_{m+1}}) = s$.  For this sequence that must exist, Lemma~\ref{lemma:extensibility} implies that 
    for all $1\leq i \leq n$ there exists an assembly $A_i$ such that $A_{\pi_m} \rightarrow^Y_1 A_i$ and $R^*(A_i)=s\cdot x_i$.  Therefore, there must exist $A=A_{\pi_m}, A_1, \ldots, A_n \in \textbf{PROD}_Y$ that satisfy the requirements of the lemma.
\end{proof}

\begin{lemma}\label{lemma:noBounded}
    A bounded seeded TA system with fewer than $n^\frac{1}{5}$ states cannot simulate $X_n$.
\end{lemma}
\begin{proof}
If a bounded system $Y$ simulates $X_n$, then by Lemma~\ref{lemma:equalLastAssembly} it must be the case that there exists ${A, A_1,\ldots, A_n \in \textbf{PROD}_Y}$ such that $R^*(A)=s$, and for all $1\leq i \leq n$ it holds that $R^*(A_i)=s\cdot a_i$ and $A \rightarrow^Y_1 A_i$. Since $A \rightarrow^Y_1 A_i$ for each $A_i$, we know that each pair of assemblies $A_i$ and $A_j$ differ at either a single point or two adjacent points (corresponding to a tile attachment or a pairwise state change).  We now consider two cases:

Case 1:  Suppose there exists an $i$ and $j$ such that $A_i$ and $A_j$ differ at non-overlapping points.  In this case, we know that the rule or attachment applied to $A$ to attain $A_i$ is also applicable to $A_j$, and vice versa. This implies there exists a common assembly $A_{i\oplus j}$ such that $A_i \rightarrow^Y A_{i\oplus j}$ and $A_j \rightarrow^Y A_{i\oplus j}$.  But since $X_n \dashv_R Y$ ($X_n$ follows $Y$), it must then be the case that $R^*(A_i) \rightarrow^{X_n} R^*(A_{i\oplus j})$ and $R^*(A_j) \rightarrow^{X_n} R^*(A_{i\oplus j})$.  This implies that $R^*(A_{i\oplus j}) = s\cdot a_i$ and $R^*(A_{i\oplus j}) = s\cdot a_j$, which is a contradiction.

Case 2: Suppose the points of difference for all assemblies $A_i$ overlap each other in at least one of their points.  Since each assembly's pair of points are adjacent (in the case that there are two), this implies that the union of all such points of difference is at most 5 points.  If $Y$ has $|\Sigma_Y|$ states, then there are at most $|\Sigma_Y|^5$ distinct state assignments possible for this 5-tile region.  Thus, if $|\Sigma_Y| < n^{\frac{1}{5}}$, then there are fewer than $n$ distinct 5-tile regions, implying that $A_i = A_j$ for some $i\neq j$, which is a contradiction since $R(A_i) \neq R(A_j)$. 
\end{proof}

\begin{lemma}\label{lemma:noPassiveSim}
    A passive seeded TA system with fewer than $n$ states cannot simulate $X_n$.
\end{lemma}
\begin{proof}
Suppose a proposed system $Y$ with fewer than $n$ states simulates $X_n$ with representation function $R^*$.  By Lemma~\ref{lemma:equalLastAssembly} there exists $A, A_1,\ldots, A_n \in \textbf{PROD}_Y$ such that $R^*(A)=s$, and for all $1\leq i \leq n$, $R^*(A_i)=s\cdot a_i$ and $A \rightarrow^Y_1 A_i$.  As each $A_i$ is attained by attaching a single tile to $A$, let point $p_i$ denote the point of this attached tile in assembly $A_i$.  Now consider two cases:

Case 1:   Suppose there exist $1\leq i,j \leq n$ such that $p_i \neq p_j$.  In this case the tile attached to form $A_j$ from $A$ can also be attached to $A_i$, and vice versa, implying that the assembly consisting of attaching both such tiles, call it $A_{i\oplus j}$, is such that $A_i \rightarrow^Y A_{i\oplus j}$ and $A_j \rightarrow^Y A_{i\oplus j}$.  But since $X_n \dashv_R Y$ ($X_n$ follows $Y$), it must be that $R^*(A_i)\rightarrow^{X_n} R^*(A_{i\oplus j})$ and $R^*(A_j)\rightarrow^{X_n} R^*(A_{i\oplus j})$, which implies that $R^*(A_{i\oplus j}) = s\cdot a_i$ and $R^*(A_{i\oplus j}) =s\cdot a_j$, which is a contradiction.

Case 2:  Suppose instead that $p_i = p_j$ for all $1\leq i,j\leq n$.  Since $Y$ has less than $n$ states, there must exist some $1\leq i,j\leq n$ such that $A_i$ and $A_j$ use the same state at point $p_i = p_j$.  This implies that $A_i = A_j$, which is a contradiction since $R^*(A_i) \neq R^*(A_j)$.
\end{proof}

Lemmas~\ref{lemma:noBounded} and \ref{lemma:noPassiveSim} show that without the unbounded state change capability of the full seeded Tile Automata model, there exists a simple class of passive systems that cannot be simulated under non-committal simulation without arbitrarily larger state spaces. This gives us the following negative results for two established models in regards to non-committal simulation, directly following from Lemma~\ref{lemma:noPassiveSim} and~\ref{lemma:noBounded}:

\begin{restatable}{theorem}{atamimp}
\label{thm:atamImp}
    The Abstract Tile Assembly Model (aTAM) is not intrinsically universal under non-committal simulation.
\end{restatable}

\begin{restatable}{theorem}{frzimp}
\label{thm:frzImp}
    The Freezing Seeded Tile Automata model is not intrinsically universal under non-committal simulation.
\end{restatable}

\section{Overview of Intrinsic Universality in TA}\label{sec:hloverview}
Now that we have shown that any bounded system cannot be intrinsically universal under non-committal simulation, we will show that Tile Automata (TA), with its unbounded state changes, is non-committal intrinsically universal.
We do so by characterizing a TA system that can simulate any other TA system.
For ease of presentation, we first give a high-level overview of the framework and some of the main techniques used to achieve the simulation of any TA system with the one presented.
Both a detailed exposition of the techniques, as well as any proofs omitted in this section can be found in the later sections covering the details.

We show that seeded TA is IU by first showing that temperature-1 seeded TA is IU at scale $O(|\Sigma|^3)$ with constant states. We then show how we can simulate a seeded TA system at any temperature with a temperature-1 system at scale-1 with $O(\min(|\Sigma|^3,\tau|\Sigma|))$ states. These combine to create a general IU result for any seeded TA system by bounding the number of states to a constant for any temperature, or by scaling based on the temperature.

\subsection{Temperature-1 Seeded TA is Intrinsically Universal}

Section \ref{sec:temperature_simulation} gives the full details for the IU results with temperature-1. We show that there exists an intrinsically universal temperature-1 system with a constant number of states if we increase the scale factor to $O(|\Sigma|^3)$. Here, we give an overview of the framework used to prove the following.

\begin{restatable}{theorem}{IUone}
\label{lem:iut1}
    There exists a tile set $(\Sigma_{U}, \Lambda_U, \Pi_{U}, \Delta_{U})$ such that, for all systems $\Gamma = (\Sigma, \Lambda, \Pi, \Delta, s, 1)$, there exists a $\Gamma' = (\Sigma_{U}, \Lambda_{U}, \Pi_{U}, \Delta_{U}, s_U, 1)$ that simulates $\Gamma$ at scale $O(|\Sigma|^3)$. 
\end{restatable}

\subparagraph{Supertiles.}\label{subsec:supertile}
Supertiles are $m \times m$ blocks that map to a specific tile in some state from the original system that is being simulated. 
Each supertile contains all information necessary for the simulation.
For specifics, please refer to the detailed walk-through of the operations in Section~\ref{sec:supertile_section}.

Figure~\ref{fig:simple_overview_supertile} shows a simplified diagram of a single supertile. Every supertile has binding sites on each of the four sides, and wires on each side that lead to a central lookup table corresponding to valid affinities and transitions for the system being simulated. Within the table each column represents a state in the system being simulated and each row a state and direction of a neighboring supertile. There are datacells at each intersection. 

The supertile makes use of several small gadgets for effective and correct communication and information transmission, as well as ensuring non-committal simulation.
The most important gadgets are listed here. For a complete explanation of their workings and purpose, see Section~\ref{sec:supertile_section}.

\begin{itemize}
    \item \textbf{Lookup Table.} Each supertile contains a lookup table that contains all the affinity and transition information of the system to be simulated. Thus, every tile has all information necessary to update itself from its neighbors. 

    \item \textbf{Transition Storage Area.} All of the transitions for a pair of states are stored within a storage area in ordered data strings. For each transition rule, only the halves of pairs pertaining to the current supertile are stored.
    
    \item \textbf{Datacells.} For each directed pair of states in the system, a datacell stores the possible transitions and affinity status between them. The datacell is a compound gadget comprised of a transition storage area, and a single tile containing the affinity and current state status. For non-committal IU, the affinity must be chosen in one step, which is why the affinity is stored in a single tile.

    \item \textbf{Transition Selection.} For non-committal IU, any change to the mapping of a supertile must occur with a single tile placement or transition. This requires careful collaboration and ordering around the tiles that can change this mapping. Most of the information must be obtained from the lookup table and brought to the edge. The transition selection gadget contains this reversible process of getting the supertile (or supertiles) ready to change the mapping (or mappings) irreversibly in such a way that it could be reversed at any point up until the single mapping transition. Once the state mapping has changed, it must communicate this information to the rest of the supertile.
\end{itemize}

\begin{figure}[t]
    \centering
    \includegraphics[width=.8\textwidth]{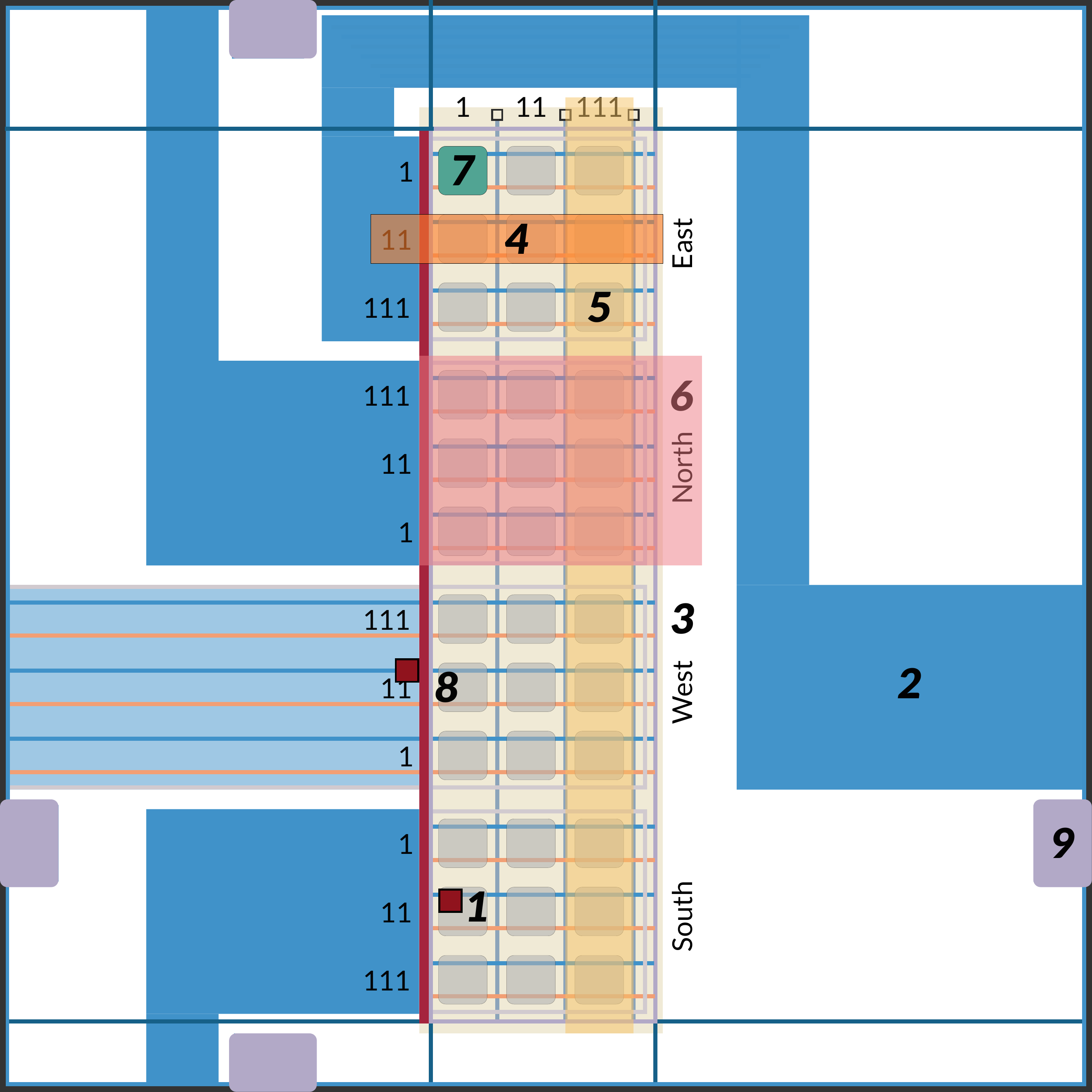}
    \caption{An overview of a supertile. (1) An agent inside of the supertile. (2) Wires connecting supertiles from each edge to the lookup table. West wires are drawn individually. (3) The lookup table storing the information about the system being simulated. (4) A row containing the information about the state of the east neighbor of the supertile. (5) The active column, representing the current state of the supertile. (6) A group of datacells storing all information for the north side. (7) A single datacell, in this case, storing the affinities and transitions for when both this supertile and the East supertile are in state 1. (8) The table control edge, with an agent waiting to enter. (9) Transition selection gadget at each edge, dictating the transition of this supertile with its east neighbor.}
    \label{fig:simple_overview_supertile}
\end{figure}

\subparagraph{Attachments.}
The attachment process for a new supertile works approximately as follows, the details of which can be found in Section~\ref{sec:attachment}. Attachment is triggered by a supertile, when it discovers that no neighbor exists adjacent to it.
The builder supertile finds that an affinity in that direction exists, it prepares itself for construction, by locking its outer edge, wiping its wires, and deactivating its gadgets. 

Next, it starts building the new supertile.
Should the supertile find a competitor trying to build in this spot, one of the two nondeterministically prevails.
The builder supertile then copies over each part of the supertile one by one.

Once built, the new supertile requests the states of all its neighbors to select its own state. From the valid states, one is chosen nondeterministically, and the representative state column is activated. Finally, the new supertile's state is sent to its neighbors.

\subparagraph{Transitions.}
The process of transitioning happens in seven general phases.
The full details can be found in Section~\ref{sec:transitioning_tiles}.

\begin{enumerate}
    \item The existence of one or more transitions between two neighboring supertiles is confirmed, and each supertile's table is locked for the duration of the transition process.  
    \item The data strings within the transition storage area of the datacell at each supertile's respective intersections are copied and transmitted to the transition selection gadget. 
    \item Agents within the transition selection gadget nondeterministically select the new states associated with a transition rule or abort the transition altoghether.
    \item Once a transition has been chosen, the new states are copied and transmitted to each supertile's respective tables for updating.
    \item In the table, the old state is deselected, and the new state is activated. 
    \item The transition selection gadget is wiped. 
    \item The tables are unlocked, and new states are transmitted to neighboring supertiles.
\end{enumerate}

\subsection{Temperature Simulation at Scale-1}
Using these techniques we will show that seeded TA at temperature-1 is intrinsically universal. 
We also show that at scale-1, we can simulate a seeded TA system at any temperature if we scale the number of states in the system.
We provide two bounds on the scale factor: $O(\min(\tau|\Sigma|,|\Sigma|^3)$.
Resulting in the following two Lemmata.

\begin{restatable}{lemma}{temponesim}
\label{thm:temp1sim}
For all Tile Automata systems $\Gamma_\tau = ( \Sigma, \Lambda, \Pi, \Delta, s, \tau )$ there exists a system $\Gamma_1 = ( \Sigma_1, \Lambda_1, \Pi_1, \Delta_1, s_1, 1 )$ that simulates it with 1-fuzz at scale-1 such that $|\Sigma_1| = O(\tau|\Sigma|)$. 
\end{restatable}

\begin{restatable}{lemma}{temponesimtwo}
\label{thm:temp1sim2}
For all Tile Automata systems $\Gamma_\tau = ( \Sigma, \Lambda, \Pi, \Delta, s, \tau )$ there exists a system $\Gamma_1 = ( \Sigma_1, \Lambda_1, \Pi_1, \Delta_1, s_1, 1 )$ that simulates it with 1-fuzz at scale-1 such that $|\Sigma_1| = O(|\Sigma|^3)$.
\end{restatable}

\subsection{Seeded TA is Intrinsically Universal}
By taking Theorem \ref{lem:iut1} in conjunction with Lemmas \ref{thm:temp1sim} and \ref{thm:temp1sim2}, we achieve the desired result that seeded Tile Automata is non-committal intrinsically universal. This follows by directly plugging in the state-scaling into the temperature-1 construction.

\begin{restatable}{theorem}{genIU}
\label{thm:genIU}
    There exists a tile set $(\Sigma_{U}, \Pi_{U}, \Delta_{U}, 1)$ such that, for all systems $\Gamma = (\Sigma, \Lambda, \Pi, \Delta, s, \tau)$, there exists a $\Gamma' = (\Sigma_{U}, \Lambda_{U}, \Pi_{U}, \Delta_{U}, s', 1)$ that simulates $\Gamma$ with 1-fuzz at scale factor $O(\min((\tau|\Sigma|)^3,|\Sigma|^9)$. 
\end{restatable}

\section{Temperature Simulation}\label{sec:temperature_simulation}
We now give a detailed construction of our universal Tila Automata system.
We show how any Tile Automata system ${\Gamma = ( \Sigma, \Lambda, \Pi, \Delta, s, \tau )}$ of any temperature $\tau$ can be simulated by temperature 1 with 1-fuzz by using ghost tiles and adding intermediary states.

In order to attach tiles that require \emph{cooperative binding}, the necessity of needing multiple neighbors to attach a single tile in order to reach necessary affinity strength, we use intermediary states to add together the affinity strengths of surrounding tiles to the interim state tile that is attempting to be placed at that location, see Figure~\ref{fig:sim_system} for an example, and Figure~\ref{fig:sim_system_construction} for the assemblies it produces.

\begin{figure}[t]
    \centering
    \includegraphics[width=.9\textwidth]{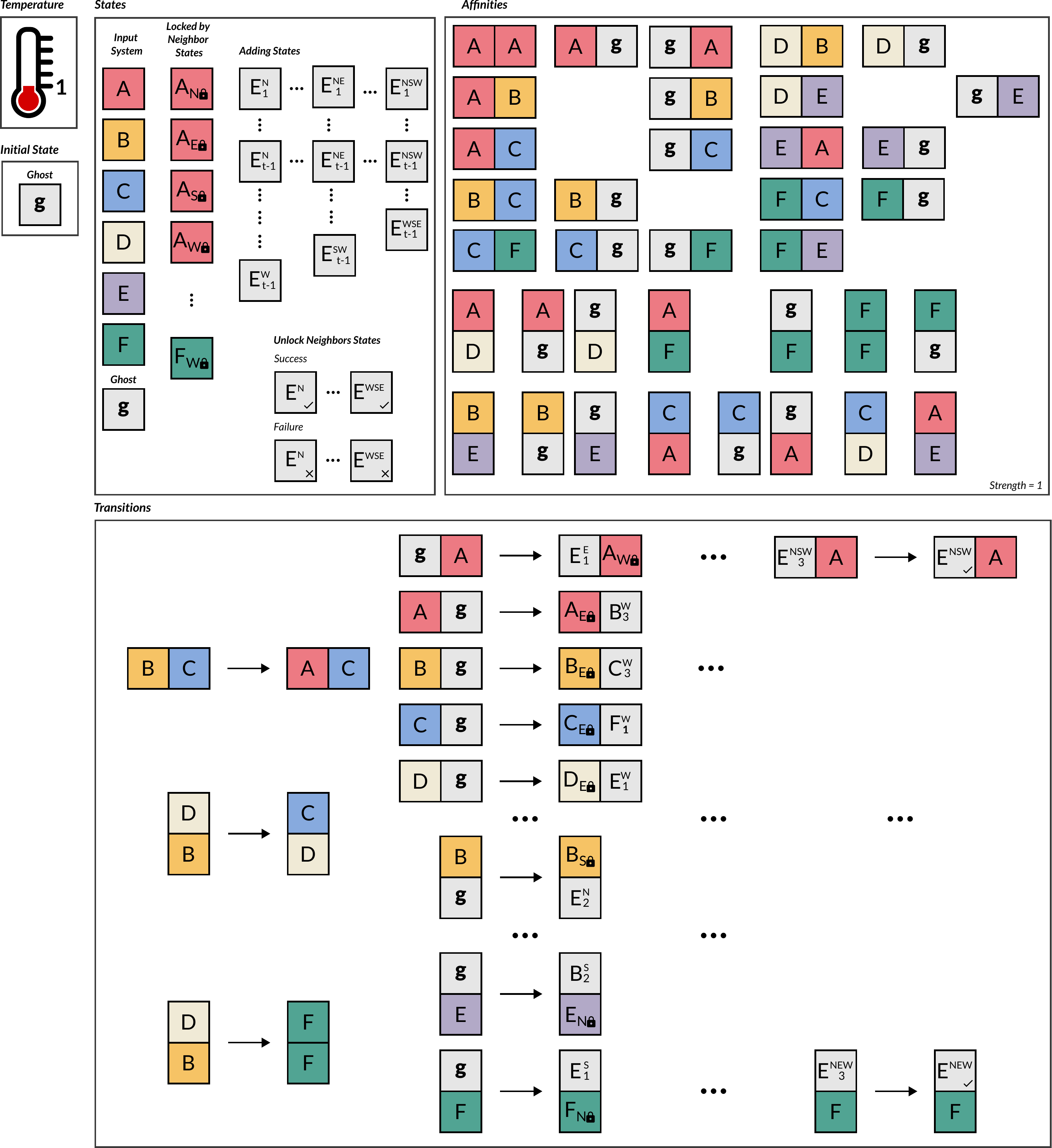}
    \caption{The temperature-1 system that simulates the system in Figure~\ref{fig:example_system}.}
    \label{fig:sim_system}
\end{figure}

\begin{figure}[t]
    \centering
        \includegraphics[width=1.\textwidth]{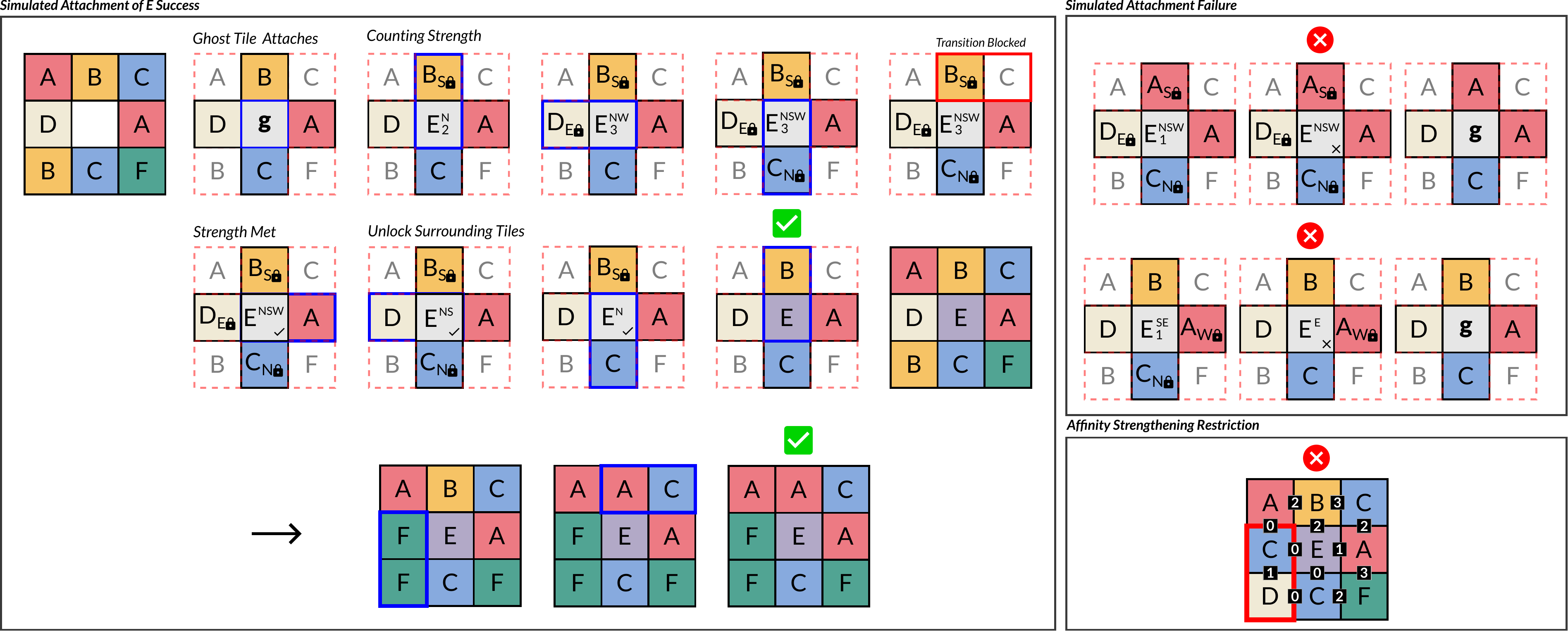}
        \caption{The construction process that the Tile Automata in Figure~\ref{fig:sim_system} builds, representing the same attachments and transitions as in Figure~\ref{fig:example_system}}
        \label{fig:sim_system_construction}
\end{figure}

\temponesim*
\begin{proof}
The set of states $\Sigma_1$ contains $\mathcal{\tau}$ states for each $\sigma \in \Sigma$. We simulate the system in Figure \ref{fig:example_system} with the one in Figure \ref{fig:sim_system}, which contains the following states:


\begin{itemize}
    \item An unlocked state $\sigma$ for every $\sigma\in\Sigma$.
    \item Locked states ${\sigma_L}_d$ for $d \in \{N, S, E, W\}$ for the directions. The lock $L$ is represented by a lock icon in Figure~\ref{fig:sim_system}.
    \item Counting states $\sigma_{i, Q}$ numbered from $1$ to $\tau - 1$, where $Q$ is all subsets of $\{N, S, E, W\}$.
    \item Success unlocking neighbor states $\sigma_{Q, \checkmark}$.
    \item Failure unlocking neighbor states $\sigma_{Q, \times}$.
    \item An empty state \textbf{g}, which we call a ghost tile.
\end{itemize}

\textbf{Empty States.} 
The state \textbf{g} has affinity with all non-ghost tiles $\sigma$, the states which map to something. This ghost state transitions with unlocked tiles adjacent to it to enter a counting state representing a tile which may attach. This process is outlined in Figure \ref{fig:sim_system_construction}. 
If the strength of the affinity is greater than or equal to the input system temperature, then the counting tile immediately transitions to a success state and starts unlocking its neighbors. If the sum is not yet $\tau$ the neighbor state is transitioned to locked and the counting tile increased based on the new binding strength. 

Additionally, the attachment may nondeterministically choose to fail and begin the unlocking process of all locked surrounding states at any time. This has two functions.
First, if the tile does not have 4 neighbors and it cannot reach the affinity strength, then it would be unable to detect the lack of neighbor on its own.
The second reason is to ensure that the strict definition of simulation can be met. 

\textbf{Simulation.} 
A ghost tile may not attach to another ghost tile or to a tile with a temporary state, nor can a temporary state affect its own affinity strength count. This ensures the system has 1-fuzz. A tile can only transition from a ghost tile to simulate attachment if there exists enough locked neighbors which reach $\tau$ so we know every assembly in $\prodasm{\Gamma_1}$ maps to something in $\prodasm{\Gamma}$. This shows equivalent production.

For following and strongly modeling we note that transitions are simulated in one step so the rules in $\Delta \subset \Delta_1$. For every attachment that could take place in $\Gamma$ we simulate this via the adding states. The failure states serve two purposes, first if we select a tile that does not have enough affinity it will eventually be abandoned and another will be selected. Second it allows us to satisfy the Strongly Models definition as when an assembly $A' \in \prodasm{\Gamma_1}$, which represents $A \in \prodasm{\Gamma}$, can abandon any attachment step to circle back and reach an assembly which represents any $B$ such that $A \to^{\Gamma} B$.
\end{proof}

\subsection{Alternate Upper Bound}\label{subsec:alternate_upper_bound}
The dominating factor of the tile set is the adding tiles. We may replace these by instead of storing the temperature we may store the current neighbors. This is a better bound in the case that $\tau \geq \mathcal{O}(\Sigma^2)$. 

\temponesimtwo*
\begin{proof}
Consider an alternate set of adding states, which store the current neighbors of the ghost tile instead of adding the strength. This encodes $\Pi$ into $\Delta$ directly without adding up $\tau$. We only need to store up to $3$ neighbors as the fourth neighbor will be read in the final transition.
\end{proof}

\newpage

\section{Supertiles}\label{sec:supertile_section}
A supertile is a block of $m$ by $m$ tiles that maps to a single tile in the system it is simulating via the $m$-block representation function.
Each supertile contains the complete rules of the system it is simulating and, hence, can perform attachment and transition operations locally with its neighboring supertiles.
To do this, the supertile contains a lookup table that stores the possible transition rules and affinities for each combination of states and neighbor directions.

Each state of the system to simulate is first mapped to a unary encoding.
The table contains four smaller subtables, one for each neighboring direction.
Each of these subtables is constructed as a matrix.
The column indicates the state this supertile represents, and the row the states the neighbor can represent.
We call an entry in this matrix a \emph{datacell}.
A datacell that is part of the East subtable stores at position $(i, j)$ the affinities and transition rules that apply if this supertile represents state $j$ and the East supertile would represent state $i$.
Lastly, each supertile has an \emph{active} column.
This column indicates which state the supertile currently represents.

Besides the lookup table, a supertile contains \emph{wires} that connect the table to the edges of the supertile and \emph{gadgets} for reading, writing, and locking the table.

\subsection{Agents \& Gadgets}\label{subsec:agents_and_gadgets}
A supertile is comprised of several gadgets, groups of tiles that together perform a specific function, such as facilitating data traversal or table lookup queries.
Agents are small packets of information encoded by tile states that traverse a supertile and can transport information from one part of the system to another.
Figure~\ref{fig:simple_overview_supertile} shows a single agent that has traversed from the supertile neighboring to the south to perform a lookup in the table.
Other tasks specific agents can perform include locking the edges of a supertile or its table, clearing wires, or coordinating construction functions.

Gadgets are groups of tiles that together serve a specific purpose.
They are reset after each use and are, therefore, reusable.
Whereas agents move through the system, gadgets are largely stationary.
Agents can interact with gadgets, and each gadget serves a specific purpose.

\subparagraph*{Wire.} The simplest gadget is the wire.
The only purpose of a wire is to allow the one-way traversal of an agent from one part of the system to another.
A wire is a one-wide string of tiles.
The states of the wire tiles not only indicate it is a wire tile, but also indicate which direction the wire is going.
An agent can traverse a wire by swapping states with a neighboring wire tile if the direction of that wire tile allows it.

Wires connect supertiles and allow them to communicate.
Since a supertile has a wire connected to its neighbor for each state it can be in, the specific wire on which a supertile~$S$ communicates with its neighbor is an implicit communication of the state of~$S$, see Figure~\ref{fig:simple_overview_supertile}.

\subparagraph*{Data Strings.} Data strings are a series of tiles carrying data capable of traveling down wires. Transition related data string consist of a start data string tile, a string of unary 1 tiles, an end data string and on occasion a prepended instruction.

\subparagraph*{Door.} Doors are tiles placed along wires to control the flow of data and construction.
They consist of two parts.
The first part is the actual door, placed on the wire in question.
The other is its \emph{handle}.
When an agent or data string reaches a door, it can pass if the door allows it.

Each door has a specific direction dependent on the wire it is on. If the wire the door is on switches directions, then the door's direction will also flip.
An agent trying to pass an open door swaps states with the door as if it were a normal wire.
The door then enters a \emph{reset} state, indicating it recently let an agent through and is currently not connected to its handle.
No other agent is allowed to pass in this state.
Once the agent has moved on, the door can swap states with the new wire tile on its original spot, resetting the door.
This prevents doors from getting lost or mis-matched to the wrong handle, see Figure~\ref{fig:door_in_action}.

When a door swaps with a wire tile, the wire tile will be transitioned to a blank wire tile and take on the wire state of the next wire tile it encounters. Agents may swap with blank wire tiles but, due to the lack of direction, may also swap back if a normal wire tile state has not been selected. The copy director may also copy into blank wire tiles. This blank wire tile method ensures that wires will not become inappropriately shuffled during traversal. 

Additionally, agents or data strings may occasionally swap with the wire tile a door needs to return to its normal position, blocking it from doing so. In this case the door sends a repelling signal to the offending tile making it traverse with a wire tile backward once, if there is not a wire tile behind it and instead another agent or data string tile the repel state will be passed from tile to tile until one is able to swap. Then, all of the previous tiles still in the repelling state may return to normal by swapping with this wire tile until it reaches the door in need of resetting. 

\begin{figure}[t]
    \centering
    \includegraphics[width=.9\textwidth]{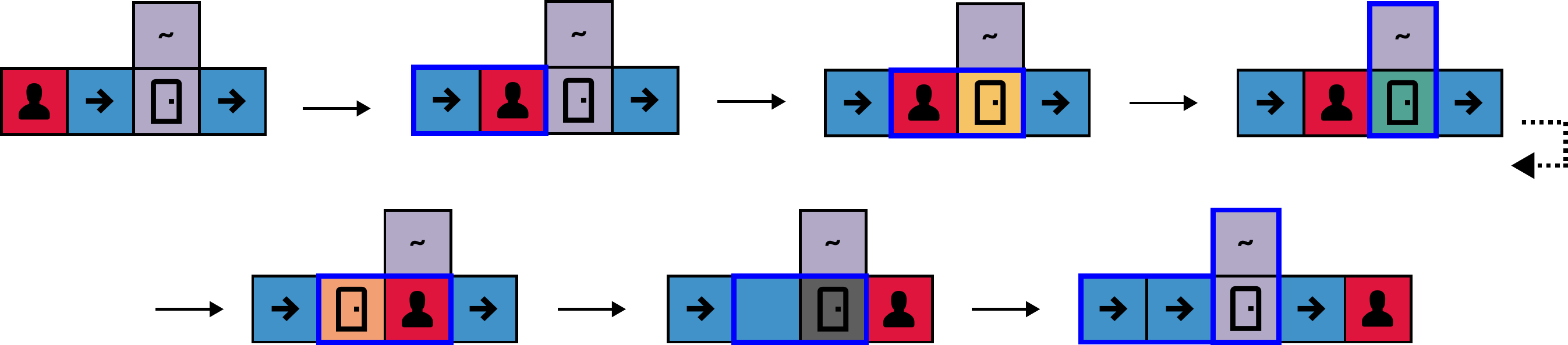}
    \caption{The door in action. Once an agent asks to pass the door, the door first confirms with its handle, after which it goes into an open state. The agent can then pass the door. The door goes into an orange warning state, after which it is only allowed to swap with a wire tile to go back to its original position.}
    \label{fig:door_in_action}
\end{figure}

\subparagraph*{Crossover Gadget.} When wires crossover within the table, at the edge of the supertile, and within construction wires, we require a gadget to control the flow of information across these wires, see Figure~\ref{fig:crossover_gadget}.
The crossover gadget has 3-4 doors arranged with crossover gadget handles at each corner, all around a center wire tile. In most cases, a locking/unlocking agent or the agent's own transition rules ensure it passes through a crossover gadget in the appropriate direction; however, in a few rare instances, the agent will need to signal through the crossover gadget to lock other doors within the crossover before traversing, unlocking them upon exit. 
\begin{figure}[t]
    \centering
        \includegraphics[width=.175\textwidth]{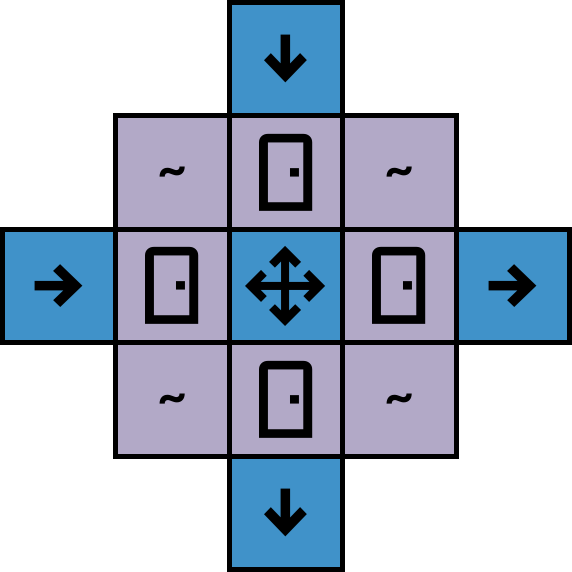}\hfil
        \includegraphics[width=.175\textwidth]{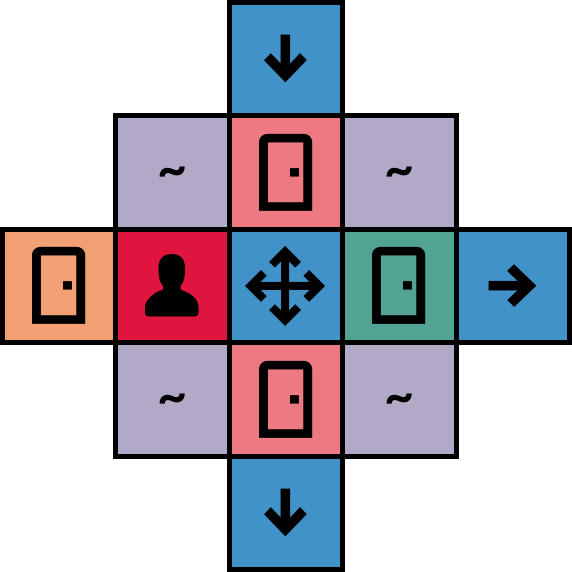}
        \caption{\label{fig:crossover_gadget}Left: Standard crossover gadget. Right: Agent traverses a crossover gadget horizontally. The red doors are locked, and the green door is open.}
\end{figure}

\subparagraph*{Punchdown Gadget.} The punchdown mechanism allows the distance between datacells to be calculated by decrementing a data string representing the necessary number of columns to traverse in the table.
See Figure~\ref{fig:punchdown_gadget} for how decrementing works with the punchdown gadget. 
When a data string comes along a wire, the punchdown gadget will ``delete'' one of the unary digits by transitioning it into a normal wire tile.
The rest of the data string is allowed through the door as normal.
The end-of-string tile resets the punchdown gadget.

\begin{figure}[t]
    \centering
    \includegraphics[width=0.3\textwidth]{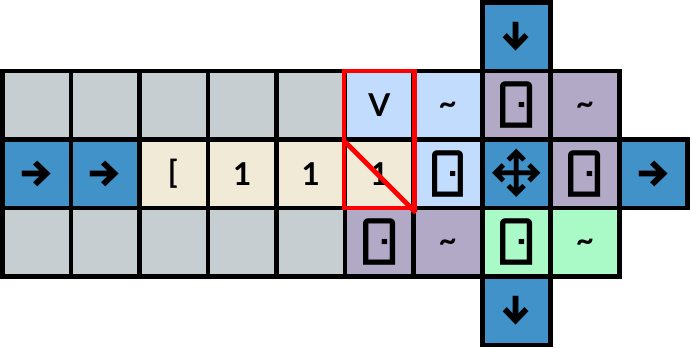}\hfil
    \includegraphics[width=0.3\textwidth]{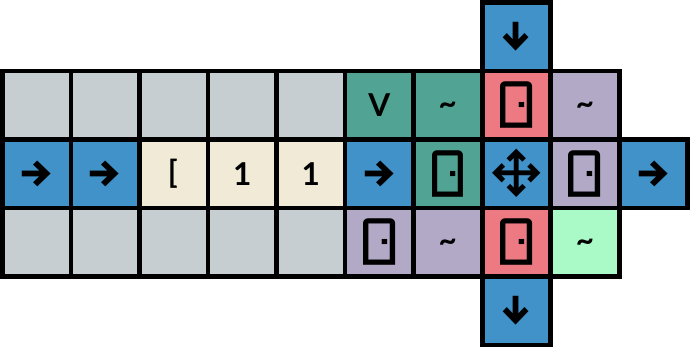}\hfil
    \includegraphics[width=0.3\textwidth]{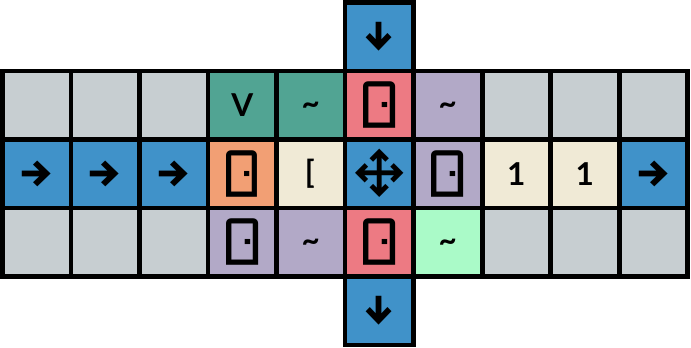}
    \caption{\label{fig:punchdown_gadget} The punchdown gadget first decrements a data string by turning the 1 into a wire tile (Left). Then, the associated door is unlocked, after the north and south crossover doors are locked (Middle). Lastly, the punchdown door resets after the end-of-data string tile swaps with the door (Right).}
\end{figure}

\subparagraph*{Transition Selection Gadget.} On each of the four edges of the supertile, directly next to the wires, lies a Transition Selection Gadget, see Figure~\ref{fig:transition_selection_gadget}.
Upon initiation of a transition with a neighboring supertile, this gadget and its mirror on the neighboring side together determine which transition to take.
The gadget is filled with the transition rules when a transition is initiated.
Once this is done, each supertile initiates a nondeterministic selection agent to walk up and down the border between the two selection gadgets.
When they meet at a rule, they may or may not transition with one another to select that transition rule to be executed on the two supertiles.
The newly selected states are then returned to the tables of the two supertiles to update their respective active columns.
The complete workings of supertile state transitions are explained in Section~\ref{sec:transitioning_tiles}.

\begin{figure}[t]
    \centering
    \includegraphics[width=.4\textwidth]{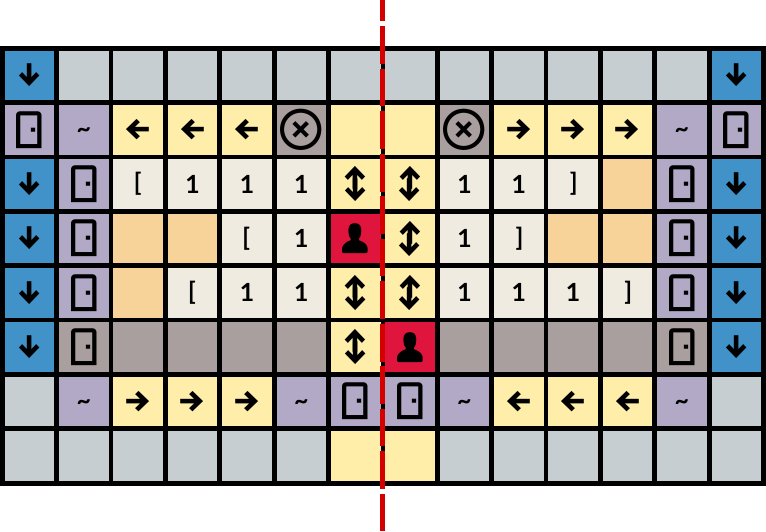}
    \caption{The transition selection gadgets of two neighboring supertiles. The border between supertiles is depicted in red. The agents non-deterministically walk up and down and can eventually select a transition or abort by transitioning with each other.}
    \label{fig:transition_selection_gadget}
\end{figure}

\subsection{Table \& Wiring}\label{subsec:table_and_wiring}

\subparagraph*{Table.} The table is the primary gadget of the supertile.
It stores the affinities and transitions of all possible pairs of states in the system and the functional state of the supertile.

The rows are ordered to prevent the crossing of wires at the border of supertiles.
Hence, the rows are ordered in reverse for the North and West subtable, as shown in Figure~\ref{fig:simple_overview_supertile}.
The columns are ordered normally.

A datacell is a single cell in a subtable and stores both the affinity and transition information for a specific neighbor direction, see Figure~\ref{fig:macrocell_overview}.
It is a compound gadget comprised of an incoming wire, an outgoing wire, wire traversal doors, a punchdown gadget, an affinity door, and the transition storage area.
It can interact with data strings, and agents can use it to get information about the state of the supertile.

Since we use temperature 1, affinities stored in a datacell are simply a Boolean; a combination of two states either has an affinity in this orientation or not.
Each column in the table has a vertical wire next to each datacell.
On this wire, for each datacell, this door indicates if this specific datacell has an affinity or not.
If there is no affinity, the door is a normal door.
Otherwise, the door is a special affinity door, indicating the affinity.

Transition rules are stored in a transition rule storage compartment at the bottom of each datacell.
They are stored as follows.
If the combination of the supertile state and the neighboring supertile state corresponding to this datacell has a transition rule, the storage contains only the new state this supertile would become if this transition is taken.
If the system we are simulating allows for multiple transitions for this pair of states, we define a fixed order for these transitions prior to the simulation.
The resulting states are then stored in the transition storage area according to this order.
The rules always match up because we predefined this order and because the supertile template is copied every time.
The storage is templated to be the size that is necessary to store the maximum number of transition rules of any state pair in the system, such that all transition storage compartments have the exact same size.
If a compartment contains fewer transitions than the maximum, it is filled with blank transitions.

\subparagraph*{Table Locking.}\label{sec:table_locking} The wires enter the table from the left.
Each wire has a special table locking door and corresponding handle, that is situated at the edge of the table.
If an agent tries to act on the table, it first must pass its corresponding door.
If that door is open, two locking agents move up and down along the left edge of the table.
These signals lock all other doors corresponding to incoming wires, allowing for only one operation on the table at a time.
If two of these locking agents meet, only one is allowed to continue on while the other disappears.
Once the locking agent reaches either the top or bottom of the left edge, it transforms into a successful locking agent and moves back toward its door.
A door that sent out locking agents and that sees another locking agent coming by, knows that its own locking agent was the one that failed, and the door will lock itself.
Even if multiple doors try to lock the table at the same time, only a single door will receive its corresponding successful locking agents coming back from both the top as well as the bottom.
This door then allows the waiting agent to enter the table.
An agent trying to enter a locked table will simply wait.
This does not lead to deadlocks, since the origin table of this agent was not locked when this agent left it.

\begin{figure}[t]
    \centering
    \includegraphics[width=.8\textwidth]{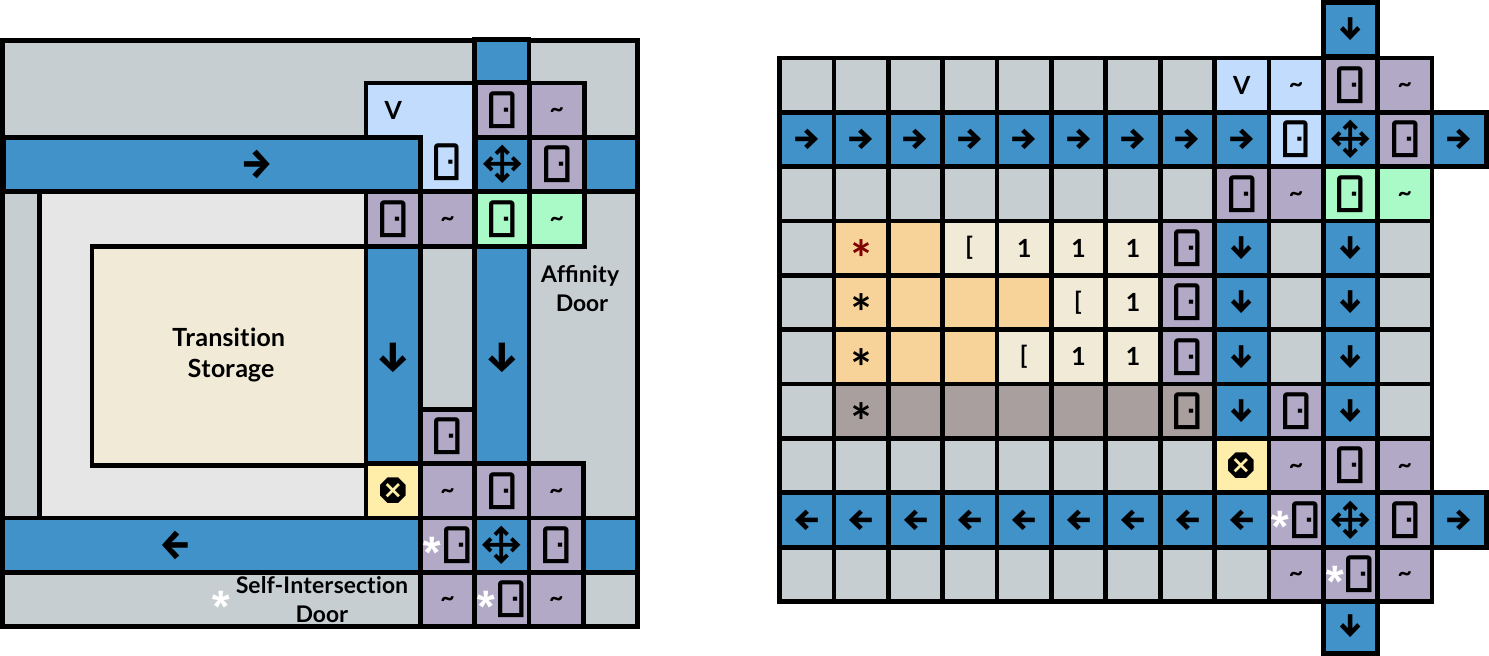}
    \caption{An overview of a datacell}
    \label{fig:macrocell_overview}
\end{figure}

\subsection{Outer Shell}\label{subsec:outer_shell}
The outer frame of both the table and the supertile are comprised of outer frame tiles, with the appropriate doors along wires to access the supertile or table, see Figure~\ref{fig:simple_overview_supertile}.
Inside, they contain an inner wire for traversing the boundary, and then inner frame tiles, again with the appropriate doors.
Initially, supertiles are built with doors that indicate that no neighbor is present.
Once a neighbor is found, they transition to their regular counterparts.

\section{Attachment}\label{sec:attachment}
The attachment process consists of several phases.
First, if the supertile detects it has a spot next to it without a supertile build, it checks if there is an affinity in that direction.
Next, the supertile is copied piece by piece into the neighboring spot.
Lastly, when the construction is completed, this newly built ghost tile sends a signal to all neighboring tiles to select an actual state for itself.

\subsection{Initiation}\label{subsec:initiation}
Whenever a supertile changes state, either via transition or via attachment, it sends out state transmission agents to all four directions.
When such an agent reaches an edge of its supertile and finds no neighbor, the process of determining whether to attach a new supertile starts, see Figure~\ref{fig:agents_search_edges_overview_attachment}.
The state transmission agent is not able to reach the true edge of the supertile if there is no neighbor present.
Instead, it reaches the inner row of doors on the edge of the supertile that controls access to the outline wire.
Since there is no neighbor, the agent changes state to a lookup state, enters the outline wire, and goes to the affinity selection wire, see Figure~\ref{fig:agents_search_edges_overview_attachment}.
Initially, all outline wires on the edges are directed such that if no neighbor is found the state transmission agent will be directed to the affinity selection wire.

The affinity selection wire wraps around the outside of the entire subtable for its given direction (E/N/W/S) allowing for the agent to drop into the active state column wire and search every possible state in the system for an attachment. For example, if the missing tile is to the West of the supertile then the wire runs above West and its outgoing wire is below West 1.

\begin{figure}[t]
    \centering
    \includegraphics[width=.9\textwidth]{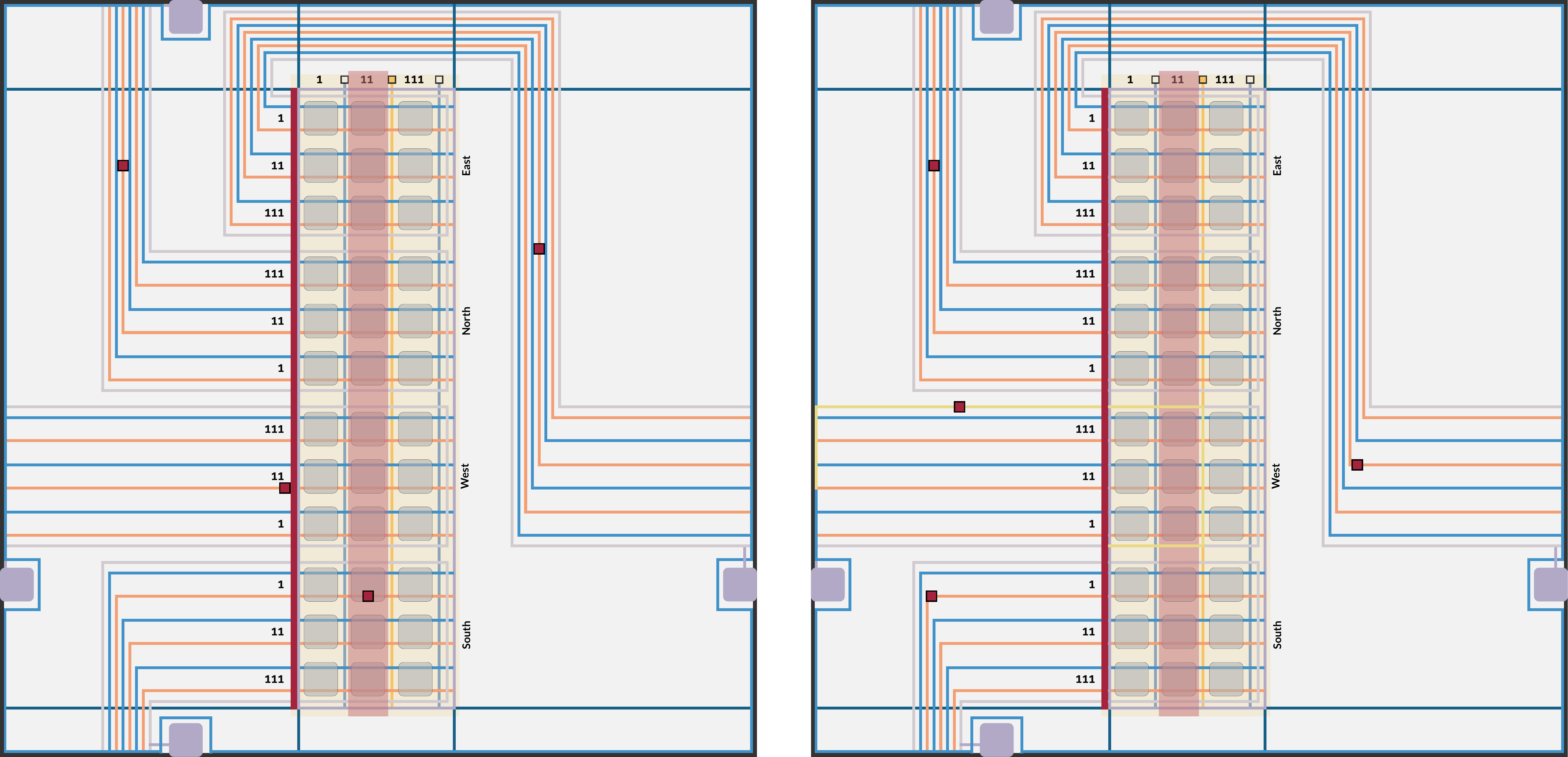}
    \caption{Left: The tile transmits its state to neighboring tiles and discovers it has no neighbors. Right: The west agent has moved over to the affinity selection wire (in yellow).}
    \label{fig:agents_search_edges_overview_attachment}
\end{figure}

\begin{figure}[t]
    \centering
    \includegraphics[width=.9\textwidth]{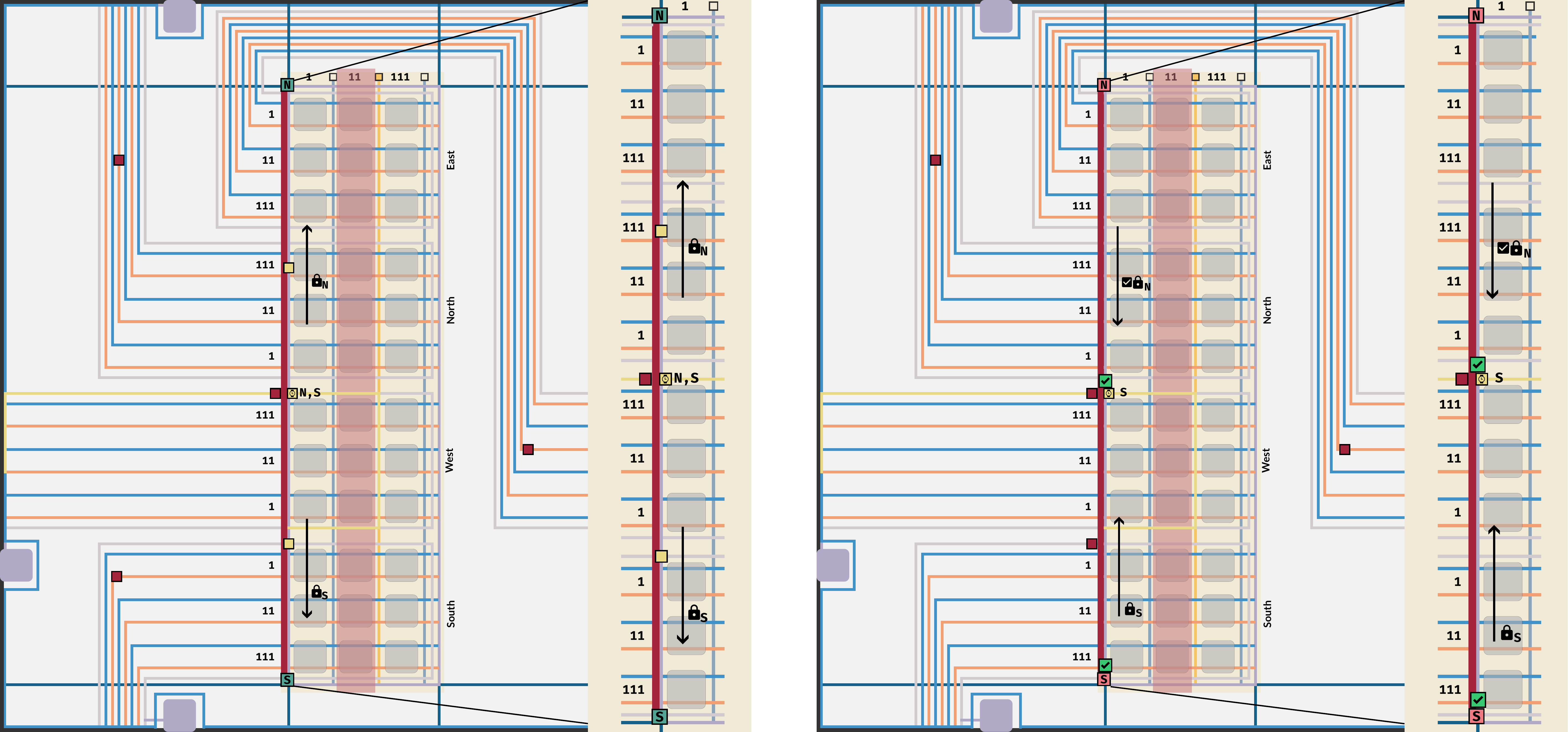}
    \caption{The lookup agent reaches table and locks the table}
    \label{fig:agent_lock_table_overview_attachment}
\end{figure}

\subsection{Checking Attachment}\label{subsec:checking_attachment}
Each datacell has an affinity door in the vertical active column wire next to it, indicating that there is an affinity between these two states in this direction.
If there is no affinity, there is a door with a no affinity state.
The initially selected possible state may not be the final state of the supertile so as not to preclude every possible state that the supertile may end up in.
This initial lookup is just to ensure we do not begin construction of a neighboring supertile if no tile can attach there in the system we are simulating.

\paragraph*{Lookup in Table Section.}
Once the signal reaches the edge of the table it will initiate the standard table locking process described in Section~\ref{sec:table_locking}, and visible in Figure~\ref{fig:agent_lock_table_overview_attachment}. The signal will then traverse the table until reaching the active state column and begins its descent down the active state column wire with the affinity doors; see figure \ref{fig:intersection_attachment_found_overview_attachment}.

The lookup agent may nondeterministically transition with any affinity door to a \emph{found} state so that the construction process can begin. It will traverse down to the state lookup exit wire and to the edge of the supertile.
The table stays locked.
Upon reaching the edge of the supertile, the agent transitions into the Copy Checkpoint.
The Copy Checkpoint is a stationary tile on the edge of the supertile, orchestrating the copy process.
See Figure~\ref{fig:intersection_attachment_found_overview_attachment} for the state lookup process at a high level.
If there are no affinities in that direction or the lookup agent never transitions with an affinity door, the agent will not change to the found state and simply unlocks the table after exiting from the bottom of the state lookup wire. No attachment is started in that case; see Figure~\ref{fig:no_attachment_overview_attachment}.

\begin{figure}[t]
    \centering
    \includegraphics[width=.9\textwidth]{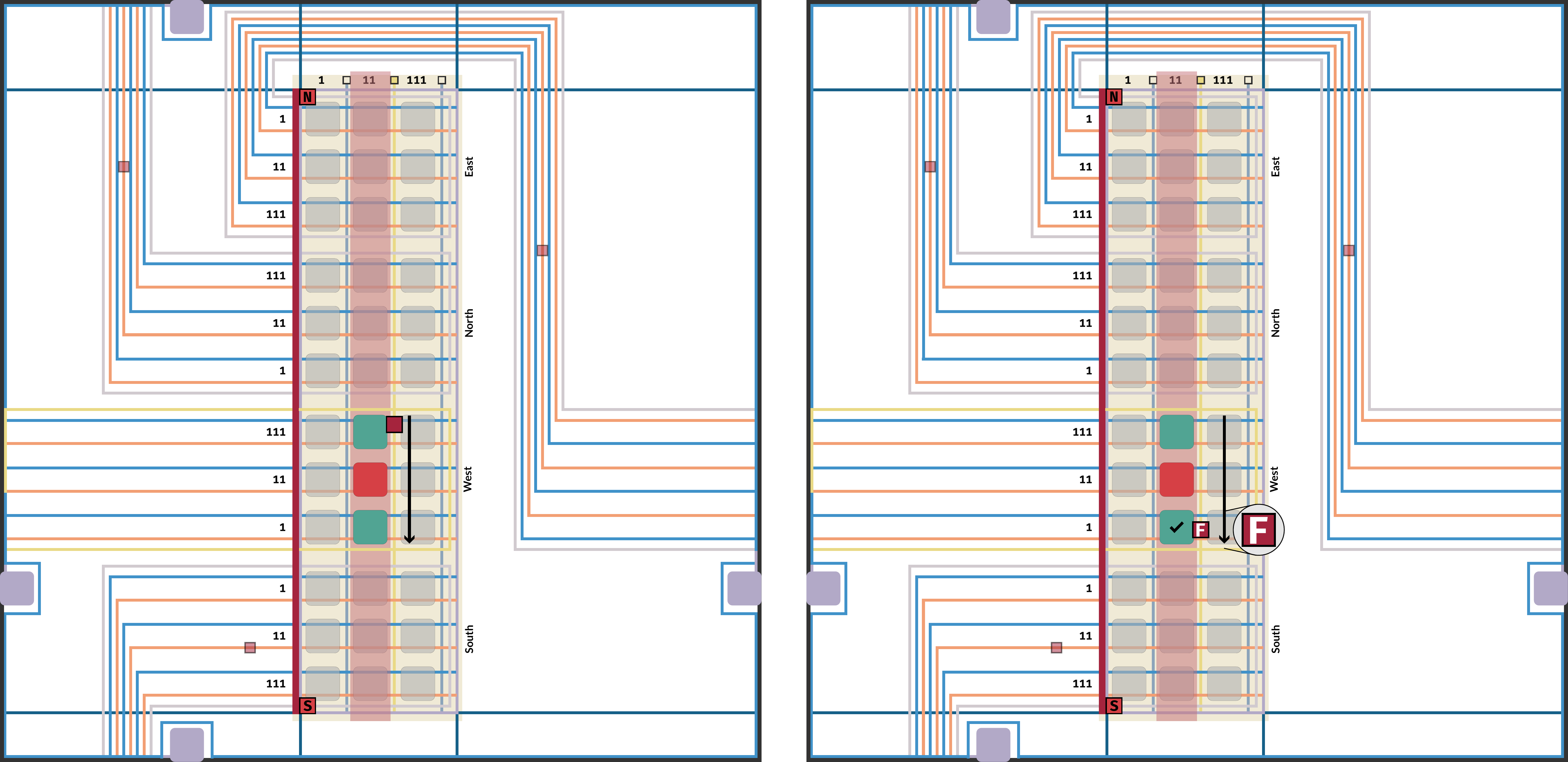}
    \caption{The lookup agent reaches the intersection with the active state and confirms there is an attachment possible for that direction}
    \label{fig:intersection_attachment_found_overview_attachment}
\end{figure}

\begin{figure}[t]
    \centering
    \includegraphics[width=.9\textwidth]{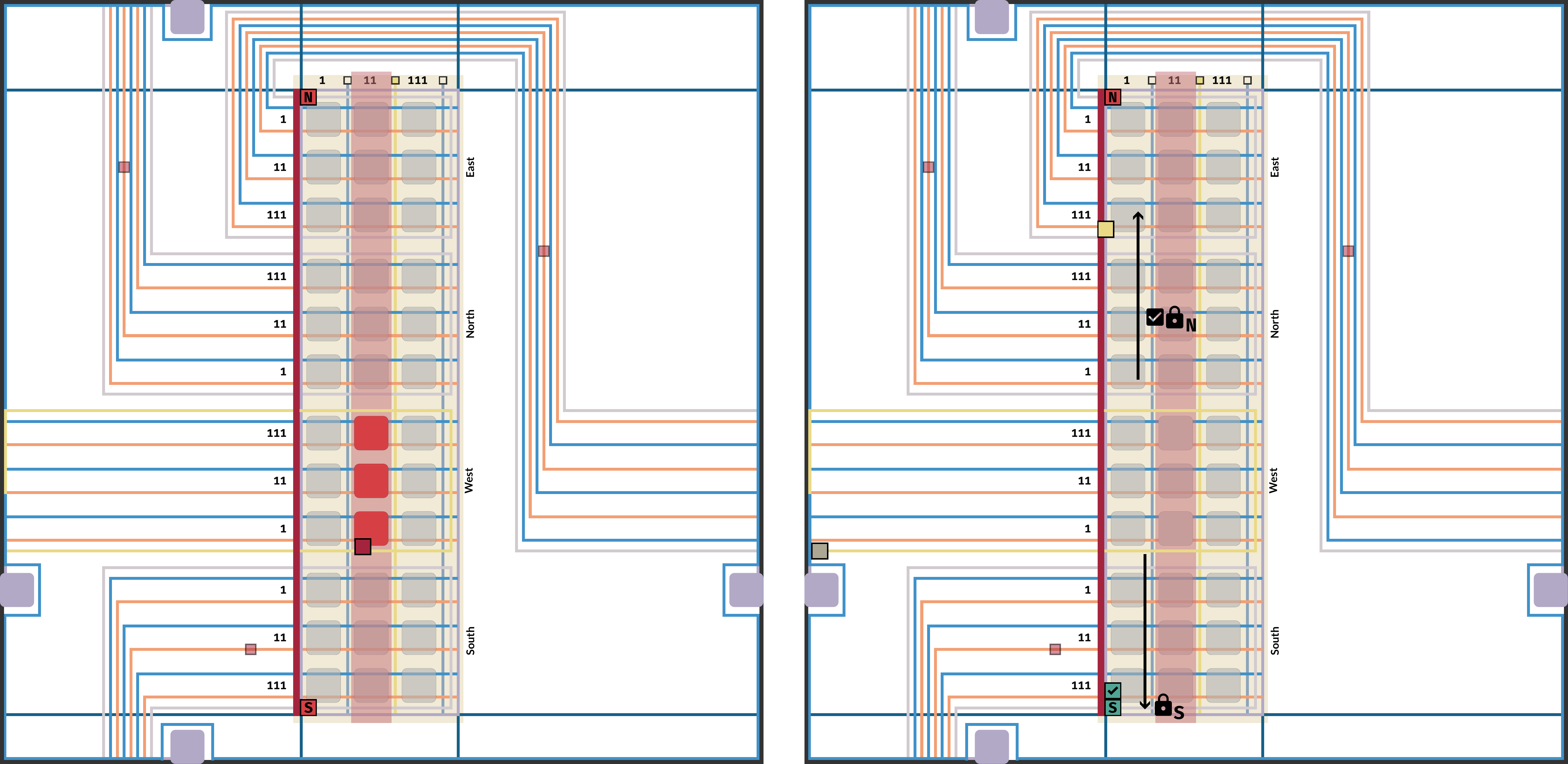}
    \caption{The lookup agent finds no attachment and unlocks the table, deleting itself when it reaches the edge of the supertile.}
    \label{fig:no_attachment_overview_attachment}
\end{figure}

We only begin construction if there is a possible attachment in that particular direction.
This does give a slight fuzz advantage over previous builds. A ghost tile (if simulating a system that requires or includes them) will not begin growing if the state has no affinities in that direction. 

\subsection{Preparing for Copying}\label{subsec:preparing_for_copying}
In order to properly copy the tile, we need to clear the interior wires, lock any further outside communication from coming in, and activate a number of processes.

\paragraph*{Locking the Supertiles Outer Frame for Construction.}
First, the Copy Checkpoint sends a locking agent around the outer supertile wire, which locks every outer door on the edges of the supertile, preventing any further agents from entering the supertile. 

\paragraph*{Clearing Table and Wires.}
Once the locking agent returns to the Copy Checkpoint, it will turn into a wiping agent. The main wiping agent will traverse the edge wire and spawn minor wiping agents to sweep every wire. They delete any waiting agents outside of the table. Moreover, the interior and edge of the table are reset to be entirely inactive. This process is shown in figure \ref{fig:lock_and_wipe_supertile_attachment_ov}. Should an agent trigger a locking process before it is wiped, the locking agents spawned will transition into wire tiles upon contact with an inactive door. The active state of the current supertile is stored within a containment area at the top of the active column to ensure it is spared from wiping. 

\begin{figure}[t]
    \centering
    \includegraphics[width=.9\textwidth]{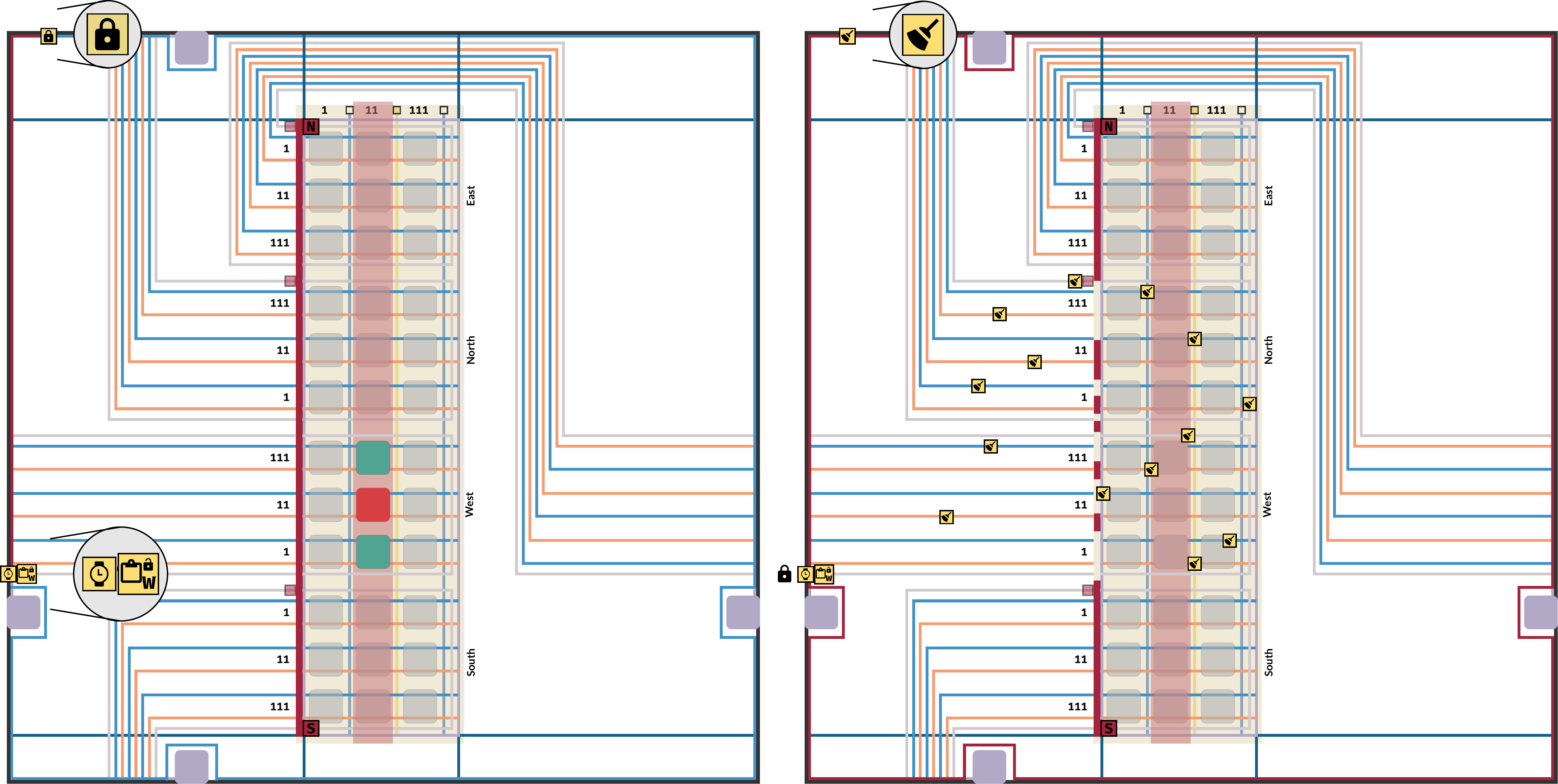}
    \caption{Copy Checkpoint (West) begins construction by locking then resetting/wiping the supertiles interior.}
    \label{fig:lock_and_wipe_supertile_attachment_ov}
\end{figure}

\subsection{Copying Supertile Outline}\label{subsec:copying_supertile_outline}
Now, the supertile is ready to start the copy process.
We will copy the supertile in several steps, piece by piece.
We start with the outline, then the table, and then the wires.
It could be that multiple adjacent supertiles are trying to build in the same location.
In this case, it is necessary to ensure that only a single supertile gains construction jurisdiction over this spot.

\paragraph*{Claiming Mirror Side.}
The mirror edge is the edge of the tile under construction that is immediately adjacent to the supertile initiating construction.
Once the wiping agent reaches the Copy Checkpoint, the Copy Checkpoint will spawn 2 claiming agents sent along the outer wire to the adjacent corners.
If building to the West, these claiming agents go to the North and South.
The purpose of these agents is to place and claim the mirror edge corners.
These are important to prevent construction conflicts.
Once the agents reach the corners, the northwest and southwest corners in our example, as seen in Figure~\ref{fig:placing_mirror_edge_attachment}, they open the supertile corner doors and try to build the corner of the neighboring tile.

They build the corner of the neighboring tile as follows.
The agent first goes through the door in the direction it wants to build.
Then, it attaches the first empty construction tile and then a second, transitioning the first into a door as it swaps into the second.
Whenever we say that a certain tile is attached or built, we mean that an empty construction tile is attached, which is then transitioned into the correct state.
After the wire tile is built, the agent swaps with the wire tile, and tries to build the crossover gadget.
It marks the crossover gadget as claimed by the East side where it came from.

It could also be the case that there already exists a corner crossover gadget.
This gadget can then either be claimed, or unclaimed.
If it is unclaimed, it is claimed, otherwise, the agent goes back to the Copy Checkpoint to report a failure.
If the crossover corner is successfully claimed, the agent also goes back, but this time to report a success.
If the Copy Checkpoint receives at least one failure, it relinquishes its claim to any successful corners (via agent) and aborts construction.

To abort a construction at this point, the only thing necessary is to unlock the doors on the outline wire, set the active state column in the table, and finally unlock the table.
These processes are the same as their counterparts at the start of construction.

\begin{figure}[t]
    \centering
    \includegraphics[width=.9\textwidth]{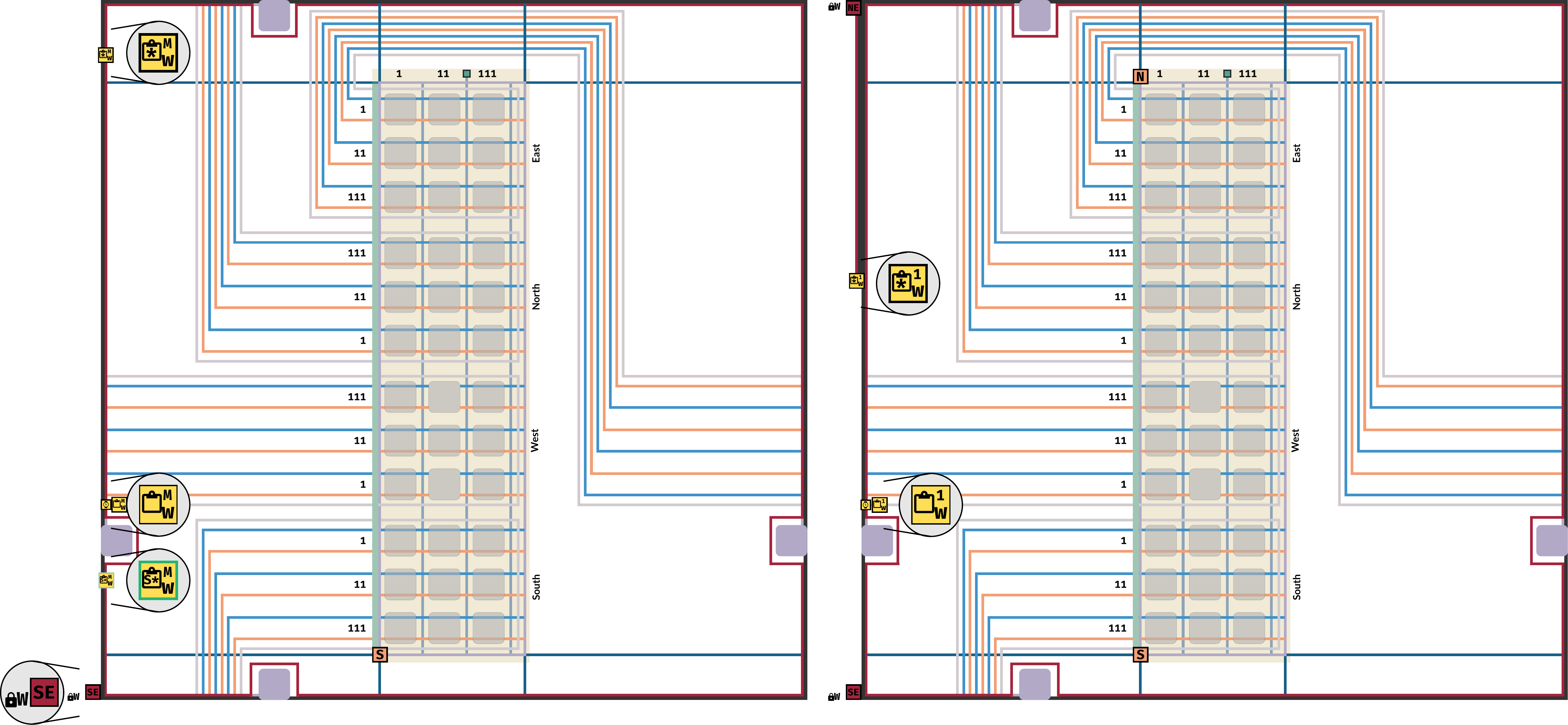}
    \caption{The Copy Checkpoint sends 2 claiming agents to claim the mirror edge. The Copy Checkpoint then sends a mirror edge agent to place the mirror edge.}
    \label{fig:placing_mirror_edge_attachment}
\end{figure}

When the Copy Checkpoint receives two successes, it knows that from the three other supertiles that could potentially try to build in this spot, only one might still be trying.
In this case, we initiate construction of the full mirror edge.
This is done as follows.
One mirror edge agent is sent to one of the corners that has just been claimed.
This agent will then build the outline wire of the mirror edge.
At the same time, it will attach (via a blank construction tile) \emph{mirror} tiles.
These tiles then transition with the tiles in the original supertile to mirror their state. Should the corner it is attempting to mirror be a crossover gadget door, the mirror tile will spawn a crossover copy agent, which when the subordinate copy agent recognizes this will instead back up, transitioning the tiles it traverses back from into blank construction tiles. When the crossover is activated, it will continue. 
In this way, the doors get placed in the correct positions.
Moreover, the mirror edge agent ensures intersections are properly constructed and construction doors are activated, see Figure~\ref{fig:placing_mirror_edge_attachment}.
Once it has reached the other corner, it goes back to the Copy Checkpoint, which can then initiate the next phase.

\paragraph*{Copying Placement General Notes.}
Before explaining the next phase, we will first detail the standard copy procedure.
This procedure is used to copy the rest of the supertile.
It uses a Copy Director, which acts in the original supertile and sends copies of tiles to the Placement Director, which is located in the newly constructed tile, and places the copies on the correct locations.
The copies are send over the outline wires and/or construction wires of the tiles.
For this to work, the route from the Copy Director to the Placement Director needs to be clear and doors along this path need to be set correctly to ensure copies of tiles end up in the correct spot.

The setup is done by a copy agent send out from the Copy Checkpoint.
It first places the Placement Director in the appropriate spot.
Then, it goes to place the Copy Director.
It takes the same path as the copied tiles will take.
While going over this wire, it ensures all wire tiles are pointing in the correct direction and doors that lead in the wrong direction are closed.

Once at the correct spot, it transitions into the Copy Director and starts copying tiles and send them to the Placement director via the path it just created.
As soon as all tiles of this part are copied, it goes to the Placement Director, deletes it, and finally returns back to the Copy Checkpoint, which can then start the next phase of copying.

\paragraph*{The Copy Process.}
For each tile that needs copying, the Copy Director follows the following scheme, visualized in Figure~\ref{fig:copy_director_process}.
First, the Copy Director sends a direction to the Placement Director (North/South/East/West).
Then, it swaps with the tile that needs copying.
This tile then spawns a copy of itself on the wire that also goes to the Placement Director.
Lastly, the Copy Director swaps back with the tile that now has copied itself.

The Placement Director ensures it is always at the end of the part that is built.
It first receives a direction. It then swaps with that direction tile which attaches an empty construction tile.
Then, when the copy arrives, the Placement Director swaps with the copy, and the copy can transfer its state to the newly attached empty construction tile.
Then, the Placement Director swaps back with the now copied tile and deletes it in the process.
This process is shown in Figure~\ref{fig:placement_director_process}.

Not every tile is copied over individually, to reduce the number of states, we copy the crossover gadget in one go. 
Instead of sending a copy of every tile in the crossover gadget, we send a single tile containing a template of full information of the crossover gadget. To stop the directionality from being an issue our copy director will send a second special direction tile before a crossover gadget. This way the agent may be in the middle of the gadget attaching blank tiles and transitioning the surrounding doors into them without knowing the direction from which the crossover came. Each door remain in waiting state until it has attached its handle. The placement director will lock any necessary doors when construction is complete. 

\begin{figure}[t]
    \centering
    \includegraphics[width=.9\textwidth]{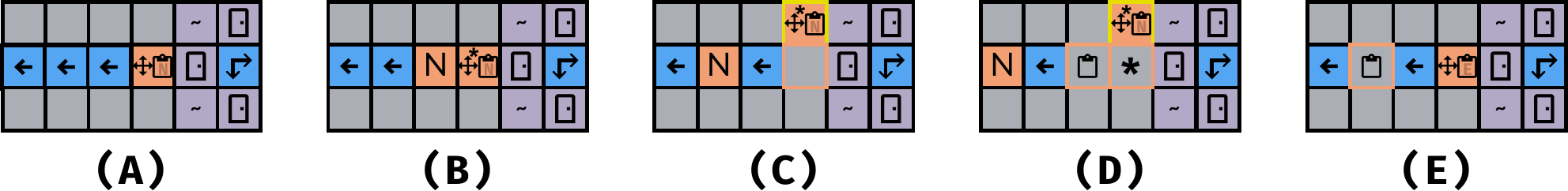}
    \caption{The general copy process.}
    \label{fig:copy_director_process}
\end{figure}

\begin{figure}[t]
    \centering
    \includegraphics[width=.8\textwidth]{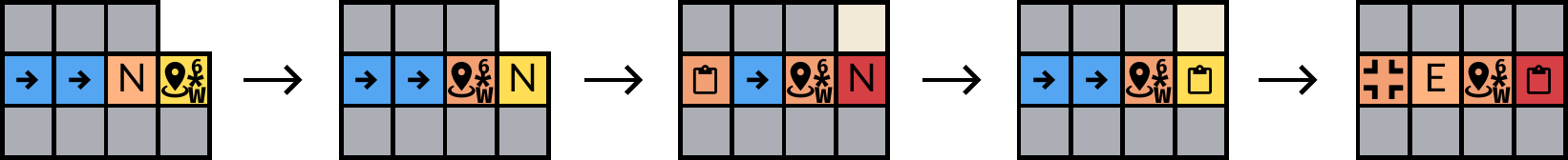}
    \caption{The placement process of a border tile up to the arrival of a crossover construction agent just before the previous border tile is deleted.}
    \label{fig:placement_director_process}
\end{figure}

\paragraph*{Constructing Adjacent Supertile Outline Wires.}
This copy process is used to build the other edges of the supertile, see Figure~\ref{fig:copying_adjacent_edges}.
At this point, only the mirror edge has been constructed.
We use the copy process to build one edge at a time.
For horizontal attachments, we first build the top edge, then the bottom edge.
For vertical attachments, we first build the left edge, then the right edge.
This is to ensure that if there is still another supertile that is trying to build in the same spot, we recognize this situation and deal with it accordingly.

To build these edges of the supertile, the Copy Checkpoint spawns a new copy agent.
This agent will put a Placement Director at the corner that was already built, then moves over the outline wire of the original supertile to the other side, where it will transfer into a Copy Director.

If there is still another supertile building in this spot, the two Placement Directors will eventually meet.
Nondeterministically, one of the two continues, while the other is removed.
The losing Placement Director sends a signal to its respective Copy Checkpoint, which then starts the abortion process similar as before.
The Placement Director that is left over will ignore and remove any copied tiles intended for the now removed Placement Director.

Along with the outline wire the transmission selection gadgets are added.
For a tile to start internal construction it must verify it has claimed all 4 corners. Only doors on its side are active all others are blocked.

\begin{figure}[t]
    \centering
    \includegraphics[width=.9\textwidth]{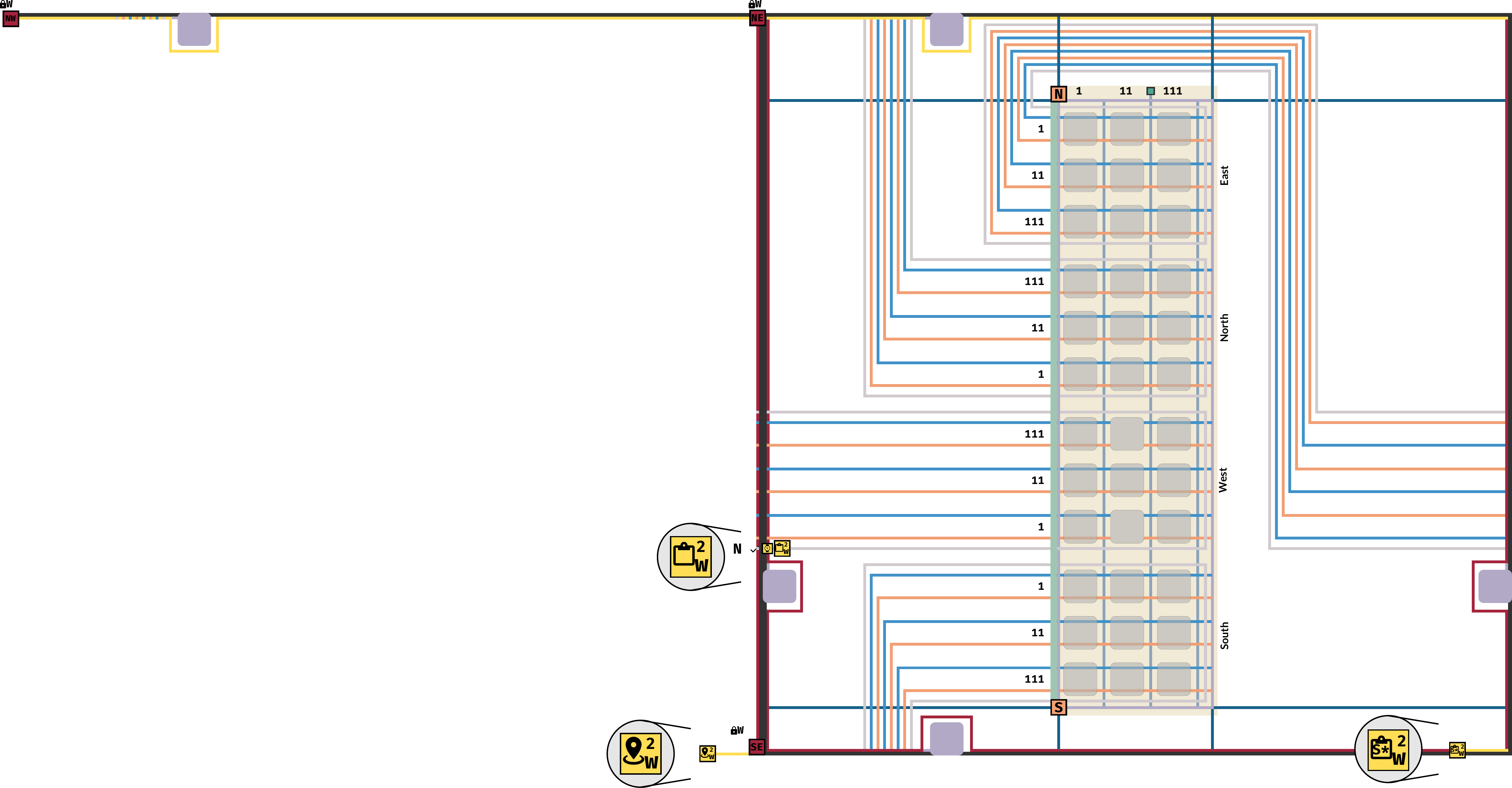}
    \caption{The copy director copies each adjacent edge. It first does the north, then the south edge.}
    \label{fig:copying_adjacent_edges}
\end{figure}

\paragraph*{Building the Far Side of the Outline.} Once we have confirmed claim of all four corners we then can build the opposing edge starting at the designated corner.
This again uses the normal copy process.
The Copy Checkpoint sends out a copy agent, which places a Placement Director.
Then, this copy agent traverses the outline wire to the appropriate place, locking doors on the way if necessary.
Lastly, it will transition into the Copy Director and start copying the last edge of the supertile, see Figure~\ref{fig:copying_farside_edge}.

\begin{figure}[t]
    \centering
    \includegraphics[width=.9\textwidth]{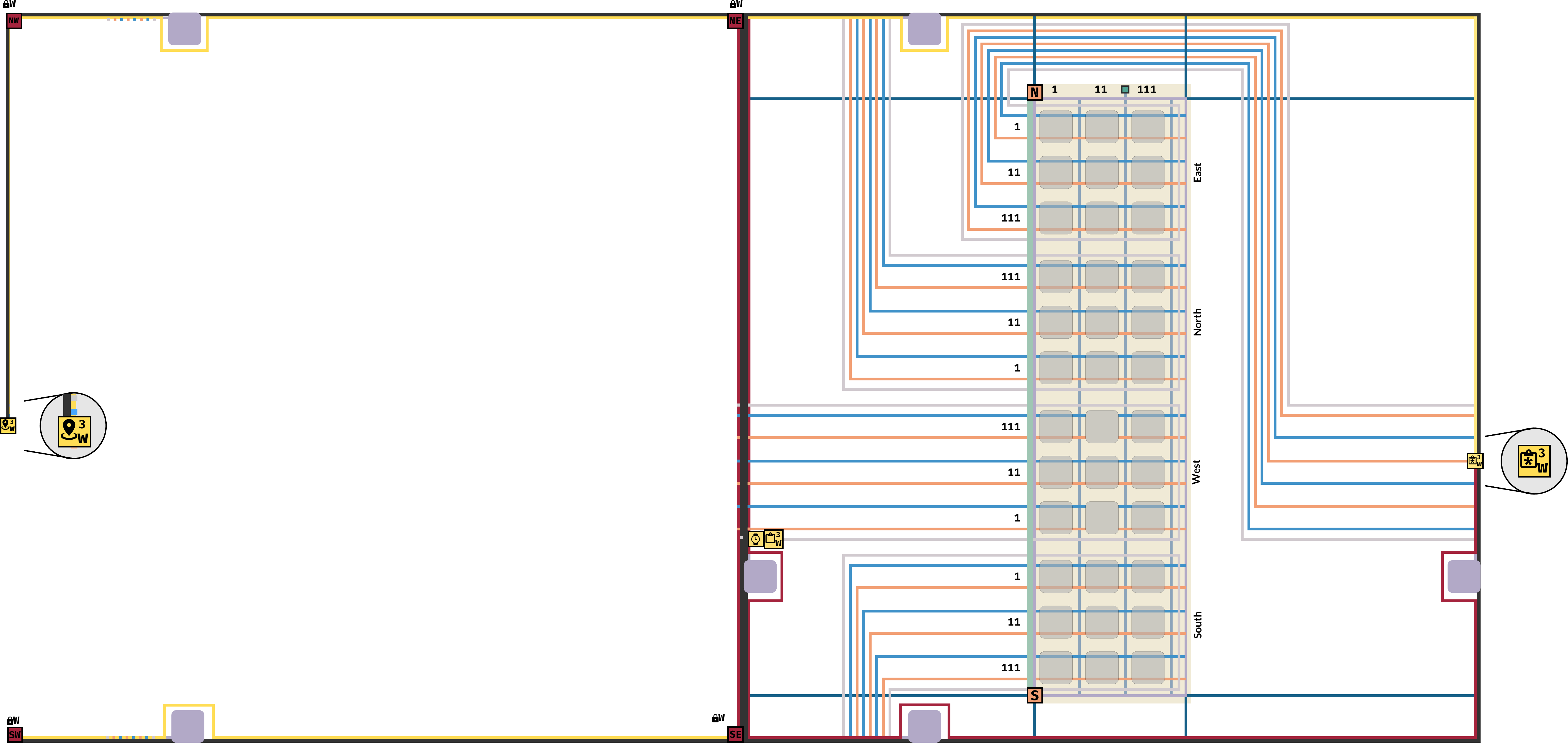}
    \caption{The copy director and placement directors copy the far side edge.}
    \label{fig:copying_farside_edge}
\end{figure}

\subsection{Construction Wires}\label{subsec:construction_wires}
Once construction of the outline of the new supertile is complete and confirmed by an agent at the Copy Checkpoint, we start copying the table.
Each table has eight construction wires.
These wires extend from the corners of the table, and mark the width and height of the table, plus its position within the supertile.

To copy the table, we first copy the construction wires.
For horizontal copying, we start by copying the horizontal construction wires, then the vertical ones.
For vertical copying, we do the opposite.
These wires are not only used to indicate the placement of the table within the supertile, but they will also be used to copy the contents of the table.

To copy them, we first open the doors connecting the respective wires to the outline wire. We then copy the wires using the normal copy process, starting at the far end. The horizontal and vertical process can be seen in figure \ref{fig:copying_horizontal_table_outline} and figure \ref{fig:copying_vertical_table_outline}. 

\subsection{Copying Table}\label{subsec:copying_table}
The construction wires already contain the outer edges of the table, including the table control edge on the eastern side of the table. Once these are in, we build the rest of the table. Importantly, the only variable aspects of a datacell are the width and the height of the transition storage. These depend on the system we are simulating.

\begin{figure}[t]
    \centering
    \includegraphics[width=.9\textwidth]{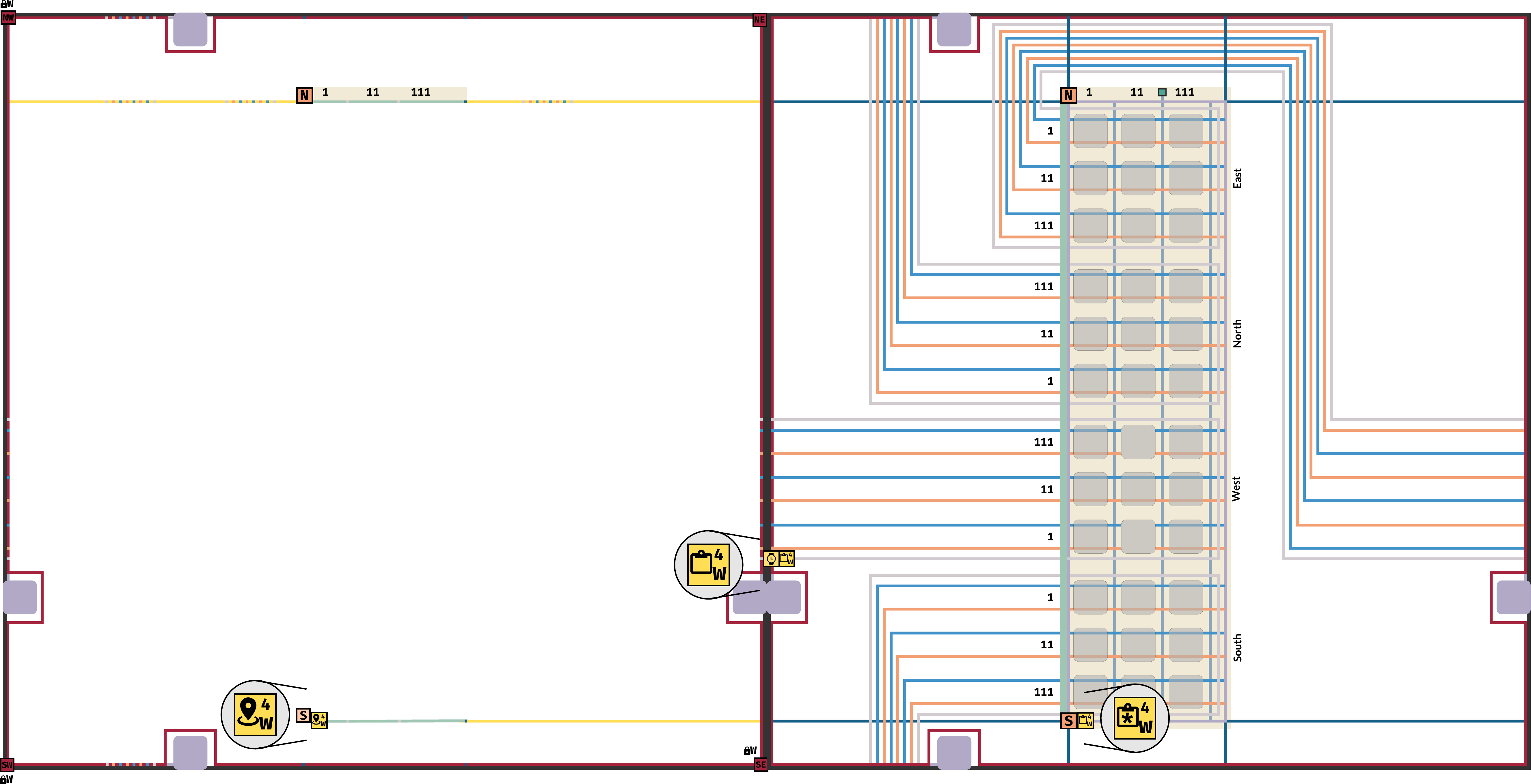}
    \caption{Copying the horizontal table outline.}
    \label{fig:copying_horizontal_table_outline}
\end{figure}

\begin{figure}[t]
    \centering
    \includegraphics[width=.9\textwidth]{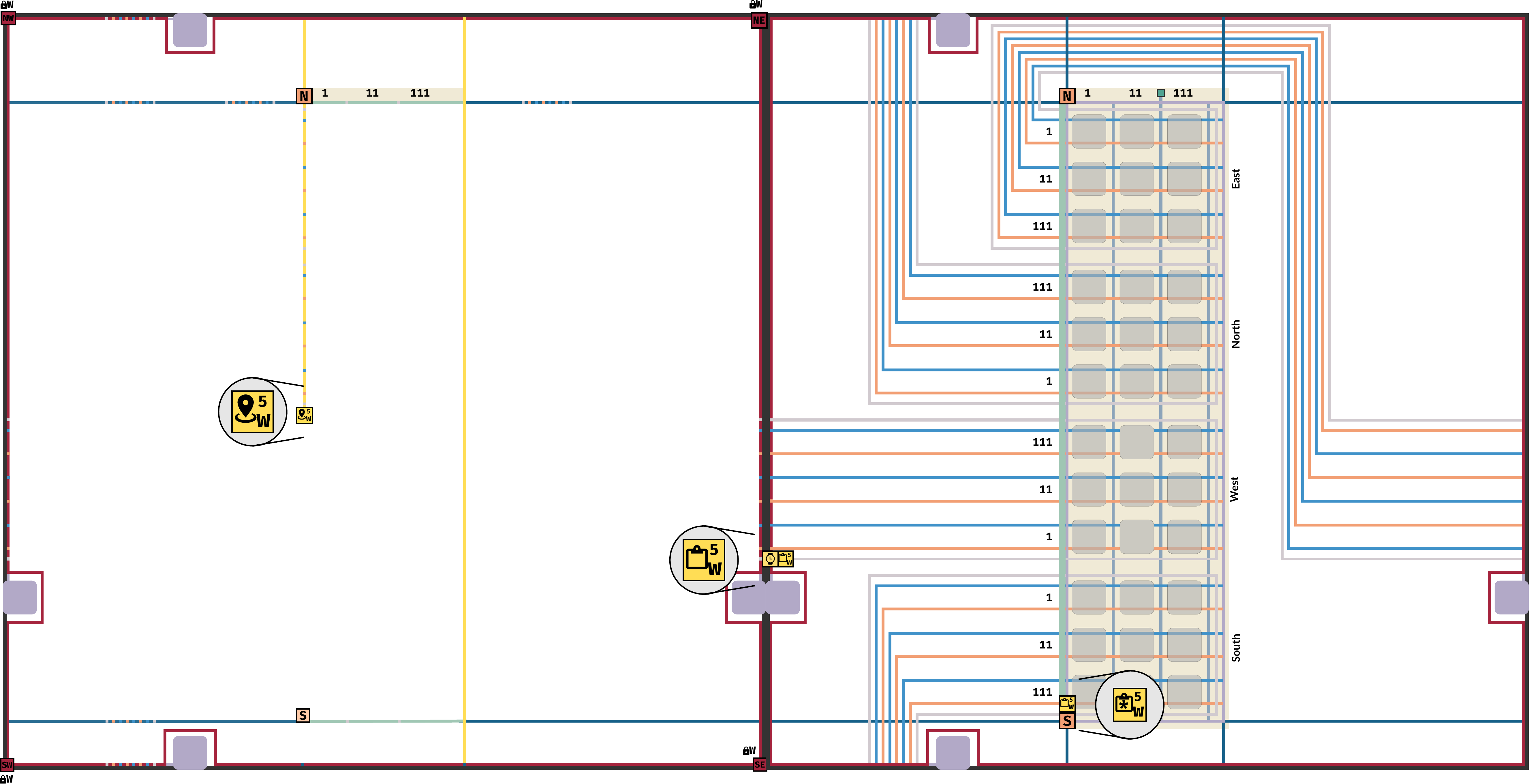}
    \caption{Copying the vertical table outline.}
    \label{fig:copying_vertical_table_outline}
\end{figure}

\paragraph*{Copying Horizontal Table Wires.} 
To transfer this information over to the new supertile, we first copy all the horizontal wires in the table.
These already contain the transition, affinity, and state lookup chute crossover doors.
We copy these over using the normal copy procedure, see Figure~\ref{fig:copying_table_rows}.
Both the Copy and Placement Directors start this procedure on the East side of the table, and the route that the copies take is via the construction wires.

Every time the Copy Director has finished a horizontal wire, it moves down to the next horizontal wire.
For every step that it takes, it sends a token to the Placement Director.

For every one of these tokens that the Placement Director receives, it goes one step forward.
The height of a datacell is implicitly transmitted 
 by the distance between the transition storage door on one wire and the end of transitions tile along the wire just below.

\begin{figure}[t]
    \centering
    \includegraphics[width=.9\textwidth]{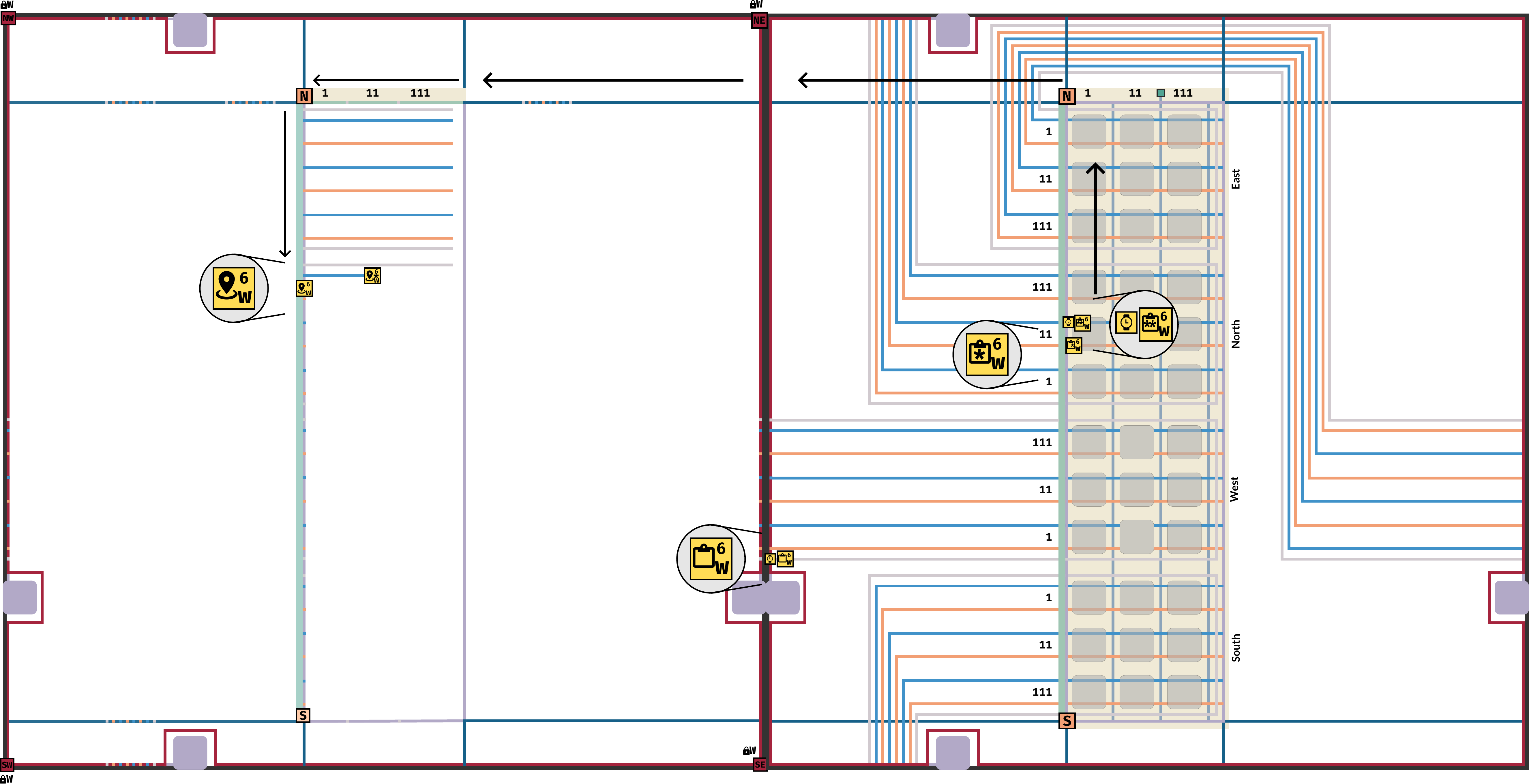}
    \caption{Copying the table row wires.}
    \label{fig:copying_table_rows}
\end{figure}

\paragraph*{Constructing Datacells.}
Next, the placement director will send subordinates down each table input row to construct datacells, reporting back as each is completed until it reaches the end of the table. 
Within the row, when the placement subordinate meets the datacell punchdown door it will traverse through the transition storage door adding doors to the west, wires to the south, and border tiles to the east, until it reaches the end of transitions tile below it. It will add the transition exit door to its east before returning to the top of the datacell, traversing through the affinity crossover to begin constructing the next datacell.  This can be seen at a high level in Figure~\ref{fig:constructing_macrocell_outlines}.

\begin{figure}[t]
    \centering
    \includegraphics[width=.9\textwidth]{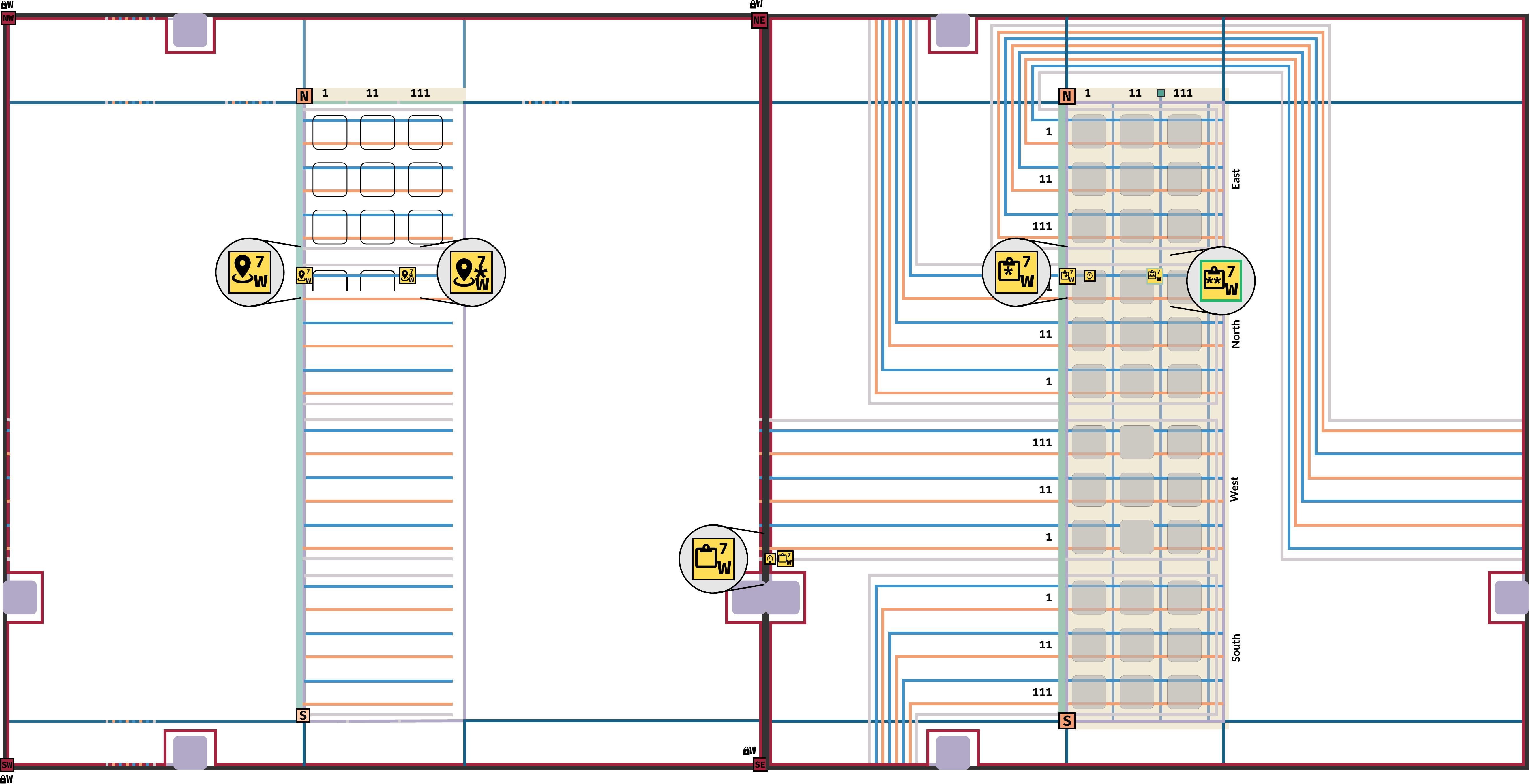}
    \caption{Constructing Datacell outlines.}
    \label{fig:constructing_macrocell_outlines}
\end{figure}

\paragraph*{Copying Transition Rules.} The process of copying datacell transition rules works as follows: Our agent transitions the transition row door into a copy activation state, which will swap with every unary tile, the end of data string state, and border tile state behind to transition each into a copy yourself state. Each will copy themselves onto the wire and then flip backwards until reaching the door and transitioning back into a normal state. The door will stay at the back of the row until the border tile returns, signifying it has copied itself. At this point, the door will traverse to the front of the row and activate the door below it to begin the same process. The door will also copy buffer tiles and includes rows that have no transitions in them and are comprised of filler blocks. 

When this is complete the fill marker is moved to the next datacell. When the row is completed the copy agent will confirm its completion to the copy director indicating that the next row can be started.

On the opposite side the data strings are activated for placement, and due to special placement states this allows them to attach their own blank construction tiles to the east. 

\begin{figure}[t]
    \centering
    \includegraphics[width=.9\textwidth]{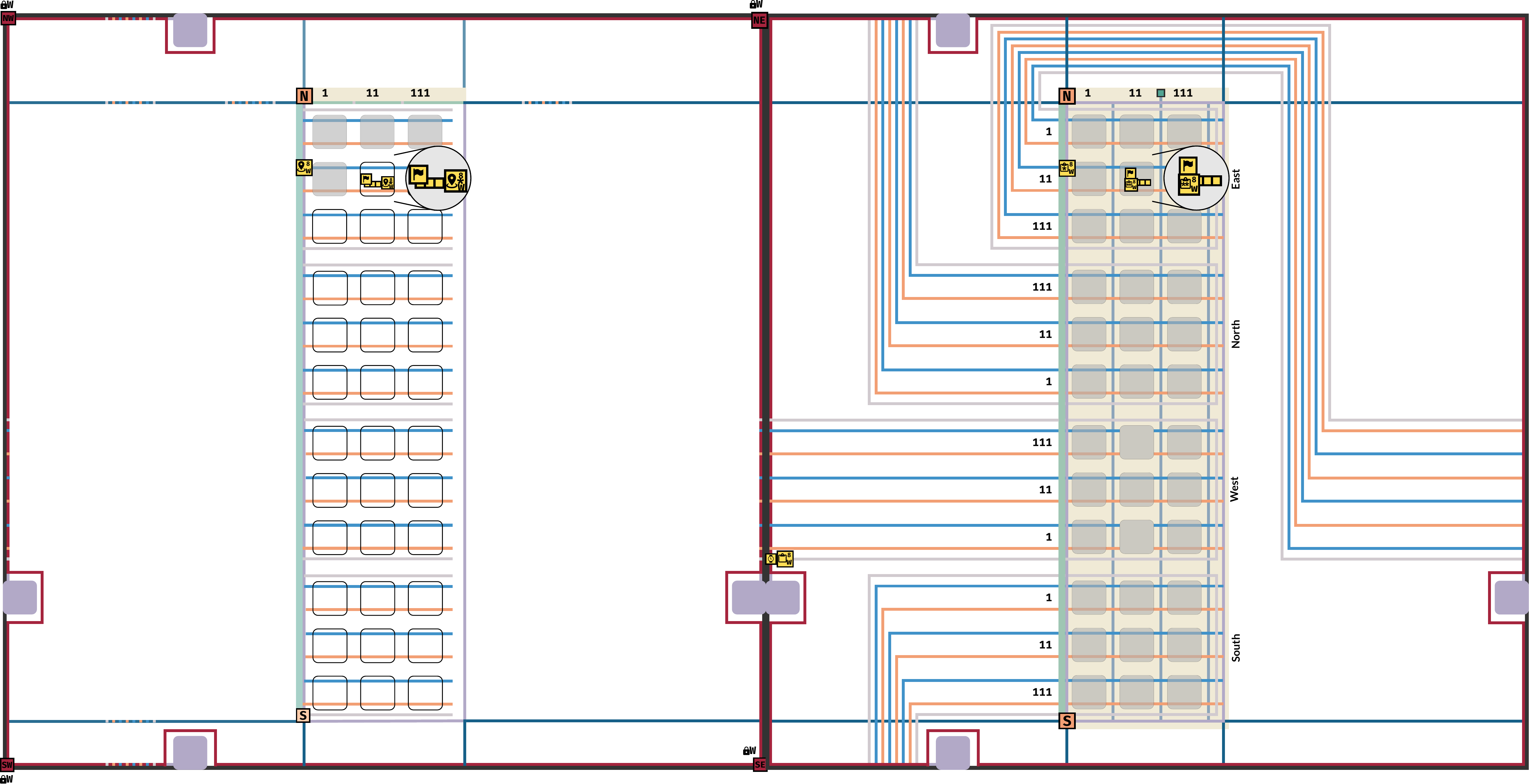}
    \caption{Filling datacells with transition rules.}
    \label{fig:filling_macrocell_transition_rules}
\end{figure}

\paragraph*{Copying Vertical Table Wires.}
As all of the necessary information has already been copied into the table, what is left is to copy the vertical wires, see Figure~\ref{fig:constructing_vertical_table_wires}. The placement director simply sends filling agents down the top of each column that will traverse south and place south wires and protective border tiles to their side as they go. When these agents reach the bottom of the table they reverse walk north until it reaches the top of the table where the placement director will absorb and check it off. As the placement director was at the final vertical wire when it finished sending the agents down it will wait there until the agent returns with confirmation that it has completed filling the southern wires. The placement director will then wait to proceed at each intersection until it has reached the control edge of the table. Once this occurs, the placement director will send a signal to the Copy Checkpoint that the copying of state transmission wires may begin. When a placement subordinate is traversing northward on the state lookup chute wires it will lock the appropriate doors to shut down the construction wires connecting chutes of different directions. 

\begin{figure}[t]
    \centering
    \includegraphics[width=.9\textwidth]{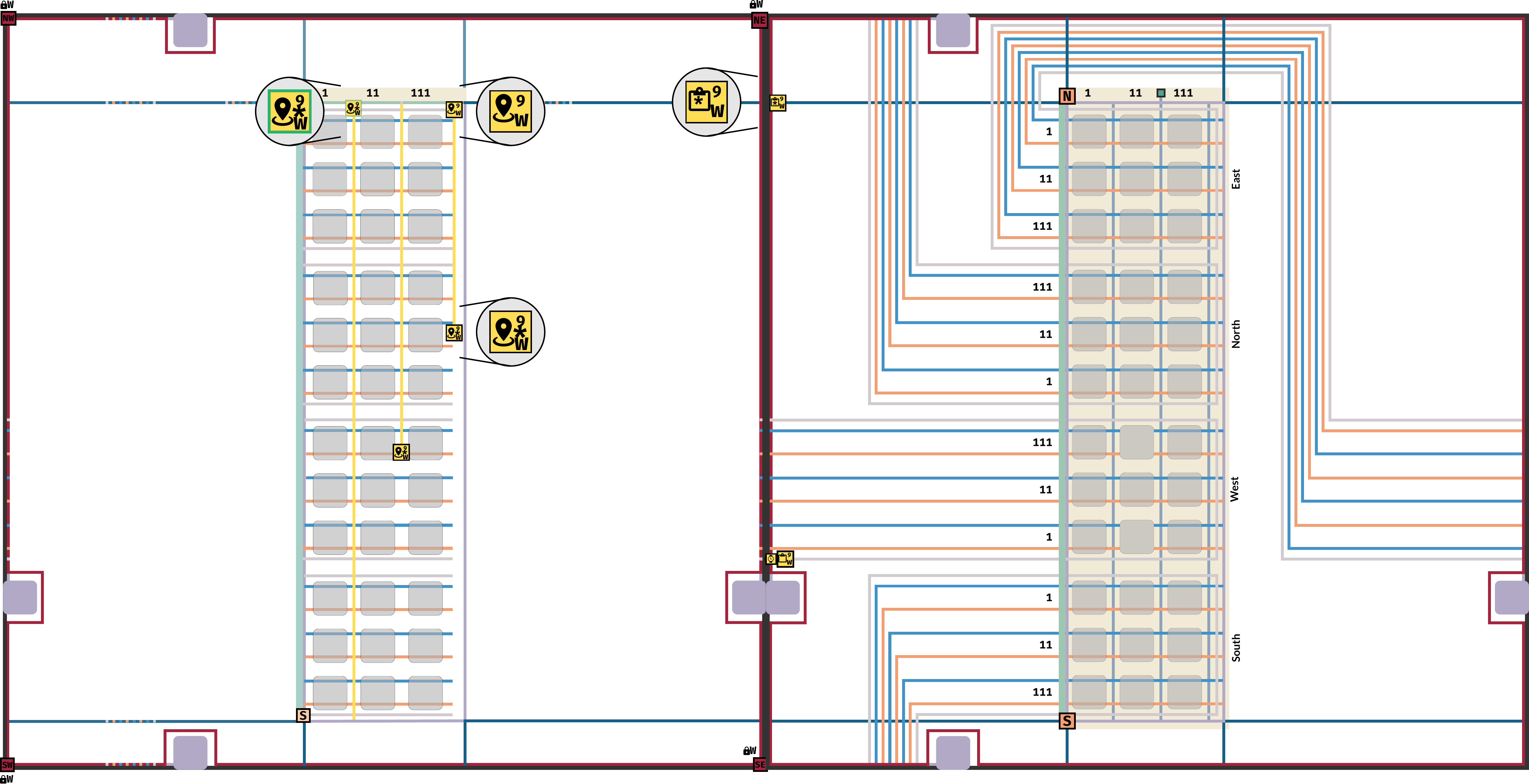}
    \caption{Constructing vertical table wires.}
    \label{fig:constructing_vertical_table_wires}
\end{figure}

\paragraph*{Copying State Transmission Wires.} After the table is complete the state transmission wires are copied, see Figure~\ref{fig:constructing_state_transmission_wires}. After the appropriate wires are set and locked, we begin with the wires entering from the east of the tile.
The Copy Directors for any direction will skip copying any vertical/horizontal construction wire and tile/table edge crossovers. 

\begin{figure}[t]
    \centering
    \includegraphics[width=.9\textwidth]{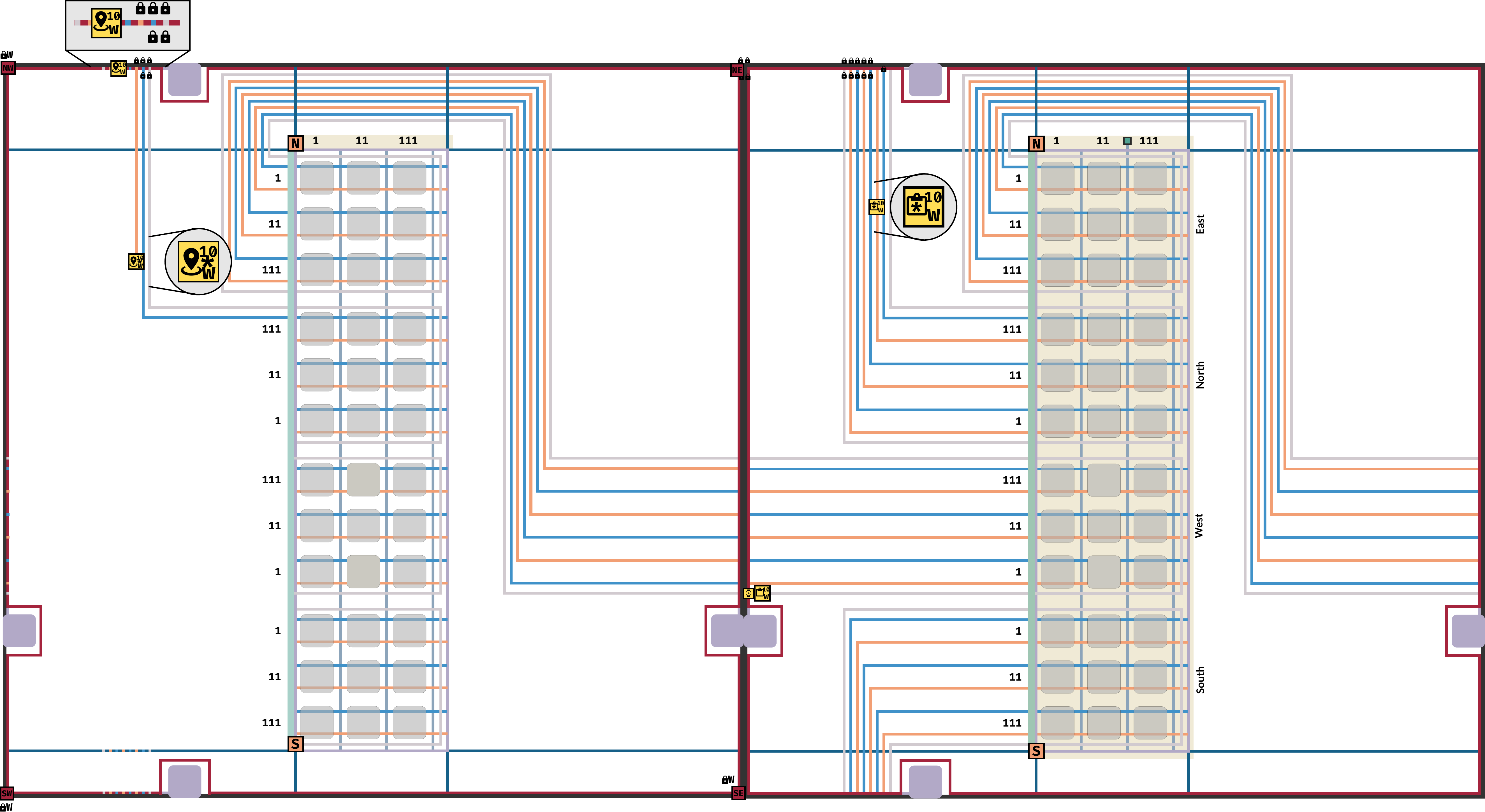}
    \caption{Constructing state transmission wires.}
    \label{fig:constructing_state_transmission_wires}
\end{figure}

\subsection{Activating Tile and Determining State}\label{subsec:activating_tile_and_determining_state}
Once the supertiles construction has finished the process of activating the supertile for use begins. First, the construction wires must be deactivated. 

\paragraph*{Construction Wire Deactivation.} After construction phase 10, phase 11 is started. To begin, phase 11 locking unlocking agents are dispersed throughout both tiles with the copy and placement directors, checking off their respective tiles intersections after each phase 11 locking unlocking agent reports back from the construction wire it was sent down until it returns to where it started. The copy and placement directors will not unlock the outer edges of their respective tiles yet.

\paragraph*{Reporting Construction Completion to Neighboring Construction Tile.} After the tile is confirmed to be complete, construction wires are deactivated, and table is activated, a signal is sent to the neighboring tile that was in charge of construction that it may reactivate most of its non-construction functions and its active state column, see Figure~\ref{fig:locking_construction_wires_reactivation}. 

\begin{figure}[t]
    \centering
    \includegraphics[width=.9\textwidth]{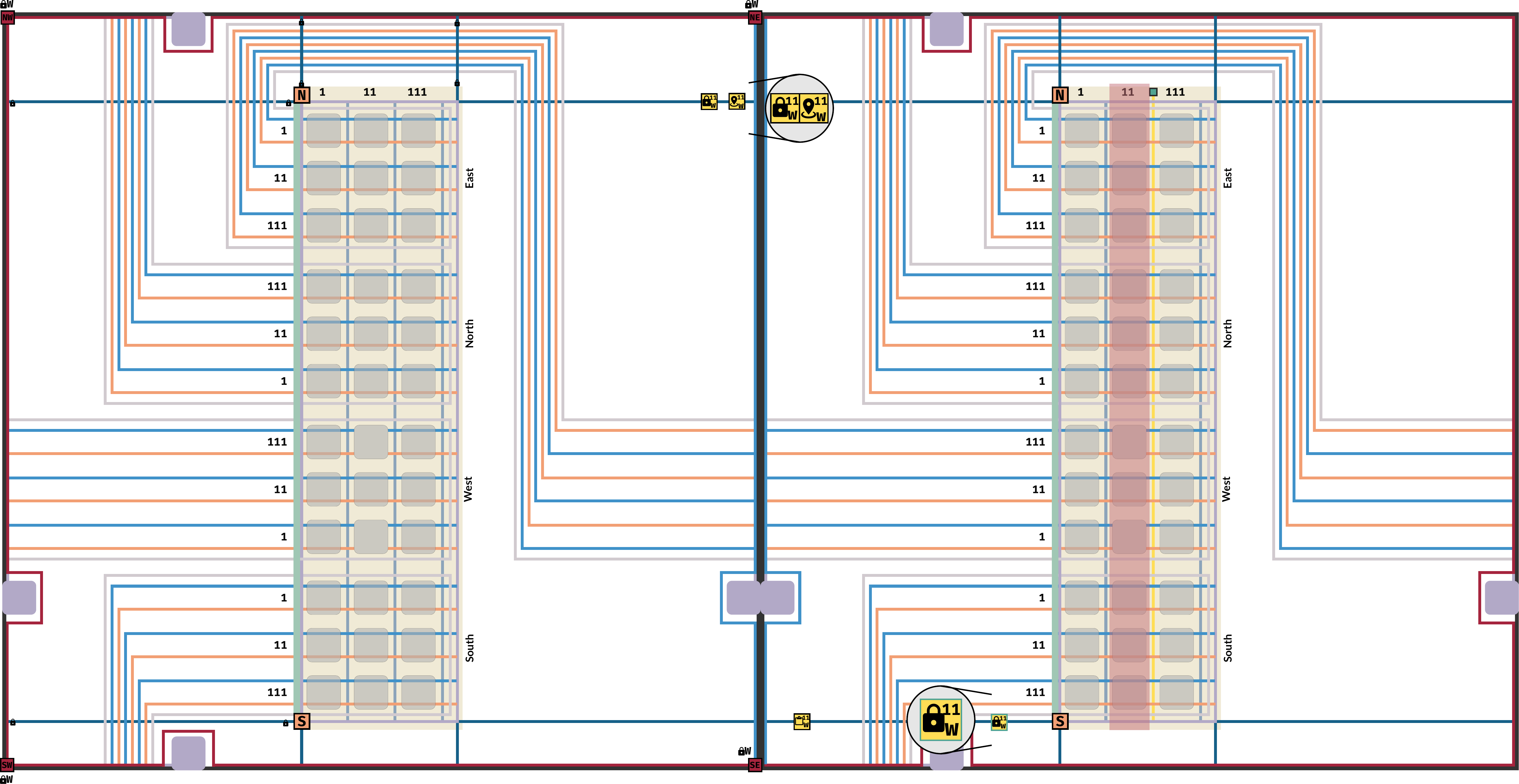}
    \caption{Locking construction wires and reactivating neighboring supertile.}
    \label{fig:locking_construction_wires_reactivation}
\end{figure}

\paragraph*{Activating Table.}
When it is confirmed that the construction wires have been locked, state transmission wires appropriately unlocked, and neighboring wires reactivated, the placement director will move into phase 12 where it activates the table edge and punchdown gadgets. 

\paragraph*{Requesting Tile State From Neighbors.} After the supertiles construction is complete it will send out a Requesting State Agent to its neighbors along its state lookup output wire for each direction. Upon reaching the edge of the supertile it will meet doors that indicate the supertile has no neighbor but it may traverse through them anyway. After the supertiles state has been selected and activated the has ``no neighbor doors" will be transitioned to standard tile edge crossover doors. If a transition request reaches these has no neighbor doors it will be dissolved on contact. Only the newly constructed supertile can transition its has no-neighbor doors and that of the tile next to it.

\paragraph*{Transmitting Neighboring Supertile States.} Once a New State Requesting Agent has entered the lookup chute in the active column in a neighboring supertile it will traverse to each row until it reaches a self-intersection with the active state for that direction. It will report the state of the neighboring tile back to the new tile table. 

\begin{figure}[t]
    \centering
    \includegraphics[width=.8\textwidth]{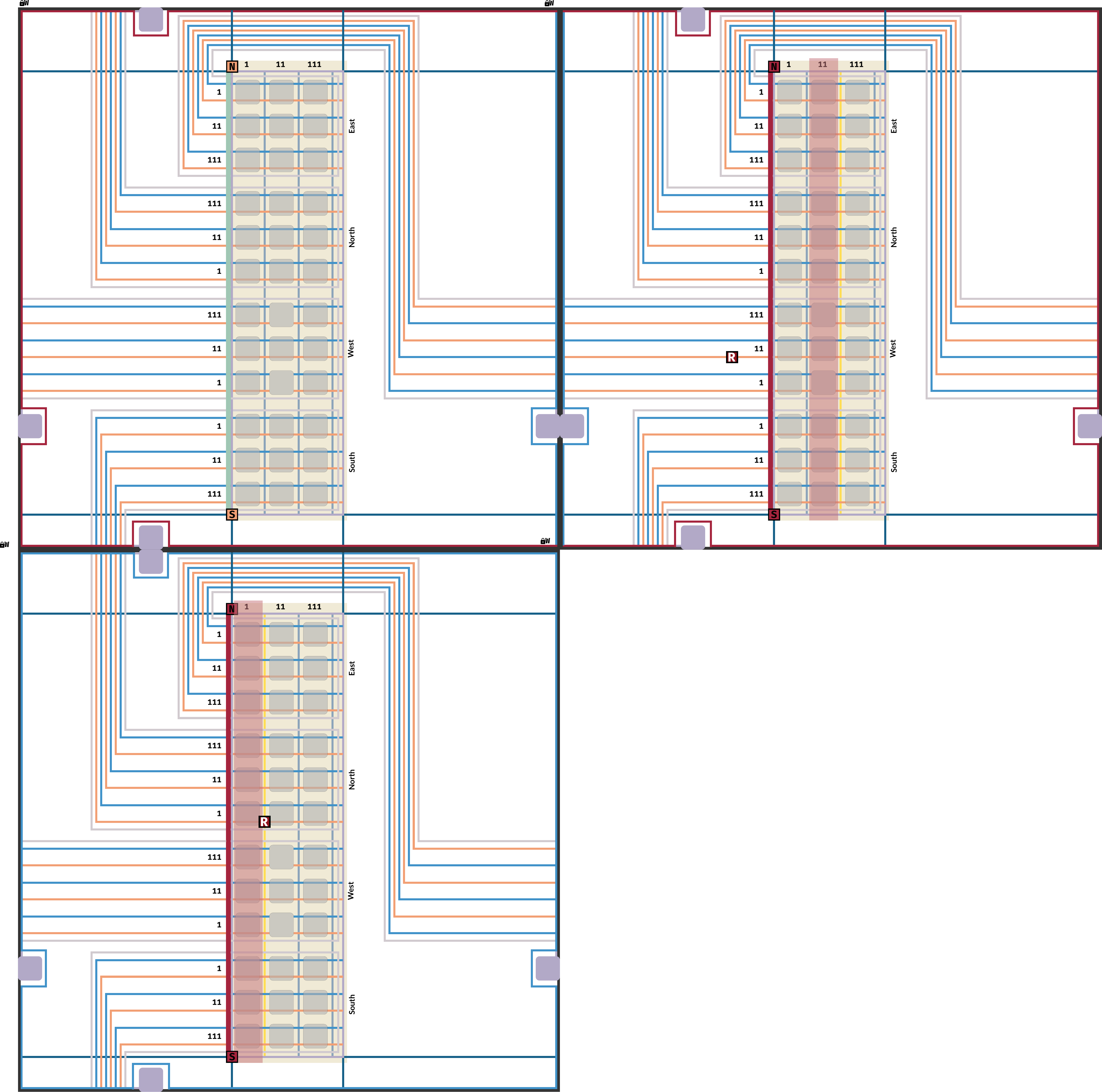}
    \caption{Receiving states from neighboring supertiles.}
    \label{fig:constructing_locking_construction_wires_reactivation}
\end{figure}

\paragraph*{Selecting (and Deselecting) Possible States.} After one of the New State Request - Neighbor Reporting Agents have won the table locking race it will enter the table and traverse each column. At the end of each datacell, in the southern slot of the vertical wire/datacell crossover gadget, there is an affinity (or no affinity) door that the agent may transition nondeterministically with to select the state of the new tile. This agent may traverse the row backward or forward at any time and may even completely exit the table and restart the locking race so that another directions New State Request - Neighbor Reporting Agent may enter the table and potentially select the state, see Figure~\ref{fig:constructing_locking_construction_wires_reactivation}.

\paragraph*{Activating Column and Regular Table.}
Once a state is selected state column activation agents are sent to the north and south of the initial selection door, turning each incoming datacell punchdown door/vertical intersection crossover into an active superstate mode and storing the state in the state storage box at the top of the column, see Figure~\ref{fig:selecting_state_attachment_ov}. 

Only after a state is selected and column confirmed to be activated in both directions does the New State Request - Neighbor Reporting Agent Eraser Door Agent Activator spawn and traverse to the top of the table and over to the left edge where it changes each of the  table control edges inner doors to its normal active state. This causes New State Request - Neighbor Reporting Agents to dissolve upon contact (leaving an omni-directional wire tile behind) and triggers the table outer doors to unlock.

\begin{figure}[t]
    \centering
    \includegraphics[width=.6\textwidth]{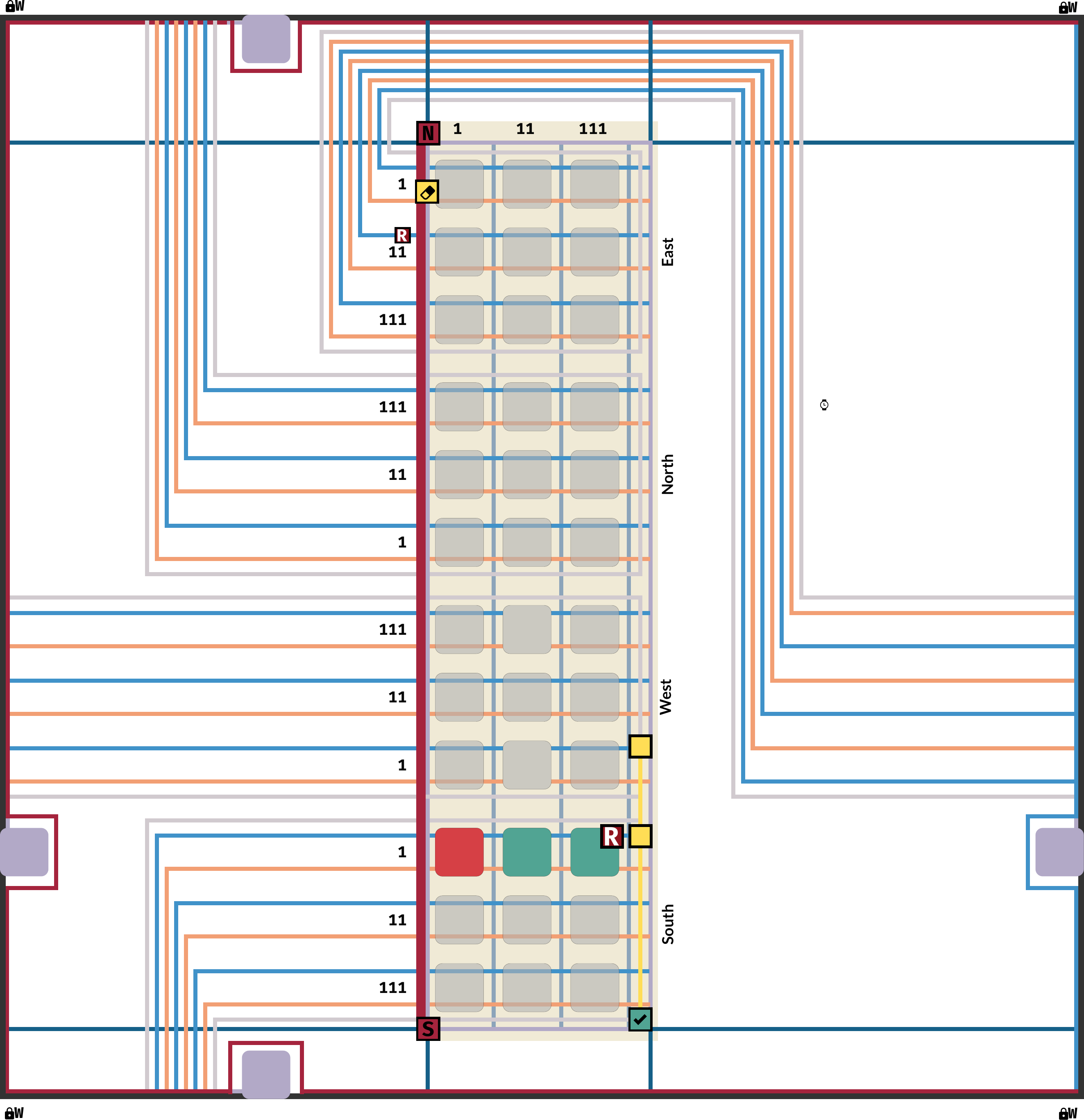}
    \caption{Selecting the state of the supertile.}
    \label{fig:selecting_state_attachment_ov}
\end{figure}

\paragraph*{Unlocking Neighboring Supertile and Testing for Neighbors.} When the New State Request - Neighbor Reporting Agent Eraser Door Agent Activator has reached the table's south marker tile it will send an Unlock Neighbor Outline and Test for Neighbors Agent. This agent will also unclaim corners of the new supertile. The agent will traverse around the edges of both supertiles unlocking them or testing for neighbors if the doors say they have none. 

\begin{figure}[t]
    \centering
    \includegraphics[width=.7\textwidth]{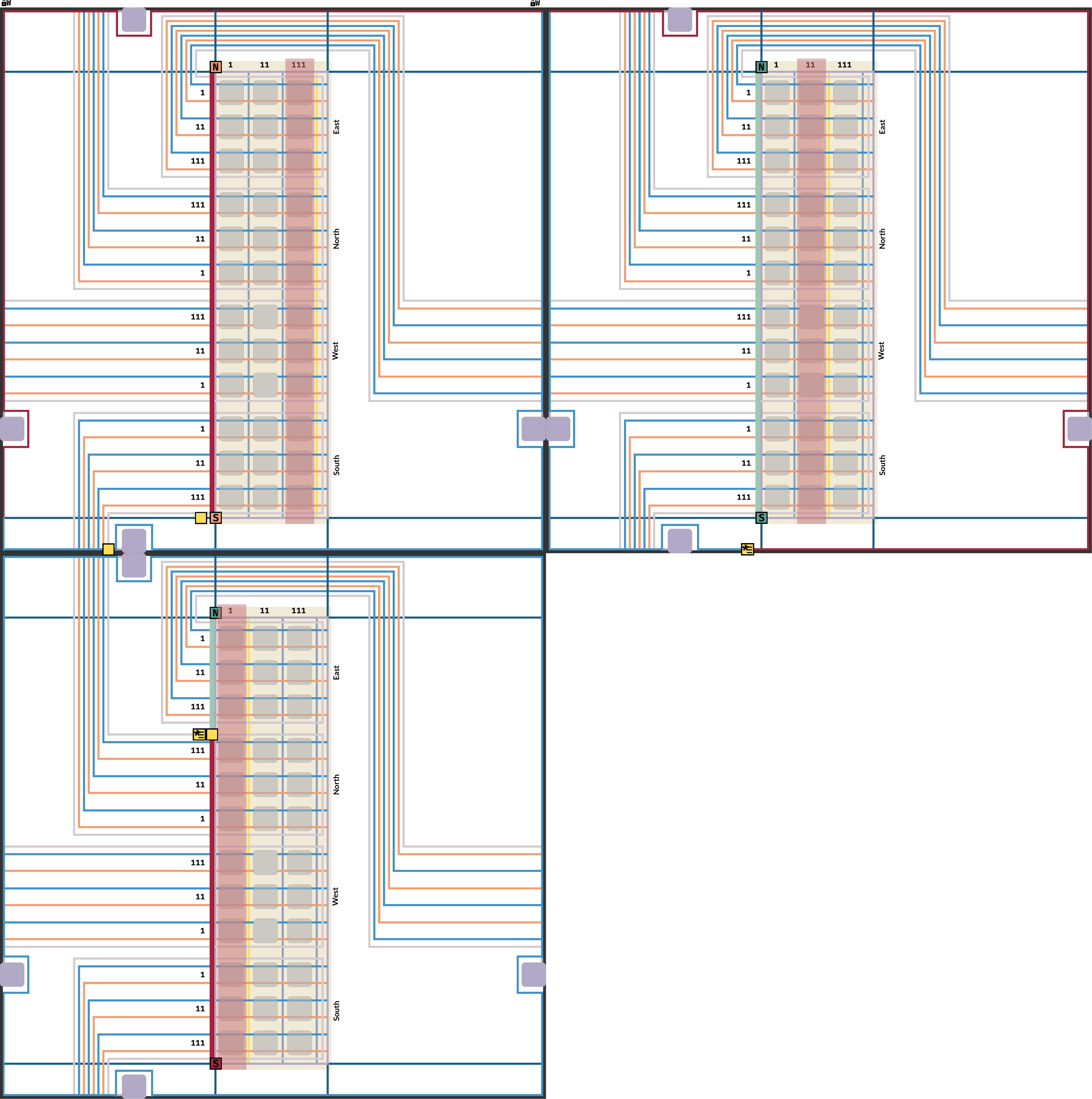}
    \caption{Testing for neighbors and unlocking supertiles.}
    \label{fig:testing_neighbors_attachment_ov}
\end{figure}

\paragraph*{Sending Out New Supertile State.} When the Unlock Neighbor Outline and Test for Neighbors Agent returns through the southern construction wire door (which it may traverse due to its special state) it will change to a State Transmission Trigger Agent that will traverse the active state column and send out State Transmission Agents at each self-intersection. The first State Transmission Agent to reach the inside edge of the table will unlock it, see Figure~\ref{fig:sending_new_state_attachment_ov}.

\begin{figure}[t]
    \centering
    \includegraphics[width=.8\textwidth]{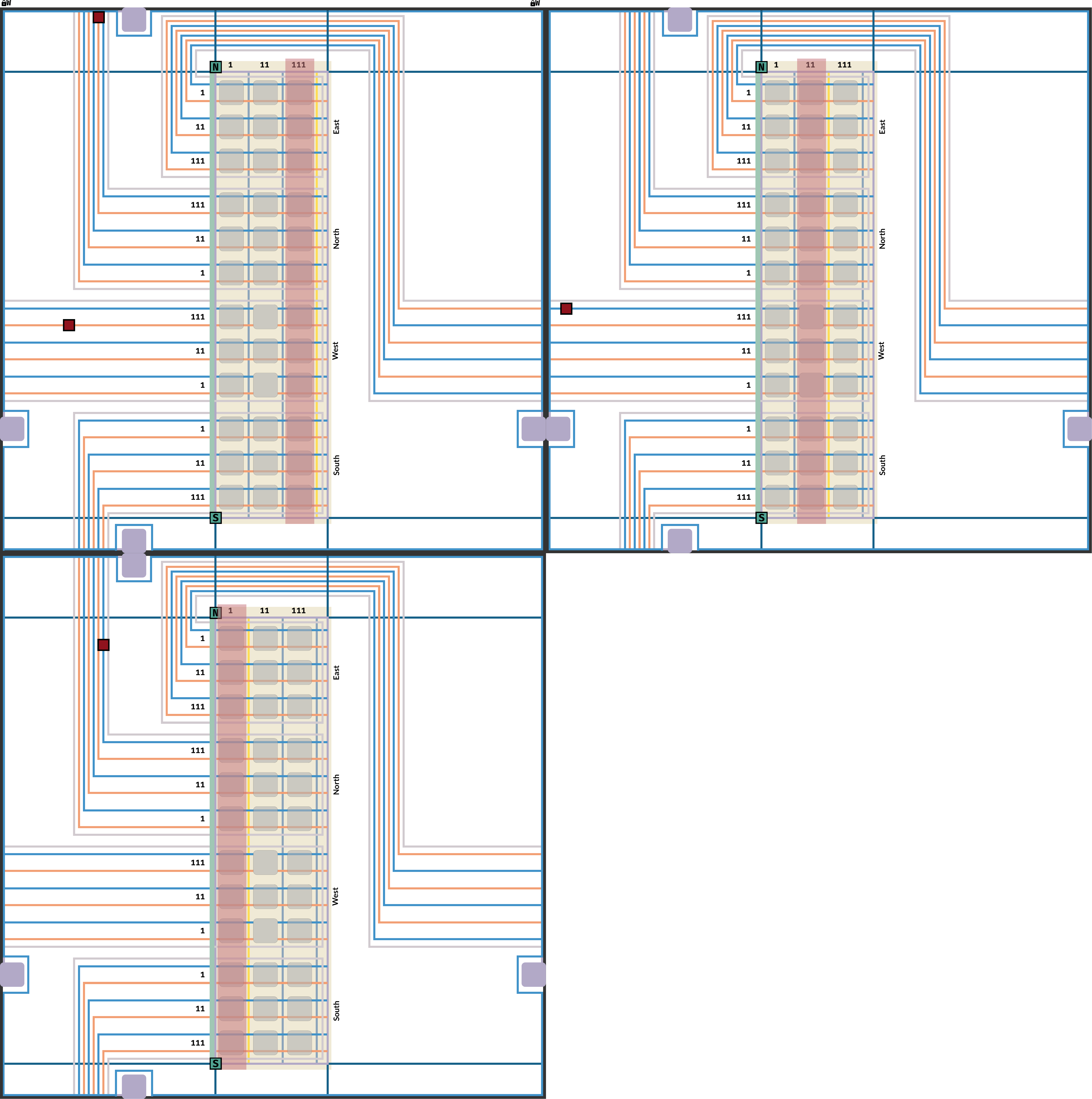}
    \caption{Sending out new state to neighbors.}
    \label{fig:sending_new_state_attachment_ov}
\end{figure}
\section{Transitioning Tiles}\label{sec:transitioning_tiles}

The transition process starts once a state notification agent from a neighboring supertile reaches another supertiles table. The table is locked by standard procedure, and if the lock is successful, the agent is admitted into the table.

\subsection{Finding Intersection}\label{subsec:finding_intersection}
The wire this agent is entering on implicitly encodes the state of the neighboring supertile.
To determine whether there is a transition, it needs to find the active state column.
The active state column has special door states, so when the agent reaches the active state door, it is no longer able to traverse further into the table.
Reaching this door ensures that the agent will have the chance to transition with the transition storage door below, which is only unlocked (if a transition exists) in the active state column.
If there are no transitions then the transition storage area door will be in a no transitions available state, see Figure~\ref{fig:transition_p1}.

If there are transitions and the state notification agent does not choose to back out of the transition process and table altogether, then the process to prepare both supertiles for transitioning begins. 

\begin{figure}[t]
    \centering
    \includegraphics[width=.7\textwidth]{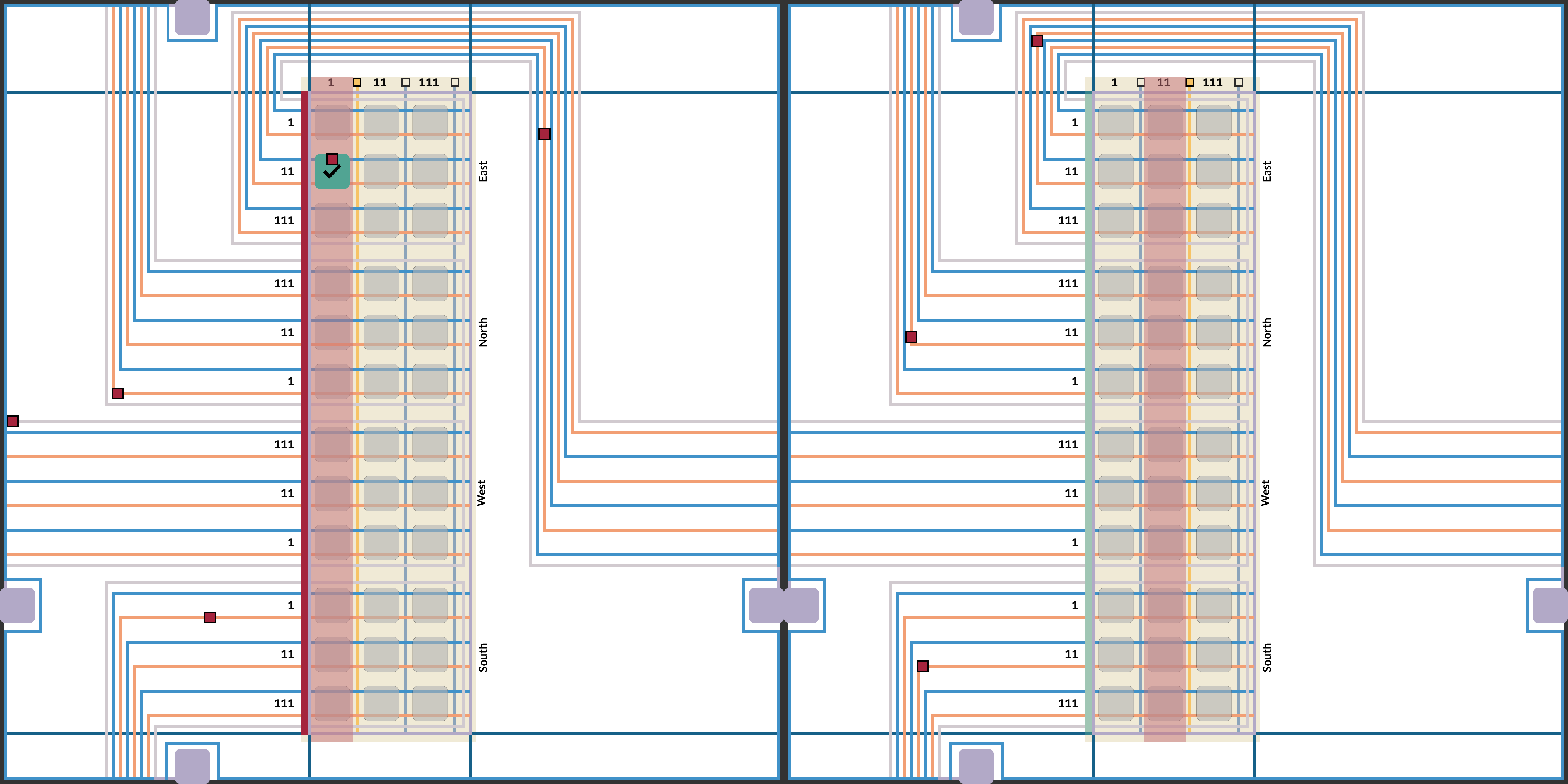}
    \caption{An agent discovers the existence of a transition with its neighbor.}
    \label{fig:transition_p1}
\end{figure}

\subsection{Transmitting Intention to Transition}\label{subsec:transmitting_intention_to_transition}
First, it must be confirmed that the neighboring supertile that sent the state notification agent is still in the state it was when it sent the agent. The agent confirming the state locks the tile into its current state. 

This process starts with the State Notification Agent transitioning into a Transition Preparation Agent—Confirm Neighbor when swapping with the unlocked transition storage area door in the active column. During the swap, the transition storage area door has ``awaiting confirmation'' appended to its state.

The Transition Preparation Agent—Confirm Neighbor traverses down the storage area wire and out the transition storage exit door. As it cannot swap with the south door of the datacell's bottom crossover door, thus we can ensure it exits out of the output wire in the same state and in the same direction it came from.

The Transition Preparation Agent—Confirm Neighbor will leave the table locked as it exits. If we find the neighboring table locked, the transition is rejected, and an agent is sent to reset the transition storage area door and unlock the table afterward. 

\paragraph*{Rejection of Neighbor Tile State}
Once the Transition Preparation Agent has locked the neighboring supertiles table and traversed the columns to the active state, it will enter through the transition storage door as within the other datacell. As the Transition Preparation Agent - Confirm Neighbor became Transition Preparation Agent - Confirm Self-Intersection when swap-transitioning with the current datacells transition storage door, it will then traverse to the bottom of the transition storage area where the self-intersection marker tile sits. 
If the self-intersection marker tile is instead a not self-intersection tile, the Transition Preparation Agent will reject the transition, traversing back to the transition storage entrance door and deselecting it, then walking out of the datacell and unlocking the table upon exit. Once it reaches the initial supertile, it will enter the table (still unlocked at that particular door) and remove the ``awaiting confirmation'' designation from the active states transition storage area entrance door, finally locking the table on its way out. 

\paragraph*{Confirmation Neighbor Tile State}
However, if the lookup is confirmed to be at the ``self-intersection,'' we know the neighbor's state has not changed. Thus we instead send a Transition Preparation Agent - Neighbor Confirmed back to the originating supertile. In addition, we send a wire setting agent to open the way from our datacell down to the state lookup chute exit wire for the respective direction and an agent to follow that to trigger the transition selection gadget, see Figure~\ref{fig:transition_p2}.

\begin{figure}[t]
    \centering
    \includegraphics[width=.7\textwidth]{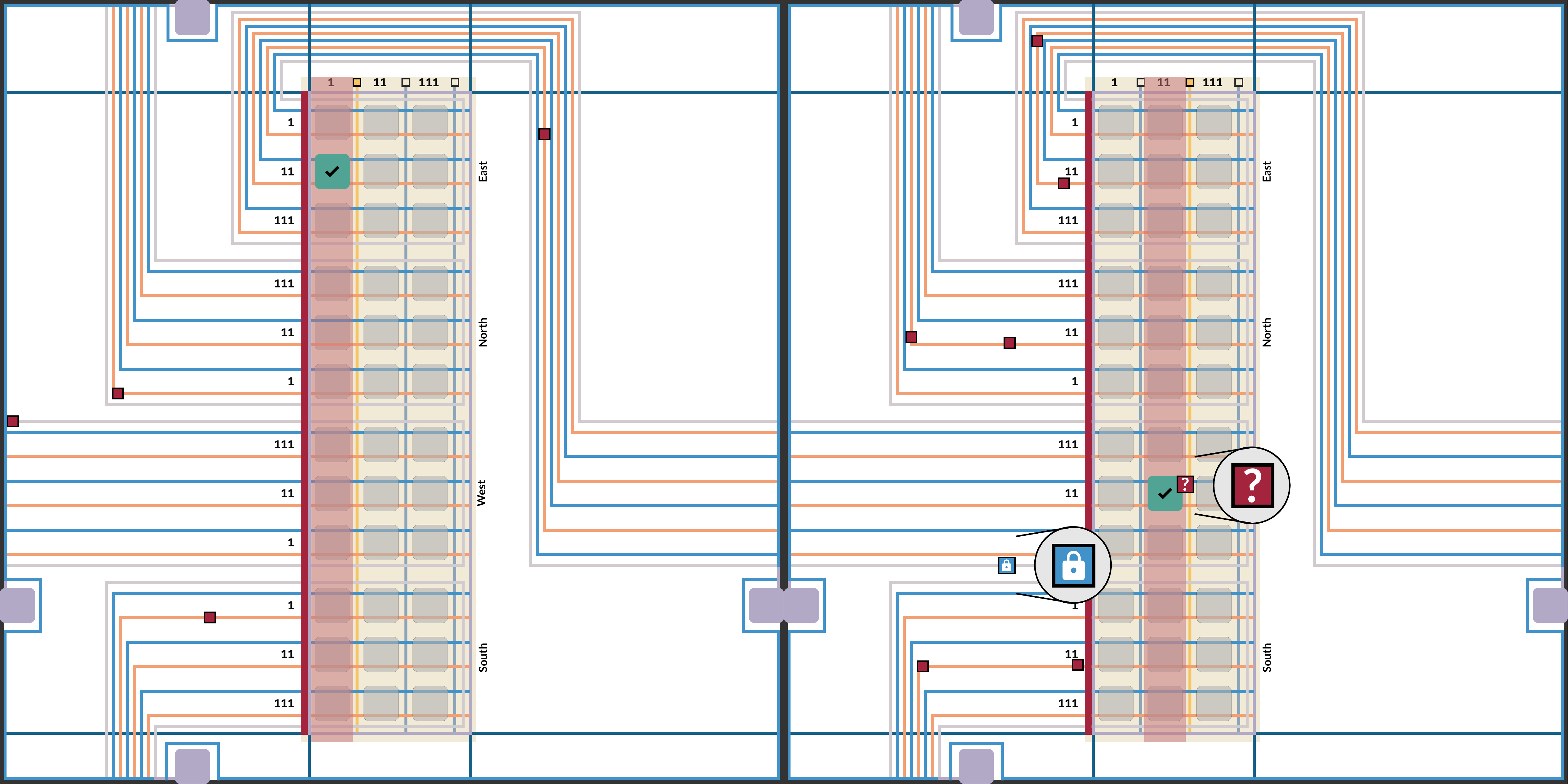}
    \caption{The agent checks whether the neighboring supertile is still in the same state and locks the neighbor's table.}
    \label{fig:transition_p2}
\end{figure}

\paragraph*{Copying Transition Rules}
Each supertile begins copying its transitions largely using the same method as during attachment. First, the door is activated, which activates the unary (or filler) tile behind it, which then transition-swaps with the door into a spawn copy state, which will spawn into the wire outside the door. Once the number has successfully copied onto the wire it will flip through its row until it hits the end of data row tile where it will return to its normal inactive state. Once the door has cycled through all of the data string tiles, it will shift to a finish buffer cycling state and stop opening, simply transitioning with buffer tiles so that they flip behind the data string as before until the beginning of the data string reaches the door again. Once this occurs, the current door will trigger the door below it to begin its copying process. This is continued until all data strings have been copied into the wire and are heading to the transition selection gadget, as in Figure~\ref{fig:transition_p3}. 

\begin{figure}[t]
    \centering
    \includegraphics[width=.7\textwidth]{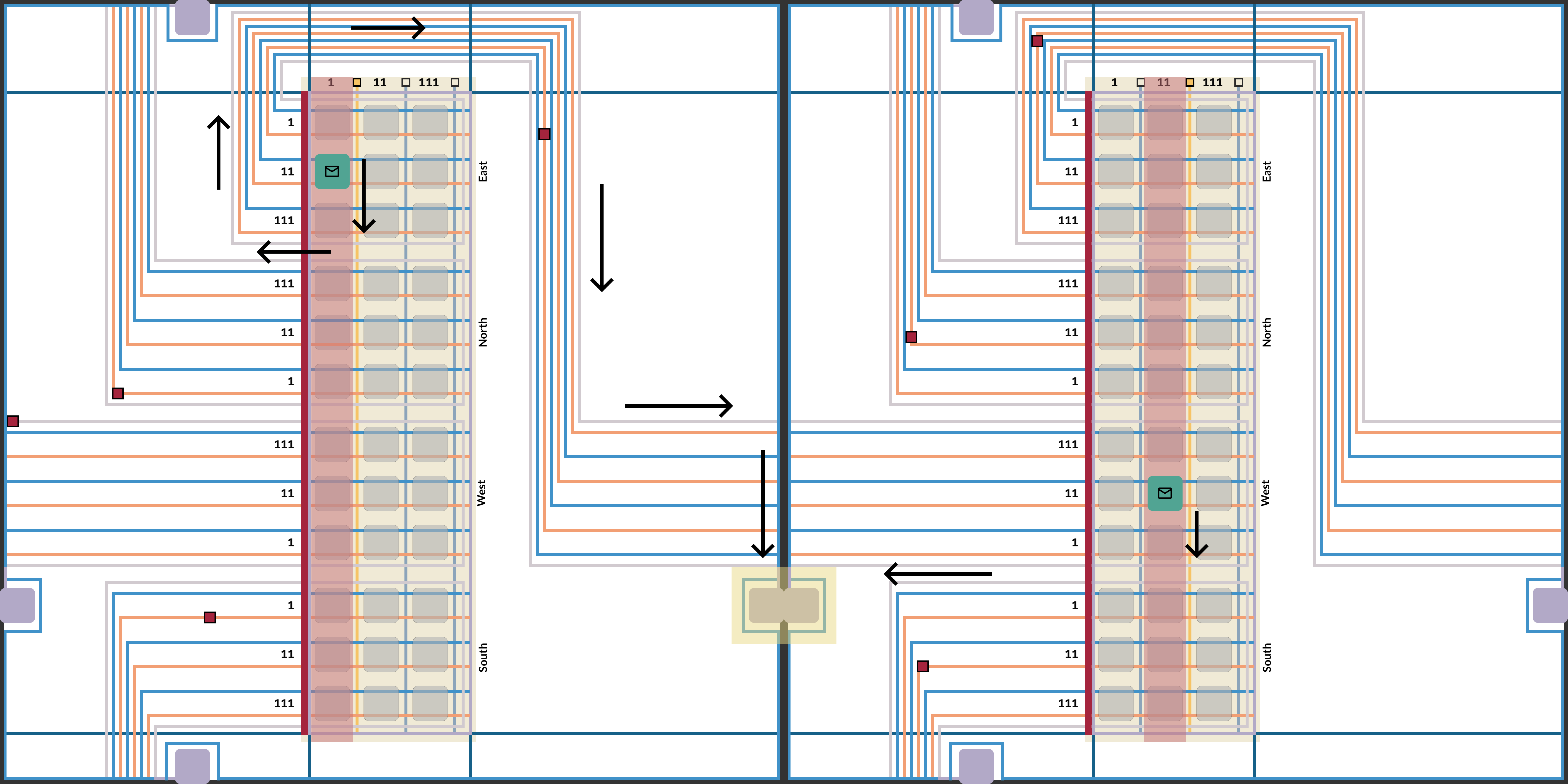}
    \caption{The transitions are sent to the transition selection gadget.}
    \label{fig:transition_p3}
\end{figure}

\paragraph*{Filling the Transition Selection Gadget.}
For supertiles transitioning with a neighbor to their east or north, a wire setting/locking agent must be sent by the Transition Director upon reaching the edge of the tile, as the transition selection gadget is on the opposite side of their lookup chute exit. 

Once the Transition Director reaches the transition selection gadget, the entrance door and the first row will be activated to begin overwriting the blank filler tiles within the gadget rows. The fill door will allow one data string to go through before it needs to be unlocked by the Transition Director again. After each row is filled, the wire tile next to the door below it is unfrozen. As we are filling top to bottom, this ensures data strings do not fill improperly to doors below. 

\begin{figure}
    \centering
    \includegraphics[width=.45\textwidth]{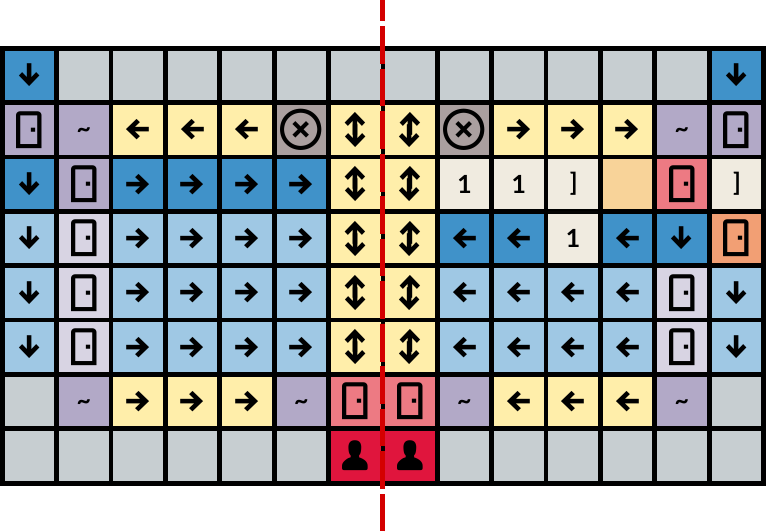}\hfil
    \includegraphics[width=.45\textwidth]{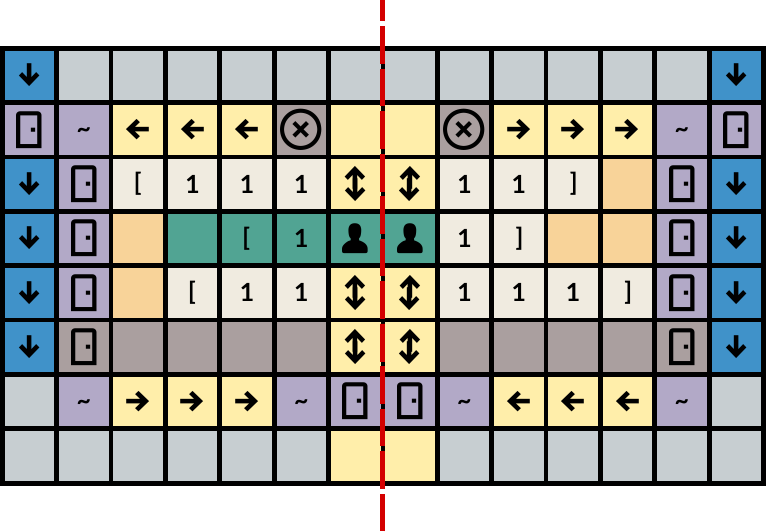}
    \caption{Left: The transition selection gadget is filled row by row. Right: Transition Selection gadget selects a transition to take.}\label{fig:transition_selection}
    
\end{figure}

\paragraph*{Selecting the Transition.}
Once all of the data strings are properly in their rows the transition director will activate the selection agent for its half of the transition selection. The agent will traverse on the transition border wire parallel to (and touching) the other supertiles transition border wire. The agents randomly walk up and down the transition selection wire.

Once both have been activated, it is possible for them to align next to one another at one of the transition rules and have the chance to select, see Figure~\ref{fig:transition_selection}. When a transition is selected, the two agents must double transition, instantly choosing the new state of their respective supertiles simultaneously. When determining the supertile of a state, the highest priority is if a row/column has been selected in the transition selection gadget, and then if there is none, the active state column of the supertile. 

The agents may also choose an abort transition row/column double transition option at any time. 

Once a transition is selected, the data string is copied out of the appropriate row/column it was being held in and led by a transition director to the top state lookup chute wire. The transition director has special override authority at the table border to open this top wire and allow the data string to traverse through the table.

\subsection{Transitioning States}\label{subsec:transitioning_states}
Each supertile is in charge of completing its respective half of the transition rule. 

\paragraph*{Punch Down Mechanism and Data String Traversal.}
The data string will traverse the input wire however many columns are specified by the data string. The punchdown tile will transition with a 1 tile turning the 1 tile into an east wire tile, thus decrementing the data string, and the punchdown tile will tell the door handle to unlock the door. Once the door swaps and transitions with the end-of-data string tile, it will reset fully with its door handle, return to its locked state, and tell the punchdown gadget to reactivate through its handle.  

\paragraph*{Deselecting the Previous State.}
The data string ends in an end-of-data string tile that, when it is punched down, updates that column as the new state. Before this can occur, the old column must be deactivated. Thus 2 agents are sent, one to the east and one to the west to search for and deselect the old active column. 

When one reaches the old active column, it will spawn agents to traverse north and south, deactivating active state doors. When they reach the bottom of the table, they report back to the Deselection Agent, who reports back to the state update pending door, waiting for confirmation of deselection (and no column from the other direction).

\paragraph*{Activating New State Column.} The activation of the column looks much the same as the deselection of the state column; see Figure~\ref{fig:selecting_deselecting_cols}. The state update pending agent (waiting next to the affinity door) sends agents north and south to activate each column door as the active state and fill the state storage tile at the top with the new active state.

If the column selected is the current active state column, then only the wipe transition selection gadget agent is sent. 

\begin{figure}[t]
    \centering
    \includegraphics[width=.7\textwidth]{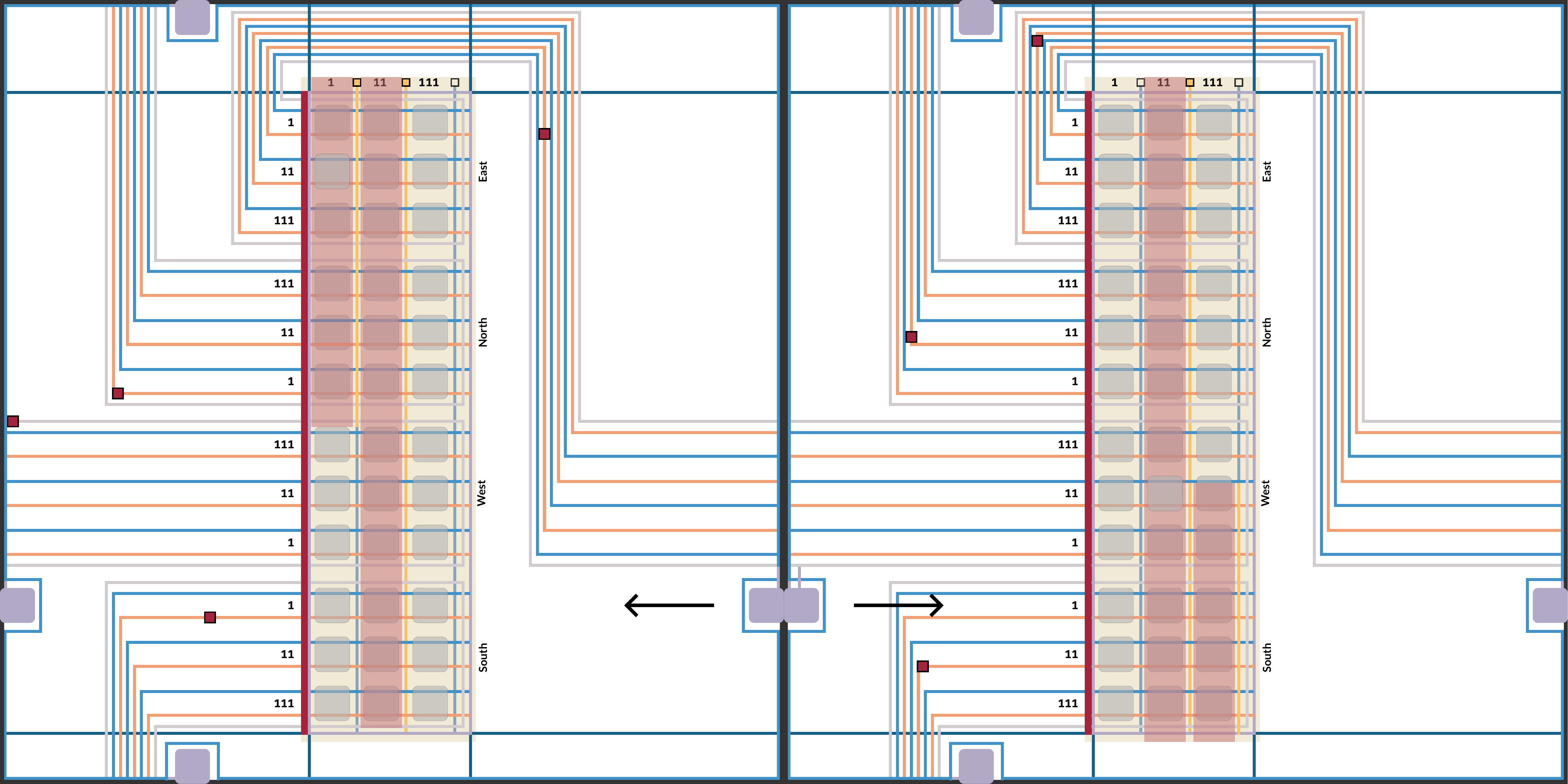}
    \caption{\label{fig:selecting_deselecting_cols}The supertiles independently transition by first deselecting the old column and then selecting the new one.}
\end{figure}

\paragraph*{Wiping Transition Selection Gadget.}
Once this has been reported to be complete, the door will spawn a wipe transition gadget agent, which will traverse through the new active state column to the bottom of the state lookup chute wire for its direction and travel to the transition selection gadget. 
Upon arrival, the agent will enter the transition selection gadget and trigger the rows/columns to cycle wipe each of the tiles in its slot. This works much like copying for the filler tile cycling; except instead, it is erasing values. After the wiping agent has checked off all slots to be complete, a recapture of the agent will be triggered. When recapturing the agent has been confirmed it can then exit the transition selection gadget, resetting the gadget door and any wires set/doors locked by the transition process at the edge of the tile. 

\paragraph*{Sending Out New State.}
Once the wiping agent has returned to the active state column, the column is triggered to send out the new active state at each self-intersection. The table is unlocked by the first state transmission agent that reaches it.

\section{Metrics}\label{sec:metrics}

As the vast majority of states are dedicated to the operation and copying of crossover doors we cover this aspect first. Next is agents and finally other gadgets states. 

The number of states was calculated as follows: 

Nearly all crossover gadgets have four doors, each of these four doors may have been set to one of two directions (standard and reversed) each one of these doors has the following states: active, waiting, open, reset, pushback, and locked. Thus, there are 7 states for standard operation over 8 possible doors making 56 per crossover gadget type as standard. All doors use the same handle set adding only 12 states to the total calculation. In total 42 crossover gadgets are necessary coming to 2352 states for standard operation. 

Additionally, during the copying process the majority of crossover gadgets have a copy yourself state for each door (though the doors enter a state reset when complete unlike other copy processes) and a copy agent for 9 additional states per type. For placement we transition with the first blank tile it swaps with to a wire tile, then we overwrite the leftover second placement tile, and lastly each of the four doors are put into their respective blank tiles, adding 5 additional states (other agents will set them to standard or reverse as necessary). As such the copying process of each crossover door requires 15 states and thus 630 for the entire construction.

There are 16 non-crossover doors in the system that each have the above standard operations (and half have independent handles) adding to 120 states.
Crossover doors have special states to indicate that there is no neighboring door, adding 96 states to the incoming and outgoing state transmission crossovers for each direction.
Corner crossover doors can be claimed by a neighboring supertile, adding 64 states.

\subsection{Agents}
Agents locking and unlocking the table requires 30 agent states and 5 additional door states. 
Initiating transitioning, copying the data strings, and ensuring they reach the appropriate locations requires an additional 31 agent states, and 24 door states. 
Selecting the new state within the gadget, copying out the data string for the new state, and wiping or aborting the transition requires 34 states.

\subsection{Copying States}
\subparagraph{General Copy States and Agents.} 
Each direction of placement tiles requires an inactive, active, complete, and special crossover double state, adding to 16 states. 

\subparagraph{Copying Crossovers.} Each of the crossover doors (regular and reversed) must have a state indicating they should copy themselves, a state indicating to spawn the same agent, an agent which must traverse 2 steps to the center of the crossover gadget then check off that the doors (not reversed) have been attached at each side. As this requires 15 states and there are 42 crossovers, this adds 630 states.  

\subparagraph{Locking Agents.} Each step that requires a locking agent needs a spawn/waiting state for the copy director and the locking agent itself needs an active, lock door 1, lock door 2, exit crossover, and locking complete state. Nearly every locking agent also needs the copy director and an unlocking agent with the same states for a total of 14 states. There are a total of 10 phases that require locking agents, but the state transmission wire construction needs these for each side. In addition, there are 30 other miscellaneous states that are used across various phases. This brings the total of these to 212 states associated with copy locking. 

\subparagraph{Placement Director.} The placement director has an awaiting direction tile, an overwrite completed direction tile, a waiting state tile, an overwrite state tile, a waiting crossover agent completion, overwrite crossover agent, lock door 1, lock door 2, exit crossover, and complete states, making 10 states over 4 construction directions for 40 states. This standard version applies to 4 phases, but cycling is done 11 times due to subphases for a total of 462 states.

Additionally, the copy director and/or placement director will spawn placement directors or subordinate placement directors and wait for their completion 22 times.
Doing this for 4 directions for 44 states per directions makes for 176 states added.

\subparagraph{Aborting Process.} The abort construction process (not including reactivation) takes 10 states to overwrite, wipe, and inform the copy checkpoint/director for each direction, adding a total of 40 states.

\subparagraph{Traversing Opposite.} In 13 phases and subphases the copy director must traverse to the opposite side of the tile. Adding 52 states.

\subparagraph{Datacell Outlines.} The subordinate prime placement director must be spawned, place a wire to the south, door to the west and a border tile to the east before moving on, when it runs into a no state tile to its south it will instead place an exit door and mark itself complete. As this doesn't depend on the direction, it only adds 6 states. 

\subparagraph{Filling Datacell.} Copying each transition rule requires the copy director to activate each tile for copying, flipping through them without sending direction tiles at this phase; they will mark themselves complete in addition to the copy director, the placement director and tiles do in this in reverse on the opposite side. With the necessary checkpoints included this adds 22 states. 

\subparagraph{Vertical Table Wires.} In addition to the check off states (counted above) the south traversal agents need to skip crossovers and the final one needs to lock on the way up adding 6 states. 

\subparagraph{State Transmission Wires.} As each copy and placement director needs to check off first and last for each direction and crossovers need to be skipped there are 32 states added. 

\subparagraph{Reactivating Neighboring Supertile.} The agent must delete the checkpoint, traverse to the top of the table to spawn a generic sub-agent that doesn't depend on construction direction, unlock the table, check for where the active state column is, spawn an activation agent, and let the newly finished tile know this process is complete adding 12 new states. 

\subparagraph{Activating New Table.} In the new supertile the table doors must be moved into special door states added to the east and west of each table edge state transmission wire crossover. This adds 14 new states. 

\subparagraph{Requesting, Receiving, and Selecting States.} Requesting and receiving states requires and agent to send them from each direction in the active state column, the state requesting agents themselves, and special state transmission agents. Selecting the state requires abort and select, if the state is not a full state then an additional special agent is required, adding 9 states.

\subparagraph{Activation, Unlocking, and Transmitting.} The activation of the new column, doing a special unlock of the table the self and neighboring tiles outlines and transmission of the new state adds 11 new states.

\subsection{Final Count}
There are an additional 40 miscellaneous states used in the construction bringing the total number of states to 4600, including 2600 non-copy states for our final ACA state count.

\section{Correctness of Construction}\label{sec:seeded_results}
Here we give proofs of correctness. We first (re)state our main lemma. 

\IUone*

We prove this via the following lemmas which each satisfy a condition of simulation. We start with a helper Lemma.

\begin{lemma}\label{lem:onestep}
    For any assemblies $A \in \prodasm{\Gamma}$ and $A_U \in \prodasm{\Gamma_U}$ such that $A = R^*(A_U)$, any assembly $B_U$ such that $A_U \to_1 B_U$ satisfies either $R^*(A_U) = R^*(B_U)$ or $R^*(A_U) \to_1^\Gamma R^*(B_U)$.
\end{lemma}
\begin{proof}
    An attachment can never change a mapping because if a supertile is incomplete it maps to the empty state. Once the datacell has been built it sends a signal to it's neighbors. Its neighbors will respond by sending an agent which walks into the table. If it reaches an intersection in the table where there is an affinity rule it immediately changes the mapping to the new state simulating an attachment. The next available transitions mark the remaining tiles in the active state column. 

    Until a superstate transition is selected none of the changes that can be made in the supertile change the mapping since they do not change the active state column. 
\end{proof}

\paragraph*{Equivalent Production.}

\begin{lemma}\label{lem:1a}
    For any assembly $A_U \in \prodasm{\Gamma_U}$, the assembly $R^*(A_U) \in \prodasm{\Gamma}$.
\end{lemma}
\begin{proof}
    Any producible supertile either (1) maps to a empty state, (2) has only an active column which signifies the state in $\Sigma$ it represents, or (3) has an active column and a selected transition in which case it maps to the state after the transition. 

    We will use induction along with Lemma \ref{lem:onestep} to prove that all assemblies are producible. For our base case we consider the seed in both systems.
    We replace each tile in the seed $s$ by supertiles representing that tile to get seed assembly $s_U$. 
    Then by Lemma \ref{lem:onestep} every move we make on assemblies $A_U$ in $\prodasm{\Gamma_U}$ creates an assembly $B_U$ which represents an assembly $B$ that is reachable by $A$ in $\Gamma$.
\end{proof}

\begin{lemma}\label{lem:1b}
    For all $A_U\in \prodasm{\Gamma_U}$, $A_U$ maps cleanly to $R^*(A_U)$ with 1-fuzz.
\end{lemma}
\begin{proof}
    The seed $s_U$ we create maps cleanly to the original seed $s$ as we only place supertiles in locations where tiles take place. 
    
    Each ghost tile is built from a neighbor boundary first. Once the boundaries are built, the ghost tile copies the contents of the supertile. It is not until the supertile is complete and has selected a state that it begins to attempt to build neighboring ghost tiles. Therefore each ghost tile is adjacent to at least one properly mapped tile.
\end{proof}

\paragraph*{Equivalent Dynamics.}

\begin{lemma}\label{lem:2a}
    For all $A,B \in \prodasm{\Gamma}$ such that $A\to^{\Gamma} B$, it holds that for all $A_U$ such that $R^*(A_U) = A$, we have $A_U\to^{\Gamma_U} B_U$ for some $B_U\in \prodasm{\Gamma_U}$ with $R^*(B_U) = B$.
\end{lemma}
\begin{proof}
    Consider any pair of assemblies $A, B \in \prodasm{\Gamma}$ such that $A \to_1^{\Gamma} B$. Pick an arbitrary $A_U$ such that $R^*(A_U) = A$. If this transition was achieved via an attachment the agent selects the active tile column by traversing the datacells at an intersection. It may also chose to not stop at the intersection and continue on or go backwards to select another tile. This allows $A_U$ to achieve any attachment performed by $A$.
    
    For transitions, all available rules will be loaded up into the transition selection gadget. If the two agents meet they may select the transition and instantly change the mapping of both tiles, transitioning from $A_U$ to $B_U$ based on our mapping. However, the non-deterministic process may not select a transition at all and will allow the agents to keep walking to select any transition, or abort.    
\end{proof}

\begin{lemma}\label{lem:2b}
    If $A_U \to^{\Gamma_U} B_U$ for some $A_U,B_U \in \prodasm{\Gamma_U}$, then $R^*(A_U) \to^{\Gamma} R^*(B_U)$ or $R^*(A_U) = R^*(B_U)$.
\end{lemma}
\begin{proof}
    If a attachment or transition does not change its mapping then we satisfy $R^*(A_U) = R^*(B_U)$. For a ghost tile to transition to a valid mapped tile, it must have an active state column.
    This active state column is only build and actually activated if there was a neighboring supertile that had the appropriate affinity.

    For a transition the agents must both match and find the same transition in order to change the mapping of the tile. Only proper legal transition may be placed in the table so all of these must be valid transitions from $R*(A')$ to $R^*(B)$.
\end{proof}

\paragraph*{Transitivity of Simulation.}
Here we show the definition of simulation is transitive, and hence we may chain many simulations together. It is possible that chaining 1-fuzz simulations results in an increase in fuzz by a constant factor. However, in our case we preserve 1-fuzz which we prove in Theorem~\ref{the:chaining_1-fuzz}.

\begin{lemma}
    The definition of simulation is transitive. 
    If each simulation is 1-fuzz and has scale factor larger than $1$ then the resulting simulation has at most 3-fuzz.
\end{lemma}
\begin{proof}
    First consider a chain of $k$ simulating systems where $\Gamma_i$ simulates $\Gamma_{i+1}$ for $0\leq i < k$. 

    Condition 1 from \emph{equivalent productions}, and both the \emph{follows} and \emph{models} conditions of equivalent dynamics are all preserved by the fact we may compose the representation functions.

    The second condition of \emph{equivalent productions}, namely the $c$-fuzz bound, requires more care as the fuzz of a simulation is not immediately preserved.
    However, we can ensure that the fuzz will be bounded by at most $3$.
    At each simulation step, the size of a supertile is getting smaller by a fraction $\alpha \leq \frac{1}{2}$.
    Since each simulation has at most one ghost tile next to its valid parts of the assembly, every simulation can add at most one ghost tile neighboring the previous one, which is a fraction $\alpha$ smaller than the previous.
    Since $\alpha \leq \frac{1}{2}$, this geometric series in the plane can reach a distance of at most $3$ from the original supertile.  
\end{proof}

Even though chaining 1-fuzz simulations can lead to a simulation using 3-fuzz, chaining our specific construction would never lead to more than 1-fuzz.

\begin{theorem}\label{the:chaining_1-fuzz}
    Chaining our simulations results in a 1-fuzz simulation.
\end{theorem}
\begin{proof}
    The individual tiles of a supertile $S$ would never go outside the boundingbox of $S$.
    Take an individual tile $t$ on the edge of $S$.
    If we would chain simulations, $t$ would be simulated using a supertile $S'$.
    Supertile~$S'$ would only build a new ghosttile outside of $S$ if $t$ would want to build outside of $S$.
    Since this never happens, chaining our simulation only results in 1-fuzz.
\end{proof}

\paragraph*{Universality Results.}

\genIU*
\begin{proof}
    Lemma \ref{lem:iut1} states that temp-1 is IU for itself.

    Chaining these two simulations will still result in a 1-fuzz simulation as ghost tiles are only built where a new tile may attach. Our construction in Theorem \ref{thm:temp1sim} has 1-fuzz and the ghost tiles that attach do not have any other affinities with neighboring tiles. Thus the supertile simulating them in Lemma \ref{lem:iut1} will not place any additional ghost tiles. For the same reason any assembly which has no attachments will not build any ghost tiles and thus have no fuzz. 
\end{proof}
\section{IU TA Simulates 2D Asynchronous CA \texorpdfstring{$N = 2$}{|N| = 2}}\label{sec:iu_ta_simulates_2d_async_ca}  
Previously, a partial proof of 1D asynchronous cellular automata (ACA) being intrinsically universal was shown in \cite{worsch2013towards}. Here, we apply techniques used throughout this paper to show two subsets of asynchronous cellular automata are intrinsically universal. We start by defining pairwise and block-pairwise ACA.

\subparagraph*{Asynchronous Cellular Automata}
An Asynchronous Cellular Automata (ACA) system is a 4-tuple $\Gamma = \{\Sigma, N,  \Delta, C\}$, where $\Sigma$ is a set of states,  $N \in \mathbb{N}$ is the neighborhood of $\Gamma$, $\Delta$ is a mapping $\Delta: \Sigma^{N} \to \Sigma$ and $C$ is a configuration that is a mapping $C: \mathbb{N}^2 \to \Sigma$. We refer to each mapping in $\Delta$ as a transition rule.

\subparagraph*{Pairwise ACA.}
A pairwise ACA is an ACA with one extra consideration. More formally, it is defined as a 4-tuple $\Gamma = \{ \Sigma, S, \Delta, C\}$, where $S$ consists of all possible subsets of size $2$ between a cell $c$ and adjacent cells in each cardinal direction, and $\Delta$ is a mapping $\Delta: \Sigma^s \to \Sigma$, where $s \in S$. If $\Delta$ is a mapping $\Delta: \Sigma^s \to \Sigma^s$, we consider $\Gamma$ a block-pairwise ACA.

Note that these automaton are a subset of radius-1 ACAs since we can transform each transition rule in $\Delta$ into larger mappings that ignore the neighbors not included in the rule. However, this increases the number of rules by a factor of $|\Sigma|^3$ since we need to account for all possible configurations of the neighborhood.  

\begin{lemma} \label{thm:block-pairwise-aca}
    Block-pairwise ACA is strongly intrinsically universal. 
\end{lemma}
\begin{proof}
    Let $R$ be a representation function for a given block-pairwise ACA system $\Gamma = (\Sigma, S, \Delta, C)$. Map each cell $c \in C$ to $c'$ such that $R(c')=C(c)$, including a mapping for empty cells, in the same manner as a seeded TA system. As with the techniques described in earlier sections, each new cell $c'$ stores state information about each of its $4$ neighboring cells and the transitional information from $\Gamma$, ensuring cells are only able to change from some $R(c')=\sigma$ to $R(c')=\sigma'$ if there exists a valid transition rule from $\Delta$ to allow it. 
\end{proof}

\begin{figure}[t]
    \centering
    \includegraphics[width=1.\textwidth]{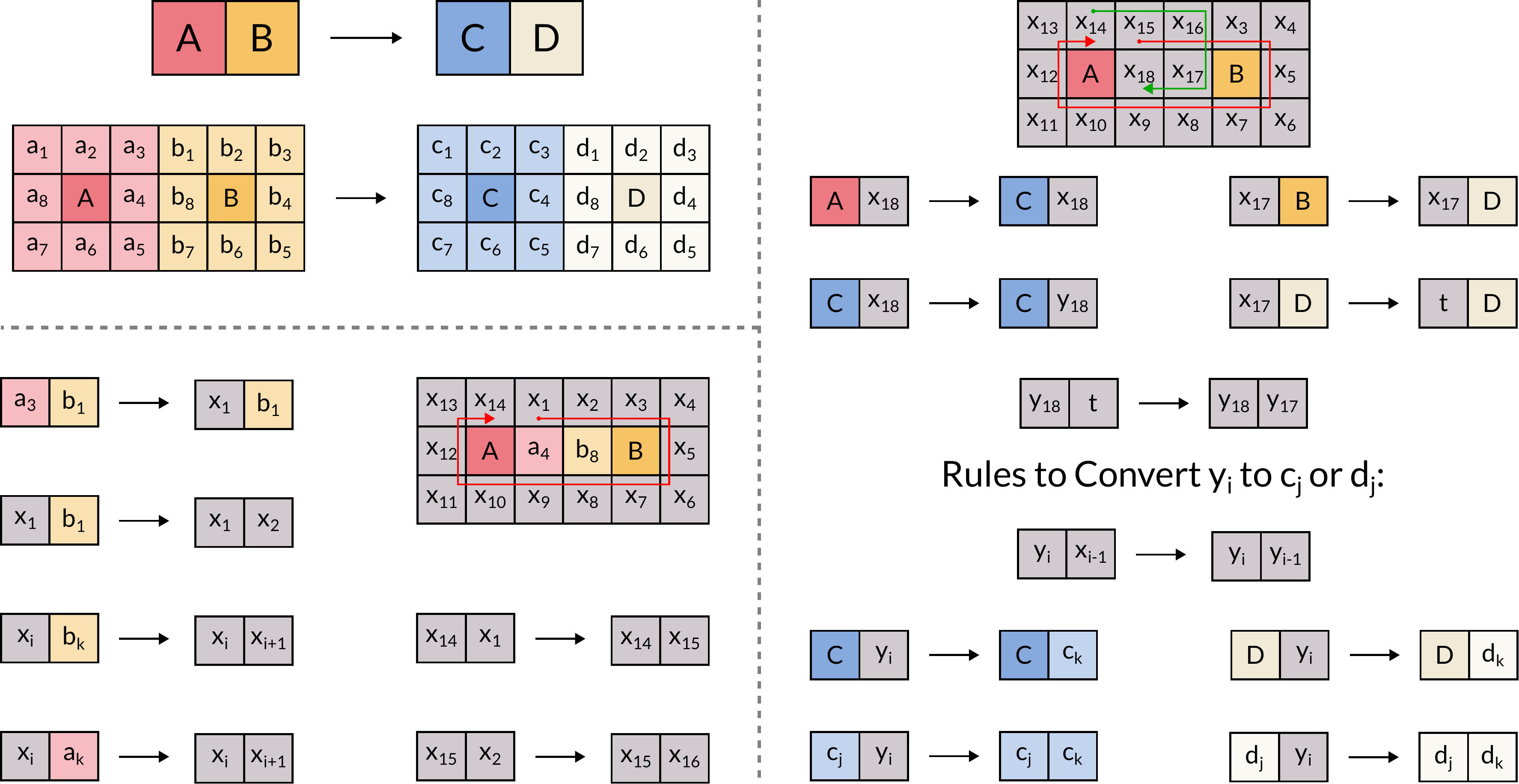}
    \caption{ 
        Simulating a dual-transition rule with only single-sided transitions. We scale the simulation by 3 and any transition occurs by ``locking'' the two tiles, transitioning the two tiles, and then unlocking them. The rules shown are a general idea, but it requires an additional 53 states. 
    }
    \label{fig:dualsim}
\end{figure}

\begin{lemma} \label{thm:dualsim}
    Dual transitions in Tile Automata can be simulated by single-sided transitions at constant scale with $O(\Delta_2)$ additional states where $\Delta_2$ is the number of dual-transition rules in the system. 
\end{lemma}
\begin{proof}
This was previously addressed in \cite{cantu2020signal} in relation to the signal tile model. Figure~\ref{fig:dualsim} gives a general overview of how to do this simulation at scale-3 with an additional 53 states. Basically, all tiles around the two $3\times 3$ macroblocks change before changing the states. This locks them into the transitions, and is reversible until state $b_8$ changes to $x_{17}$.
The $x$'s then change to $y$'s after the $A\rightarrow C$ and $B\rightarrow D$ change. The $y$'s then turn to $c$'s and $d$'s. 
\end{proof}

\begin{restatable}{theorem}{acaIU}
\label{thm:pairwise-aca}
    The Asynchronous Cellular Automata model with a cardinal-direction neighborhood of size-2 and radius-1 (pairwise ACA) is strongly intrinsically universal.
\end{restatable}
\begin{proof}
    Pairwise ACA is a special case of block-pairwise ACA. However, any cell transitions based on its neighbors. Thus, all transitions are single-sided in terms of Tile Automata. Thus, we modify the block-pairwise IU result from Theorem \ref{thm:block-pairwise-aca} to only use single-sided transitions through scaling as shown in Lemma \ref{fig:dualsim}.
    This means that there is a constant-size set of states that is intrinsically universal.
\end{proof}

\section{Conclusion}\label{sec:conclusion}
We showed that no passive or freezing tile assembly model can be non-committal intrinsically universal.
However, we showed that the seeded Tile Automata model, with its unbounded state changes, is non-committal intrinsically universal.
This is done by showing TA is intrinsically universal even under temperature $1$ using $1$-fuzz.
Moreover, a Tile Automata system using temperature $\tau > 1$ can be simulated using a system that uses temperature at most $1$.
Chaining these two simulations shows that there exists a tile set that can simulate any Tile Automata system.
This intrinsic universality result has direct implications for certain Cellular Automata.
Moreover, the result directly implies that the original aTAM model can be simulated using Tile Automata.

There is significant room to optimize and minimize the tile set.
For example, the number of tile states necessary to copy a supertile is large, whereas big sections of the supertile will always be the same, independent of what system we are simulating.
Furthermore, the temperature simulation, and consequently the universal simulation, uses a lot of states.
It might be possible to combine both simulations into one, by storing the affinity strength in the datacell.
A ghost tile would then need to check all neighboring supertiles for their affinity strengths and add them up, before deciding which state it will become.

Another obvious open problem is that of dimensions other than two.
It is still unknown whether the Tile Automata model is intrinsically universal if you extend the model to one, or to three or higher dimensions.
Even though our simulation could technically simulate a one dimensional tile set, the supertiles would still use two dimensions themselves.

Finally, as in the aTAM model, our construction heavily relies on the fact that (locally) only a single tile can attach at a time.
Because of this, our current construction only shows the seeded Tile Automata model to be intrinsically universal.
Hence, the question arises whether or not the non-seeded Tile Automata model is intrinsically universal.

\bibliographystyle{plainurl}
\bibliography{main}
%


\end{document}